\documentclass[twocolumn,pre,floatfix,superscriptaddress,nofootinbib]{revtex4-2}
\usepackage[T1]{fontenc}
\usepackage{mathrsfs}
\usepackage{amsmath}
\usepackage{amssymb}
\usepackage{graphicx}
\usepackage{braket}
\usepackage[svgnames]{xcolor}
\usepackage{comment}
\usepackage{natbib}
\usepackage{pifont}
\usepackage{algorithmic}
\usepackage{algorithm} 
\usepackage{multirow}
\usepackage{tabularx}
\usepackage{bm}
\usepackage{amsthm}
\usepackage{float}
\usepackage[utf8]{inputenc}
\usepackage{thm-restate}
\usepackage{booktabs}

\usepackage{listings}

\makeatletter
\newcommand\blfootnote[1]{%
  \begingroup
    \renewcommand\thefootnote{}%
    \renewcommand\@makefntext[1]{\noindent ##1}%   ← no marker box, no indent
    \footnotetext{#1}%
  \endgroup
}
\makeatother

\usepackage[normalem]{ulem}
\usepackage{hyperref}
\usepackage{dsfont}

\hypersetup{
  colorlinks=true,
  linkcolor=blue,
  citecolor=ForestGreen,
  urlcolor=DarkOrchid
}

\newcommand{\HP}{\mathbb{H}}

\newcommand{\etal}{\textit{et~al.}}

\begin{document}

\title{Quantum-Informed Portfolio Selection:\\ An End-to-End Pipeline Validated on Trapped-Ion Hardware with Real Market Data}

\author{Romina Yalovetzky\textsuperscript{$\mathsection$}}
\thanks{Corresponding author: romina.yalovetzky@jpmchase.com.}
\affiliation{Global Technology Applied Research, JPMorgan Chase, New York, NY 10017 USA}

\author{Martin J. A. Schuetz\textsuperscript{$\mathsection$}}
\affiliation{Amazon Advanced Solutions Lab, Seattle, Washington 98170, USA}
\affiliation{AWS Center for Quantum Computing, Pasadena, CA 91125, USA}

\author{Zichang He}
\affiliation{Global Technology Applied Research, JPMorgan Chase, New York, NY 10017 USA}

\author{Jiayu Shen}
\affiliation{Global Technology Applied Research, JPMorgan Chase, New York, NY 10017 USA}

\author{Yue~Sun}
\affiliation{Global Technology Applied Research, JPMorgan Chase, New York, NY 10017 USA}

\author{Rudy~Raymond}
\affiliation{Global Technology Applied Research, JPMorgan Chase, New York, NY 10017 USA}

\author{Shauna Sahay}
\affiliation{JPMorgan Chase, New York, NY 10017 USA}

\author{Kishore Perla}
\affiliation{JPMorgan Chase, New York, NY 10017 USA}

\author{Ruben S. Andrist}
\affiliation{Amazon Advanced Solutions Lab, Seattle, Washington 98170, USA}

\author{Grant Salton}
\thanks{Work performed while at Amazon Advanced Solutions Lab.}
\noaffiliation

\author{Helmut G. Katzgraber}
\thanks{Work performed while at Amazon Advanced Solutions Lab.}
\affiliation{55 North Management ApS, K{\o}benhavn, Denmark}

\author{Roger Bongiovanni}
\thanks{Work performed while at JPMorgan Chase, New York, NY 10017 USA.}
\noaffiliation

\author{Niraj Kumar}
\affiliation{Global Technology Applied Research, JPMorgan Chase, New York, NY 10017 USA}

\author{Rob Otter}
\affiliation{Global Technology Applied Research, JPMorgan Chase, New York, NY 10017 USA}

\date{\today}

\begin{abstract}
Portfolio diversification---a cornerstone of modern investment management---can be formulated as a Maximum Independent Set (MIS) problem on asset correlation graphs. Solving this problem at scale is computationally challenging, motivating the exploration of quantum algorithms for practical financial optimization. We propose an end-to-end pipeline leveraging qReduMIS, a recursive hybrid quantum-classical algorithm. Rather than using quantum optimization to directly produce a final solution, qReduMIS leverages independent set measurements from the Quantum Approximate Optimization Algorithm (QAOA) to identify frozen nodes---vertices likely to belong to optimal solutions---thereby guiding and unblocking subsequent (provably optimal) classical reductions on the remaining graph. We benchmark qReduMIS on real financial data from four major market indices with up to 225 assets, executing experiments on Quantinuum's 98-qubit trapped-ion Helios system, with QAOA circuits acting on kernels of up to 78 qubits and 1016 two-qubit gates. While standalone QAOA fails to find the optimal solution for two of the largest indices (S\&P~100 and Nikkei~225), qReduMIS achieves success probabilities of $0.40$ and $0.95$, respectively, with average approximation ratios $\geq 0.96$ across all four indices. We perform a systematic benchmark on the Quantinuum H2-1 noisy emulator over 73 asset correlation graphs of varying size showing that, for $p=2$ QAOA layers, the optimal time-to-solution scaling exponent of qReduMIS is $3.2$ times smaller than that of standalone QAOA. 
\end{abstract}

\maketitle
\blfootnote{ $^{\mathsection}$ These authors contributed equally.}

\section{Introduction}
Diversified portfolio construction is fundamental to balancing risk and 
return in finance~\cite{markowitz1990foundations, elton2009modern, 
fabozzi2008portfolio}. A central challenge is selecting a subset of 
assets from a vast and ever-changing universe such that the resulting 
portfolio is robust to market fluctuations and avoids excessive exposure 
to correlated risks~\cite{boginski:05, merton1972analytic, 
bouchaud2003theory}. A powerful abstraction for this task is the 
\emph{Maximum Independent Set} (MIS) problem, a paradigmatic NP-hard 
combinatorial optimization problem~\cite{feo:94, gemsa:16, hale:80, 
dong:22, boginski:05, kalra:08, hidaka2023correlation}. In the resulting 
formulation---which we term the \emph{portfolio selection problem} 
(PSP)---assets are represented as nodes in an asset correlation graph, with edges 
connecting pairs whose correlation exceeds a given threshold; the MIS 
then identifies the largest subset of mutually uncorrelated 
assets~\cite{kalra:08, hidaka2023correlation, boginski:05}. Crucially, these asset correlation graphs are dense, exhibit non-local 
connectivity, and have heterogeneous degree distributions---structural 
properties that differ fundamentally from the sparse, regular, or 
geometrically constrained graphs typically studied in combinatorial 
optimization.

Traditional methods for solving MIS, including exact algorithms and heuristics, often struggle to scale to the non-local graphs arising in financial applications \cite{cazals2025identifying, xiao2017exact}, where edges encode pairwise correlations that do not respect any
underlying local structure. Recent advances in quantum computing have opened new avenues for tackling such hard combinatorial optimization problems. 
Quantum optimization algorithms such as Quantum Annealing 
(QA)~\cite{kadowaki:98, farhi:00, farhi:01, das:08, hauke:20} and Quantum Approximate Optimization Algorithm 
(QAOA)~\cite{farhi:14, zhou:20} offer potential speedups, but face 
significant challenges: both exhibit suppression of success probability 
with increasing instance hardness---of which problem size is one 
driver. While there has been great efforts in protecting quantum optimization algorithms with error mitigation and correction~\cite{he2025performance, perlin2026fault, omanakuttan2025threshold, jin2025iceberg}, this is a challenge that persists even without hardware 
noise~\cite{quiroz2025quantifying, mandl2024amplitude, ebadi:22} and 
can drop to zero for hard 
problems~\cite{ebadi:22, schuetz2025qredumis}. Previous experimental demonstrations of QAOA have spanned a range of
hardware platforms and tackled a variety of combinatorial problems---MaxCut
on superconducting processors~\cite{otterbach:17, harrigan2021quantum},
long-range Ising and weighted MaxCut problems on trapped-ion
hardware~\cite{pagano:20, sciorilli:25, decross:23, montanez:25}, and
MIS on unit-disk graphs using Rydberg atom arrays~\cite{ebadi:22, byun:22}, as well as more recent gate-based MIS demonstrations
on superconducting processors~\cite{dasgupta:26}. Despite this breadth, the \emph{structure} of most of the instances considered has been limited: the graph families considered are either hardware-native (e.g.,
unit-disk graphs on neutral-atom geometries) or low-density and regular (e.g., random
$3$-regular MaxCut). None of these classes captures the
dense, non-local, heterogeneous-degree structure characteristic of
correlation-based graphs arising in real-world applications such as
portfolio selection.

Hybrid quantum-classical algorithms have emerged as a promising alternative
~\cite{bravyi:20, finvzgar2024quantum, acharya:24, he2026regularized}, combining classical
logic with quantum subroutines that leverage the 
sampling capabilities of quantum hardware. For MIS, Brady and Hadfield~\cite{brady2023iterative} introduced the Iterative Quantum Algorithms paradigm, generalizing recursive QAOA (RQAOA) to accommodate hard problem constraints and thereby enabling its application to constrained problems such as MIS. Building on this recursive paradigm, Fin\v{z}gar~\etal{}~\cite{finvzgar2024quantum} proposed the Quantum-Informed Recursive Optimization (QIRO) framework, and Wybo~\etal{}~\cite{wybo2026scalable} introduced a quantum-enhanced greedy algorithm for regular graphs, confining each quantum call to light-cone 
subgraphs on a 20-qubit superconducting device. By exploiting locality, it scales to arbitrarily large instances (up to 5k nodes in their experiments). Exploiting locality works well for sparse graphs but does not readily extend to dense, non-local instances.

The recently proposed qReduMIS algorithm~\cite{schuetz2025qredumis} integrates exact classical reduction techniques with a Quantum Processing Unit (QPU) acting as a co-processor that identifies so-called ``frozen'' nodes---nodes with high (or low) marginal probability of belonging to a maximum independent set---thereby unblocking subsequent classical 
reductions. This framework has demonstrated the potential to overcome the exponential suppression of success probability observed in QA and Simulated Annealing (SA) for MIS on random unit-disk graphs, though its applicability to structured, 
real-world graphs remained an open question for two reasons: first, the original implementation relied on quantum annealing hardware
whose limited qubit connectivity~\cite{schuetz2025quantum} cannot natively
embed the dense, non-local topology of real-world graphs such as market
correlation networks~\cite{boginski:05}; and second, the performance of the frozen-node identification mechanism was validated only on synthetic unit-disk instances, whose geometric structure and sparsity differ fundamentally from the heterogeneous, dense graphs arising in practical applications.

In this work, we directly address this gap by presenting an end-to-end 
pipeline for quantum-informed portfolio selection that builds on and 
extends the qReduMIS framework to gate-based QAOA. We demonstrate this 
pipeline on real asset correlation data across four major stock indices 
with up to 225 assets, on Quantinuum's 98-qubit trapped-ion Helios 
system \cite{ransford202698}---with QAOA circuits acting on kernels of up to 78 qubits. Our 
contributions are three-fold:
\begin{itemize}
\item \textbf{End-to-end demonstration on hardware.} We deploy the full 
pipeline---from asset correlation data to diversified portfolio 
selection---on trapped-ion hardware across four major stock indices. Notably, standalone QAOA fails to find the optimal solution for the two
largest indices, while qReduMIS (powered by QAOA) achieves success probabilities of up to
$0.95$ with average approximation ratios $\geq 0.96$ across all four
indices.
\item \textbf{Algorithmic extension to gate-based devices.} Market graphs 
exhibit dense, non-local connectivity that exceeds the capabilities of typical quantum annealing hardware. At the algorithmic level, we extend qReduMIS from a quantum annealing algorithm \cite{schuetz2025qredumis} to gate-based QAOA. At the hardware level, we deploy the resulting pipeline on trapped-ion devices, whose all-to-all qubit connectivity is naturally suited for the non-local structure of market graphs---in contrast to Rydberg annealers with local (unit-disk) connectivity. 

\item \textbf{Rigorous benchmark.} We empirically evaluate the pipeline 
across multiple figures of merit on Quantinuum's noisy emulator of the $H2$-1 device, demonstrating that qReduMIS scales 
significantly better with problem size than standalone 
QAOA---reducing the optimal TTS scaling
exponent by $3.4$ times for $p=2$ QAOA layers.
\end{itemize}

The remainder of this paper is organized as follows. In Section~\ref{section:problem_specification}, we formalize the portfolio selection problem as the MIS problem, Section~\ref{section:pipeline} introduces the quantum-informed pipeline to tackle these problems, Section~\ref{section:numerical} discusses the main numerical results, Section~\ref{sec:methodologoy} details the methodology, and Section~\ref{conclusion} concludes with a discussion of implications and future directions.

\section{Problem formulation \label{section:problem_specification}}
\label{sec:problem}

We first formalize our problem statement of selecting multiple financial assets within the framework of the Maximum Independent Set (MIS), which is a well-known graph problem in theoretical computer science. 

\textbf{Definition: Portfolio Selection Problem (PSP).} \textit{For a given universe of $N$ assets with inter-asset correlations $C = [c_{ij}] \in \mathbb{R}^{N \times N}$, and correlation sensitivity (or threshold) $0 \leq \lambda \leq 1$, find the largest (diversified) basket of assets, with no correlation between any pair of selected assets exceeding, in absolute value, a maximum level $\lambda$.} 

The use of the absolute value reflects the role of this step as
\emph{diversification}, not hedging: two strongly anti-correlated assets
are not independent sources of return but mirror images of each other,
and treating them symmetrically with strongly positively correlated pairs
ensures that the selected basket carries genuinely independent
information.
If desired, it is straightforward to adjust this choice in our pipeline as to allow the selection of anti-correlated assets in our basket. 

It is straightforward to see that solving PSP is equivalent to solving MIS on an unweighted graph $\mathcal{G}=(\mathcal{V}, \mathcal{E})$ derived from a \textit{market graph} \cite{kalra:08, hidaka2023correlation, macMahon:15}, where nodes $\mathcal{V}$ represent financial assets and edges are determined by pairwise correlations. Specifically, an edge $(i,j) \in \mathcal{E}$ is placed between assets $i$ and $j$ whenever their correlation $c_{ij}$ satisfies $|c_{ij}| > \lambda$, for a given threshold $\lambda$. Hereafter, we use ``asset correlation graph,'' ``asset graph,'' and ``market graph'' interchangeably to refer to this unweighted graph.

\begin{figure*}
    \centering
    \includegraphics[width=\linewidth]{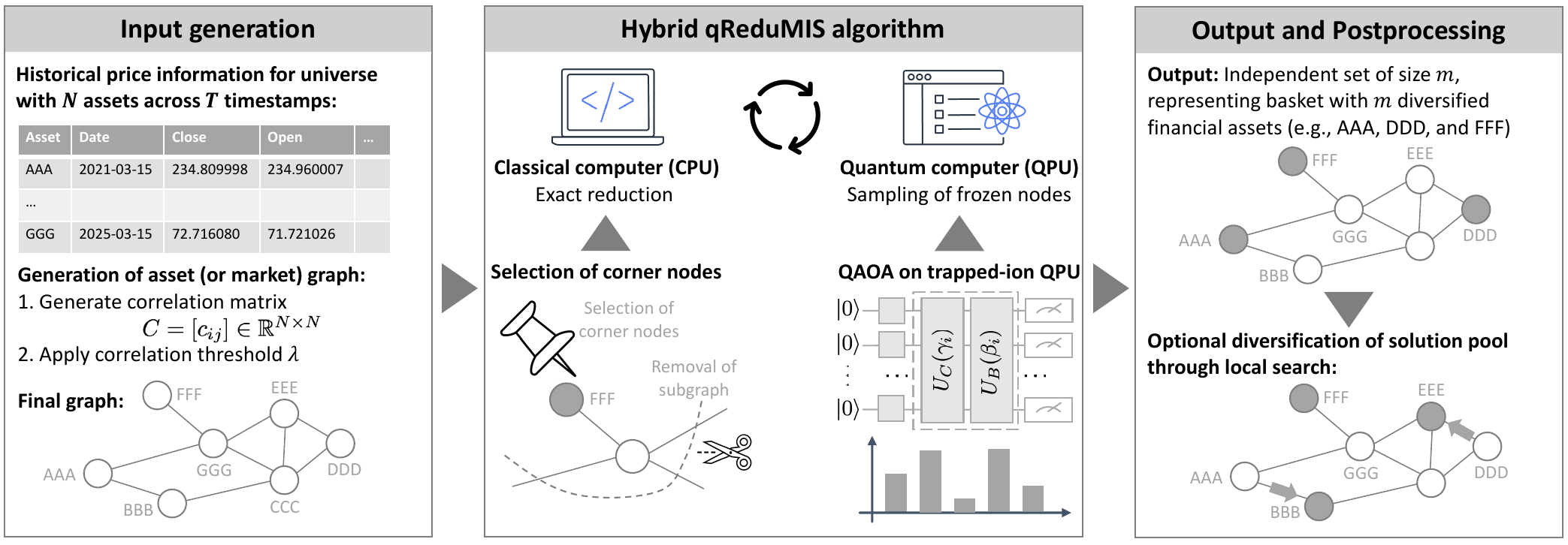}
    \caption{End-to-end pipeline for quantum-informed portfolio selection given a universe of $N$ assets. The input asset (or market) graph is generated from the correlation matrix $C$ based on the assets' returns across $T$ timestamps (diagrammatically shown in the table) and the correlation threshold $\lambda$. This graph defines the input to the hybrid (recursive) qReduMIS algorithm to solve the PSP-MIS problem. qReduMIS intermixes exact reduction (in polynomial time) executed on classical hardware (CPU) and sampling of frozen nodes on the QPU, as to unblock the classical exact reduction wrapper when necessary. The set of frozen nodes is identified from quantum measurements (indicated as the output histogram) obtained through execution of QAOA circuits with the MIS Hamiltonian as cost Hamiltonian run on trapped-ion quantum hardware. The output is given by an independent set of size $m$ that defines a diversified basket of $m$ uncorrelated assets. Colored nodes represent one candidate solution,  with the option to generate additional solutions via simple local search (as indicated by the arrows swapping selection of neighboring nodes).}
    \label{fig:pipeline}
\end{figure*}

\textbf{The portfolio selection problem as MIS (PSP-MIS) on market graphs.} Let $x_{i}=1$ if asset (node) $i=1, \dots, N$ is selected as part of the diversified basket, and $x_{i}=0$ otherwise. 
Given the market graph $\mathcal{G}$, we can then formalize our PSP optimization problem as 
an optimization problem in Quadratic Unconstrained Binary Optimization (QUBO) form, with a Hamiltonian that includes a (soft) penalty to
non-independent configurations as  
\begin{equation}
H = -\sum_{i} x_{i} + U \sum_{(i,j) \in \mathcal{E}} x_{i}x_{j},
\label{eq:hamiltonian_mis}
\end{equation}
with a negative sign in front of the first term because the largest independent set is searched for within a minimization problem, and where the penalty parameter $U$ enforces the independence constraints. Throughout this paper, we take the same setup as Ref.~\cite{farhi2025lower} and set $U=1$. For $U > 1$ the ground state of $H$ is guaranteed to be a MIS; at $U = 1$ the maximum independent sets remain ground states but are degenerate with certain non-independent configurations. Since our post-processing step (Appendix~\ref{appendix:numerical_experiments}) repairs independence-constraint violations, feasibility is enforced after sampling rather than through the penalty itself, making $U = 1$ a sufficient choice.

By doing this, the PSP is formulated as the well-known (NP-hard) MIS problem \cite{garey:90}, where the goal is to find the largest independent set in a graph (i.e., the largest subset of vertices such that no two vertices in this set are connected by an edge), with the independence constraint ensuring the absence of strongly correlated assets (relative to the threshold parameter $\lambda$) in the diversified basket. From now on, we refer to this problem as PSP-MIS.

Note that the correlation among the assets can be estimated by different methods. In this work, we focus on the pair-wise correlation of the daily returns, as explained in detail in Appendix~\ref{appendix:PSP}.

\textbf{Portfolio management.} The PSP framework enables risk diversification by identifying the largest subset of low-correlated assets for portfolio inclusion, which can be integrated into a multi-stage management pipeline where portfolio selection precedes composition~\citep{hidaka2023correlation}. Maximizing this set can potentially enhance flexibility, allowing managers to optimize allocations, adjust sector exposures, and respond to market changes. The PSP solver also reduces the problem size for downstream optimization, potentially addressing scales beyond standard mixed-integer programming. Furthermore, we derive and discuss explicit bounds on the portfolio quality for two allocation strategies---uniform and inverse variance---which are commonly used and have been validated by backtesting on real-world data in previous work~\cite{hidaka2023correlation}; see Theorem~\ref{theo:bound_risk} in Appendix~\ref{appendix:PSP}.

\section{End-to-end pipeline for quantum-informed PSP-MIS \label{section:pipeline}} 

We introduce an end-to-end pipeline for quantum-informed PSP-MIS which consists of three components as shown in Fig.~\ref{fig:pipeline}. First, the inputs are PSP-MIS instances, which are represented by market graphs. Given a universe of $N$ assets, their correlations expressed in a matrix $C$, and a correlation sensitivity (or threshold) $\lambda$, the asset (or market) graph is constructed. This is input into the qReduMIS algorithm. The qReduMIS algorithm is a recursive algorithm to solve the MIS problem for general graphs. qReduMIS performs a number of iterations invoking iteratively the classical computer for the (exact) reduction and then executing QAOA circuits for the resulting kernel graphs on the QPU, in this case a trapped-ion quantum computer. Frozen nodes are sampled from the QAOA measurements and their removal updates the kernel graph that is input to the next algorithm iteration. The final output is an MIS containing $m$ diversified assets. These constitute the largest basket of assets whose pairwise correlations are, in absolute value, below the sensitivity $\lambda$. One can find other alternative independent sets (portfolios) via simple postprocessing for portfolio diversification through local search. In practice this may be interesting because we can generate different portfolios of the same quality with little overhead. These diversified baskets can then serve as input to downstream portfolio construction and allocation stages, as discussed above.

\textbf{The qReduMIS algorithm.} Throughout this paper we use specific terminology to distinguish the
different execution levels within qReduMIS (e.g., \emph{trial},
\emph{iteration}, \emph{QPU call}, \emph{qshot}); precise definitions
are provided in the glossary in Table~\ref{tab:glossary} in Appendix \ref{appendix:glossary}. A diagram illustrating how qReduMIS solves the MIS problem on an example graph is shown in Fig.~\ref{fig:algo_diagram}. qReduMIS begins by applying the classical reduction
techniques to the input graph $\mathcal{G}$, producing a smaller kernel graph
$\mathcal{K}$ while maintaining sets of selected and removed nodes. We primarily
use exposed corner node removal, a simple yet effective kernelization method that
efficiently reduces $\mathcal{G}$ to an irreducible kernel by removing exposed
corner nodes (also known as simplicial vertices), which are guaranteed to be part
of some maximum independent set (MIS)~\citep{butenko:02, strash:16, schuetz2025quantum}.
This process preserves optimality and can be generalized heuristically for weighted
problems~\citep{schuetz2025qredumis}. In Fig.~\ref{fig:algo_diagram}~(b), node~$1$ is a so-called corner
node and it is selected. Thus, its neighbor (node~2) cannot be in
the solution and it is added to the set of removed
nodes. With this, the reduced kernel $\mathcal{K}$ is identified, indicated in
orange in Fig.~\ref{fig:algo_diagram}~(c). Then, this kernel is input to the quantum hardware, as
indicated in violet in Fig.~\ref{fig:algo_diagram}~(d). A quantum optimization algorithm, in this case
QAOA, is executed to generate independent set measurements. In this example, the
measurements (after postprocessing to fix up violations) are shown in the inset table, indicating the nodes belonging to each
independent set (bitstrings) and the respective number of counts corresponding to
those sets.

\begin{figure*}
    \centering
    \includegraphics[width=0.8\linewidth]{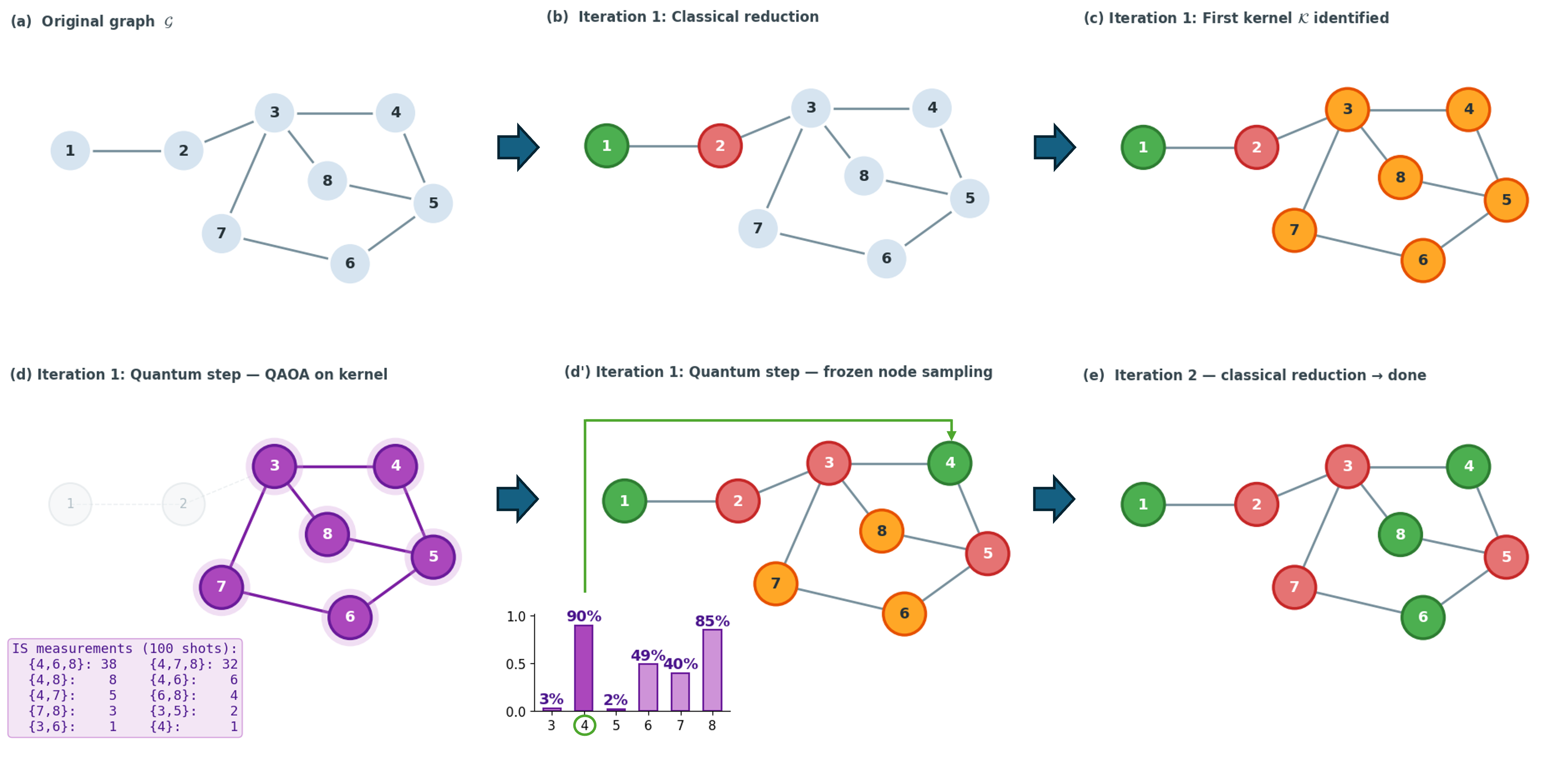}
    \caption{Schematic of the qReduMIS algorithm on an 8-node graph
    $\mathcal{G}$ (illustrative; not experimental data).
    \textbf{(a)}~Original graph~$\mathcal{G}$.
    \textbf{(b)}~Classical reduction: vertex~$1$ (degree~1) is selected;
    its neighbour~$2$ is removed.
    \textbf{(c)}~The kernel $\mathcal{K}$ (orange) is forwarded to the
    quantum subroutine.
    \textbf{(d)}~QAOA is executed on~$\mathcal{K}$ (violet); measured
    independent sets and counts are shown in the inset.
    \textbf{(d')}~Marginal probabilities computed from the measurements
    reveal node~4 as the top-ranked node. Following a greedy strategy over the marginals, it is selected and its
    neighbours~$\{3,5\}$ removed, triggering a new classical reduction.
    In practice, ties in marginal probability are broken at random.
    \textbf{(e)}~Classical reduction fully resolves the remaining subgraph,
    yielding the maximum independent set $I^{*}=\{1,4,6,8\}$ of size~4. A simple postprocessing step based on local search (e.g., swapping
    node~$1$ with~$2$ or node~$6$ with~$7$) can further enhance solution
    diversification while preserving solution quality.}
    \label{fig:algo_diagram}
\end{figure*}

The standard way of using QAOA, and most of the state-of-the-art quantum
algorithms, is to perform an analysis over the independent set measurements with
the goal of finding the optimal solution. With qReduMIS we propose to not only
perform this to keep track of the best solution output by the hardware, but more
importantly, we leverage these measurements to identify \textit{frozen nodes}.
These are nodes with a high probability of being either included in or excluded
from optimal solutions, corresponding to what we define as the in-set and out-set
criteria, respectively. This work focuses specifically on the in-set criterion and
the selection of a single frozen node, though the approach can be extended to
handle multiple frozen nodes simultaneously. To identify a frozen node, we observe
the distribution of the marginal probability of each of the nodes from the QAOA
measurements, as shown in the histogram in Fig.~\ref{fig:algo_diagram}~(d'). In this case,
node~$4$ has the highest marginal probability (90\%) and it is
identified as a frozen node and is selected to be
in the solution and it is removed from the updated kernel graph. As a consequence of its selection, its neighbors (nodes 3 and 5) are removed as they cannot be part of the
solution. With this, the kernel is updated for the next iteration and it consists of the three orange nodes in Fig.~\ref{fig:algo_diagram}~(d'). Then, the next iteration starts with this updated kernel as input and
the classical reduction is performed. Node~8 is now isolated and is selected, and node~6 is a pendant
and is also selected while its neighbour~7 is removed. With this, the algorithm
concludes as shown in Fig.~\ref{fig:algo_diagram}~(e), and an optimal solution is found (nodes in green). Note that in the ideal setting, nodes $4$ and $8$ would have the same marginal probability. 

While several variants of qReduMIS are conceivable, ranging from greedy heuristics to exact branch-and-reduce-type algorithms, here we focus on a semi-greedy implementation thereof. It generates a subset of the top-ranked nodes and then randomly selects one of them.

\textbf{Frozen node sampling.}
We propose a fundamentally different way of leveraging the independent set
measurements, with the goal of correctly identifying frozen nodes whose removal
produces an updated kernel that can be further reduced by the classical reduction.
By keeping track of the largest independent set measured on the QPU, the case in
which just one classical reduction, followed by standard QAOA, solves the problem
to optimality, can be recovered. But most interestingly, our approach can also find
the optimal solution in the cases in which standalone QAOA over the first kernel graph cannot. Imagine that
instead of measuring optimal solutions, the QPU outputs suboptimal solutions (e.g., with size less than three in the example of Fig.~\ref{fig:algo_diagram} these are solutions of size less than three). Still, it is
likely that by looking at the marginal probabilities, we
could identify node~$4$ as a frozen node, or
any other ground-truth frozen node (i.e.,
one that belongs to at least one MIS). Note that the
probability of sampling a frozen node is strictly larger than the probability of
sampling the optimal solution (refer to Theorem~\ref{theo:prob_frozen} in
Appendix~\ref{appendix:runtime}). Thus, unlike in the example, the optimal solution is not observed in the measurement; however, steps in Fig.~\ref{fig:algo_diagram}~(d') and (e) proceed identically, and the algorithm still outputs an optimal solution at the end. In this case, qReduMIS would succeed in finding
the optimal solution while standalone QAOA only finds suboptimal solutions. 

We emphasize that qReduMIS is not merely a preprocessing technique; rather, it introduces a novel way to leverage the output of quantum optimization algorithms. qReduMIS leverages a hybrid \textit{reduce-and-sample} paradigm; it goes
through iterations of classical reduction followed by the routine executed on
the QPU. The number of iterations depends on the problem itself and on the trial, as
this is a semi-greedy algorithm.

The number of independent set measurements used to build the distribution of marginal probabilities for each node is a hyperparameter. We call this value $\mathrm{qshots}$, which represents the number of quantum shots (measurements) performed in each call to the QPU. The optimal number of $\mathrm{qshots}$ can vary depending on the specific problem instance and the stage of the algorithm (iteration) when the QPU is called. In practice, one could implement a scheduler to adjust $\mathrm{qshots}$ dynamically. For more details, see Appendix~\ref{appendix:empirical_tts}.

Market graphs tend to be highly interconnected, as most assets exhibit
some degree of correlation, and many implicit correlations exist (e.g.,
two assets may be correlated through a shared intermediary). After applying the correlation sensitivity (or threshold) $\lambda$ to obtain the (unweighted) market graphs, these become sparser but still substantially connected. Crucially, the
resulting edges do not necessarily reflect local graph structure, producing a non-trivial topology that lacks regularity or locality. This non-local edge structure of market graphs places specific demands on the quantum hardware: a backend with high qubit connectivity is strongly preferred, as limited connectivity would require additional overhead to implement the necessary gate operations between non-neighboring qubits. This motivates our choice of trapped-ion quantum processors, which offer all-to-all qubit connectivity and are thus naturally suited to the dense, non-local structure of unweighted (thresholded) market graphs. To generate independent-set samples for identifying frozen nodes, we use QAOA \citep{farhi:14, zhou:20}, which is a gate-based variational algorithm, as the quantum subroutine in qReduMIS. QAOA is one of the most widely studied quantum algorithms for combinatorial optimization on near-term, gate-based hardware \citep{farhi:14}, with experimental demonstrations on both superconducting qubits \citep{otterbach:17, harrigan2021quantum} and trapped ions \citep{pagano:20}.

Note that QAOA parameters can, in principle, be optimized independently via a classical-quantum variational loop and subsequently incorporated into the qReduMIS framework as fixed circuit parameters. Although such optimization is in general a challenging task~\cite{bittel2021training}, pre-trained parameters combined with established scaling techniques~\citep{galda2021transferability, brandao2018fixed} can be employed to mitigate this difficulty. In this work, we bypass the optimization step entirely by employing pre-trained parameters scaled analytically to each kernel's average degree~\citep{vcepaite2025quantum}. This is sufficient because qReduMIS does not rely on the QPU to produce exact solutions but rather leverages it to detect signatures of frozen nodes---a task for which QAOA performs remarkably well even without fully optimized parameters. This choice extends the qReduMIS framework to gate-based quantum devices, complementing prior implementations on analog quantum architectures.

\section{Experimental results on trapped-ion quantum hardware \label{section:numerical}}

We now assess the performance of qReduMIS experimentally. Our evaluation is organized in two parts. In Section~\ref{subsection:large_PSP}, we demonstrate the end-to-end pipeline at scale on Quantinuum's 98-qubit trapped-ion Helios system, solving PSP-MIS instances derived from four major stock indices with up to 225 assets. In Section~\ref{subsection:benchmark_size}, we perform a systematic benchmark on the Quantinuum H2-1 noisy emulator, studying how performance scales with the problem size and hardness. Throughout, we compare standalone QAOA against a semi-greedy implementation of qReduMIS powered by QAOA with the in-set frozen-node selection strategy.

To ensure a fair comparison, both methods are evaluated over the first kernel graph $\mathcal{K}$---the reduced graph obtained after the initial classical reduction. This choice grants standalone QAOA the same classical preprocessing benefit used by qReduMIS, so that any observed advantage cannot be attributed to the reduction alone. It also enables comparison at larger problem sizes, since the kernel fits within the qubit budget of current devices. Importantly, all problem instances considered in this work are not fully reducible, meaning that the classical reduction rule alone is insufficient and at least one QPU call is required to find the optimal solution; this ensures that the quantum component plays a non-trivial role in every instance. As shown in Ref.~\cite{schuetz2025qredumis}, qReduMIS with a QPU subroutine outperforms random frozen-node selection, confirming that the quantum measurements provide strategically useful information beyond a purely random baseline. Despite this favorable baseline for standalone QAOA, we show that qReduMIS achieves superior performance across multiple figures of merit, confirming that its advantage stems from the iterative interplay between classical reduction and quantum subroutine calls, and from its novel use of QPU measurements to identify frozen nodes rather than to directly sample optimal solutions. All QAOA measurements---both in standalone mode and within
qReduMIS---undergo identical postprocessing that repairs
independence-constraint violations, ensuring
that any observed performance difference is attributable solely to the
algorithmic framework. More details are provided in Section~\ref{sec:methodologoy}. 

Our benchmark focuses on comparing qReduMIS against standalone QAOA to assess the value of the hybrid reduce-and-sample paradigm relative to direct quantum optimization; a comparison against state-of-the-art classical solvers is included for completeness in Section~\ref{subsection:benchmark_size}, though the primary objective is to establish that qReduMIS extracts significantly more useful information from limited QPU resources than standard quantum approaches.

\textbf{PSP-MIS instances.} We consider PSP-MIS instances constructed from assets belonging to major stock indices. We utilize the weekly returns to construct the correlation matrices using the datasets from Chang \etal{}~\cite{chang2000heuristics}, made available
via the open-source repository of Knazakis~\cite{knazakis_portfolio}. The sizes of these indices are reported in Table~\ref{tab:instance_characteristics}. From these indices, we consider instances of varying hardness to study the performance of the proposed algorithm across the difficulty spectrum. Problem hardness is
driven by structural properties of the instance---primarily problem size,
graph density, and tree width~\cite{cazals2025identifying}. However,
hardness is not solely an intrinsic property of the instance: it also
depends on the solver, since an instance that is challenging for one
algorithm may be easily solved by another. It has been empirically observed that, for local-update Markov-chain Monte Carlo (MCMC) methods such as simulated annealing and, through a related mechanism tied to the instance's degeneracy structure, for quantum annealing, the conductance-like hardness parameter $\mathbb{H}$ captures instance
difficulty \cite{andrist:23, ebadi:22, schuetz2025qredumis}. It is defined as $\HP =D_{\mathrm{|MIS|-1}}/(|\mathrm{MIS}|\cdot D_{|\mathrm{MIS}|})$, where $D_{\alpha}$ denotes the degeneracy of the independent sets of size $\alpha$. We include an analysis of performance as a function of $\mathbb{H}$ in
Appendix~\ref{appendix:benchmark_extensive}.

\begin{figure*}
    \centering
    \includegraphics[width=\linewidth]{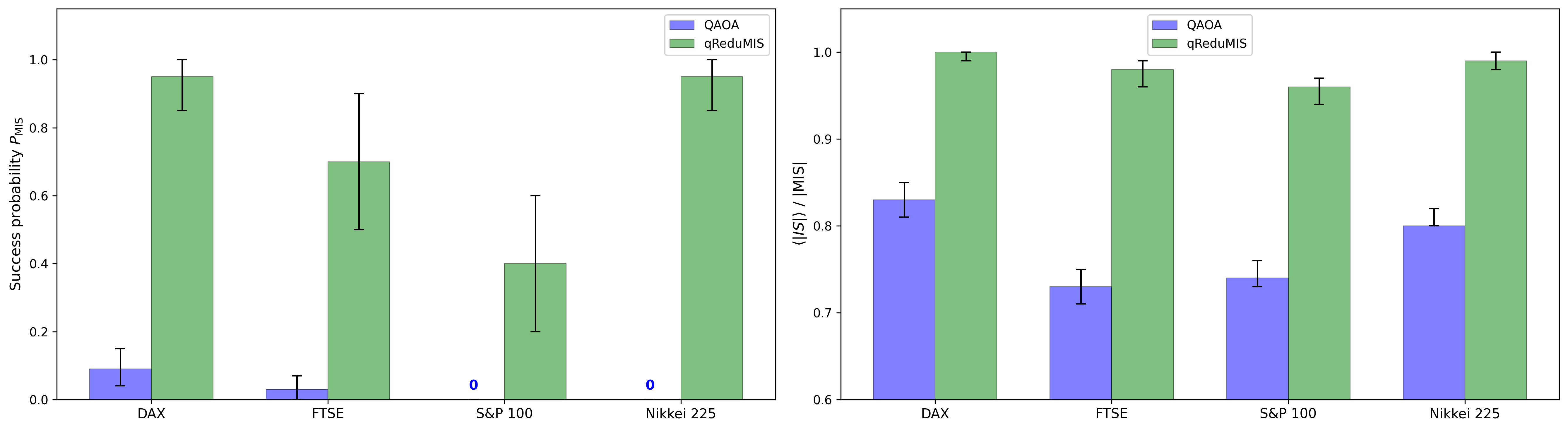}
    \caption{Performance comparison of standalone QAOA (blue) and qReduMIS (green)
    on the first kernel $\mathcal{K}$ for four stock market indices
    (DAX, FTSE, S\&P~100, Nikkei~225).
    \textbf{Left:} Success probability;
    \textbf{Right:} Average approximation ratio $\langle \text{AR} \rangle$.
    Error bars indicate 95\% bootstrap confidence intervals obtained via
    10,000 resamples. For S\&P~100 and Nikkei~225, standalone QAOA achieves
    $P_{\text{MIS}} = 0$. qReduMIS consistently achieves
    higher success probability, and near-optimal approximation
    ratio across all indices. Standalone QAOA uses 100 quantum shots;
    qReduMIS uses 20 classical trials with 10 quantum shots per QPU call, requiring no more than five QPU calls. 
    All experiments use $p=2$ QAOA layers on Quantinuum's Helios hardware
    (98 qubits).}
    \label{fig:main_results_hw}
\end{figure*}

\textbf{Figures of merit.} We evaluate three complementary figures of merit. First we consider the end-to-end success probability $P_{\text{MIS}}$ defined as the fraction of independent trials returning a maximum independent set, which captures the success rate on a single trial of the algorithm. The second metric is the average
approximation ratio $\langle |\text{IS}| \rangle / |\text{MIS}|$, quantifying
average solution quality. For 
standalone QAOA, the weighted average is over shot counts; for 
qReduMIS, it is over the outputs of each classical trial. The third is the optimal time-to-solution
(optTTS)~\cite{PhysRevX.8.031016,Kowalsky_2022,mohseni2022ising,leng2023quantum},
which is used extract runtime scaling. This metric captures the trade-off between per-trial runtime and success
probability, defined as
\begin{equation}
\mathrm{optTTS} = \min_{t_f}\, t_f \cdot \frac{\ln(1-p_d)}{\ln(1-p_S(t_f))},
\label{eq:def_opttts}
\end{equation}
where $t_f$ is the runtime of a single trial, $p_S(t_f)$ is the
probability of finding the optimal solution within that runtime, and
$p_d=0.99$ is the target confidence. We measure $t_f$ in cumulative
quantum shots, treating each shot as a single unit of time; this
uniform-cost model makes the reported optTTS for qReduMIS an upper bound
on the wall-clock optTTS in practice. Full details on the runtime model for both standalone
QAOA and qReduMIS are provided in Section~\ref{sec:methodologoy}. Ground-truth solutions were obtained with OR-Tools~\cite{cpsatlp}.

\begin{table}
    \centering
    \begin{tabular}{l|c|c|cc|ccc}
        \hline
        \multirow{2}{*}{Stock Market} 
            & \multirow{2}{*}{Size}
            & \multirow{2}{*}{$\lambda$} 
            & \multicolumn{2}{c|}{Kernel $\mathcal{K}$} 
            & \multicolumn{3}{c}{QAOA Circuit} \\
        \cline{4-8}
            & & & Nodes & Density & Qubits & Gate & Depth \\
        \hline
        DAX 100    & 85  & 0.24 & 49 & 0.29 & 49 & 692     & 104 \\
        FTSE 100   & 89  & 0.32 & 45 & 0.28 & 45 & 574     & 82  \\
        S\&P 100   & 98  & 0.24 & 68 & 0.20 & 68 & 914     & 96  \\
        Nikkei 225 & 225 & 0.62 & 78 & 0.17 & 78 & 1016 & 110 \\
        \hline
    \end{tabular}
    \caption{PSP-MIS instances and QAOA circuit characteristics for the
    four stock market indices. For each index, the table reports the
    total number of assets (Size), the correlation threshold $\lambda$,
    the number of nodes and edge density of the first kernel graph
    $\mathcal{K}$ obtained after classical reduction, and the QAOA
    circuit characteristics (number of qubits, 2-qubit gate count, and
    2-qubit circuit depth) corresponding to $\mathcal{K}$. The number of
    qubits equals the kernel size $|\mathcal{K}|$, since QAOA acts
    directly on the kernel graph. All circuits use $p=2$ QAOA layers.
    This circuit is used in the first QPU call of qReduMIS; subsequent
    iterations operate on progressively smaller kernels and thus require
    strictly smaller circuits.}
    \label{tab:instance_characteristics}
\end{table}

\subsection{Large-scale quantum-informed PSP-MIS on hardware \label{subsection:large_PSP}} 

We evaluate the end-to-end pipeline on Quantinuum's 98-qubit trapped-ion Helios system. To the best of our knowledge, this constitutes one of the largest gate-based QAOA demonstrations on real-world combinatorial optimization instances to date. We consider the four stock market indices
and, for each one, select a correlation threshold $\lambda$ such that the
first kernel graph $\mathcal{K}$ fits within the qubit limit and its density
keeps the 2-qubit gate count below 1000.
Table~\ref{tab:instance_characteristics} reports the characteristics of each
kernel and the corresponding QAOA circuit. We set $p=2$ QAOA layers for both
standalone QAOA and the QAOA component of qReduMIS, with parameters scaled
following Eq.~\ref{eq:scaling_params}. Note that the same circuit is used for
both standalone QAOA and the first QPU call of qReduMIS; subsequent qReduMIS
iterations operate on progressively smaller kernels, requiring strictly
shallower circuits.

Importantly, the classical reduction step is not merely a convenience for fair comparison: in the case of the Nikkei~225 (225 assets), encoding the full input graph would require $225$ qubits, exceeding the $98$-qubit capacity of the Helios device, so the instance is addressable on this hardware \emph{only} because qReduMIS operates on the reduced kernel ($|\mathcal{K}|=78$ qubits). For the remaining indices, the reduction yields substantially smaller and shallower QAOA circuits than would be required to encode the input graph directly.

\begin{figure*}
    \centering
    \includegraphics[width=\linewidth]{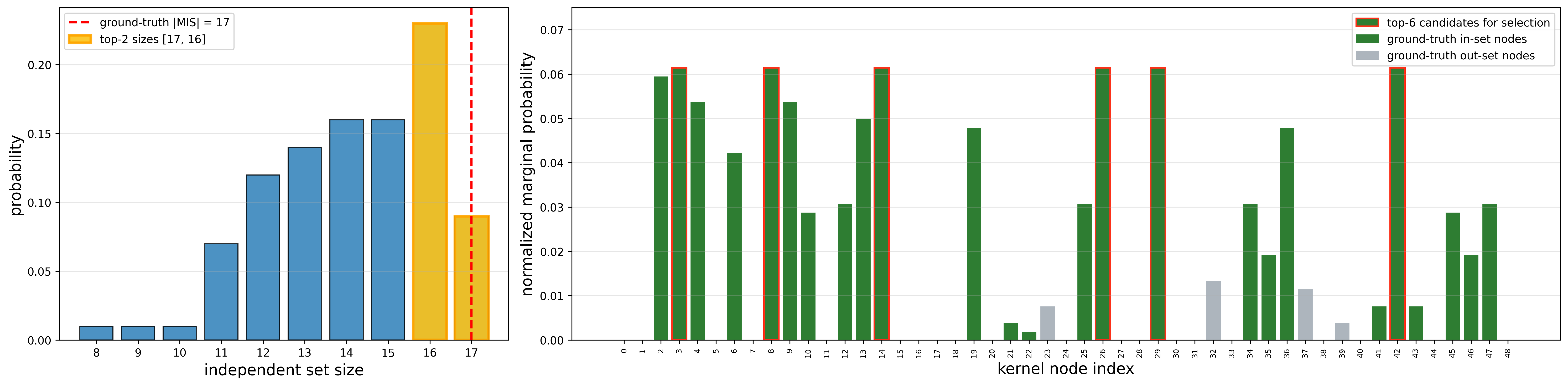}
    \caption{QAOA measurements on the DAX 100 kernel ($|\mathcal{K}|=49$ nodes, 
    100 shots on Quantinuum Helios) reveal a strong frozen-node signal despite low 
    success probability. 
    \textbf{Left:} Probability distribution over independent-set sizes from 
    standalone QAOA; the dashed red line marks the ground-truth $|\text{MIS}|$ 
    and yellow bins indicate the top-2 largest sizes measured. The probability 
    of sampling the optimal solution is below 10\%. 
    \textbf{Right:} Normalized marginal probability of each kernel node, 
    computed over the top-2 largest independent sets (yellow bins on left panel). Green bars indicate 
    in-set nodes (belonging to at least one MIS) and grey bars indicate 
    out-set nodes (absent from all MIS solutions). Approximately 96\% 
    of the marginal probability mass concentrates on ground-truth in-set  nodes, confirming that QAOA measurements carry sufficient signal 
    for qReduMIS even when the optimal solution is rarely sampled directly. The six bars with red edges mark the top-6 in-set frozen nodes, candidates for selection by marginal probability; all six are ground-truth in-set frozen nodes.} %\ms{Explain (27) vs (22) in inset. Among green bars one could furthe color-highlight the top-$K$ candidates for selection, e.g., $K=5$ in orange.}}
    \label{fig:comparison_counts}
\end{figure*} 

For standalone QAOA we use 100 quantum shots. Due to practical constraints on
quantum hardware access, we set the number of classical shots of qReduMIS to
$M = 20$ and report all figures of merit with 95\% bootstrap confidence
intervals (10,000 resamples with replacement) to quantify the resulting
statistical uncertainty. Each QPU call uses $\mathrm{qshots} = 10$ quantum
shots. As discussed, $\mathrm{qshots}$ is a hyperparameter; we have observed
that for the problem instances considered, $\mathrm{qshots} = 10$ achieves a
good trade-off between the total number of QPU calls to solution~($t_f$) and
the end-to-end success probability. This trade-off is captured by the optTTS
calculation, where $R(t_f)$ (Eq.~\ref{eq:TTS_tf_1}) balances the two
quantities.

Fig.~\ref{fig:main_results_hw} presents two figures of
merit---end-to-end success probability and average approximation ratio
$\langle \text{AR} \rangle$---for each stock market index. All point
estimates are accompanied by 95\% bootstrap confidence intervals. We employ an enhanced postprocessing strategy
that not only removes constraint-violating nodes but also greedily reinserts
non-conflicting ones. This grants standalone QAOA an additional advantage,
since qReduMIS does not rely on sampling the optimal solution directly but
rather on identifying frozen nodes via measurement statistics. Despite this
enhanced postprocessing, standalone QAOA fails to sample the optimal solution
within 100 shots for the S\&P~100 and Nikkei~225 indices, yielding
$P_{\text{MIS}} = 0$. In contrast, the noisy signal produced by the hardware remains sufficient for
qReduMIS to identify frozen nodes and iteratively solve the instance to optimality. Across all four indices, qReduMIS achieves a dramatically higher success
probability than standalone QAOA. 

The confidence intervals on $P_{\text{MIS}}$ are obtained by a nonparametric
percentile bootstrap over the independent trials (resampling the $n$ trials
with replacement and taking the $2.5$--$97.5$ percentiles of the resulting
$\hat{P}_{\text{MIS}}$). Their width therefore reflects binomial sampling
uncertainty: it scales as $n^{-1/2}$ and is largest for instances with
intermediate success probability, where the per-trial variance
is maximal. Because qReduMIS is evaluated
over $n=20$ independent runs, its intervals are correspondingly coarser than
those of QAOA. For completeness, we report the $\text{optTTS}$ in Table~\ref{tab:performance_results} of Appendix~\ref{appendix:benchmark_extensive}.

The average approximation ratio further underscores the gap. As shown in Table~\ref{tab:performance_results} in Appendix~\ref{appendix:benchmark_extensive}, qReduMIS
consistently achieves $\langle \text{AR} \rangle \geq 0.96$, whereas
standalone QAOA ranges between $0.73$ and $0.83$, indicating that even
after postprocessing, the independent sets produced by standalone QAOA tend
to be substantially smaller than the optimum.

Regarding the number of QPU calls required for qReduMIS, a single trial of qReduMIS runs until the stopping criterion is met or the kernel becomes fully reducible. Refer to Appendix~\ref{appendix:mkt_graphs_hw} for a discussion.
Across the four indices considered, no more than five QPU calls are required for all their respective trials to reach the optimal solution. With $\mathrm{qshots}=10$ quantum shots per QPU call, this means that no
more than a total of $50$ quantum shots are required for an end-to-end
execution of qReduMIS.

\textbf{Frozen-node signal from QAOA measurements.} We illustrate the quantum mechanism underlying qReduMIS with the experimental results corresponding to the DAX 100 index. Standard QAOA is performed to produce a distribution of independent set measurements with the hope that the largest solution sampled corresponds to the optimal solution. We observe on the left panel of Fig.~\ref{fig:comparison_counts} that the optimal solution is indeed being measured but its probability is very low, below $10\%$. In qReduMIS, we instead observe the marginal probability distribution over the nodes in the kernel to identify candidates for selection. We identify them by recognizing frozen nodes via their marginal probabilities: in-set frozen nodes exhibit high marginal probability, indicating a strong likelihood of belonging to the independent set, while out-set frozen nodes exhibit low marginal probability, indicating consistent exclusion. A correct identification means that the selected in-set (out-set) frozen nodes coincide with ground-truth in-set (out-set) nodes---that is, nodes present in at least one MIS solution and nodes absent from every MIS solution, respectively.

We enumerated all degenerate MIS solutions in order to identify the ground-truth in-set and out-set nodes. We observe that although standalone QAOA success probability is low, approximately 96\% of the normalized marginal probability mass is concentrated on ground-truth frozen (in-set) nodes (as shown on the right panel with the green bars). qReduMIS exploits this concentration: rather than requiring a full MIS to
be sampled in a single shot, it suffices that one ground-truth
in-set (out-set) is selected as randomly sampled from the top-ranked marginal nodes (the ones with the red edges). Specifically, we
select the $M=4$ highest-marginal-probability nodes and randomly choose one
as the frozen node to fix, removing it together with its neighbors. Note that $M$ should be adaptive: for small kernels, a large $M$ may exceed the number of candidate nodes. This
semi-greedy selection introduces diversity across trials while still
leveraging the strong frozen-node signal in the marginal distribution.

These results demonstrate that qReduMIS effectively solves PSP-MIS at practical scale on current hardware; we next investigate how its performance scales systematically with the problem size and hardness.

\subsection{Systematic benchmark on PSP-MIS instances \label{subsection:benchmark_size}} 

We assess the performance of the end-to-end pipeline to solve PSP-MIS instances of different hardness using the Quantinuum's noisy emulator of the $H2$-1 device. We generate PSP-MIS instances of sizes ranging from 26 to 44 in increments of two, and for each size, we randomly select assets from the Nikkei 225 index to make 20 problem instances per problem size. We consider different problem sizes and set $\lambda$ equal to the mean of all entries of the correlation matrix $C$:
\begin{equation}
\lambda = \frac{1}{N^2}\sum_{i,j} |C_{ij}|, 
\label{eq:average_corr}
\end{equation}

where $N$ is the number of assets. Note that the value of $\lambda$ depends on the problem instance.

The PSP-MIS instances are reduced, resulting in $73$ instances that are not fully reducible, spanning both different sizes and first kernels of different densities. We discuss the reduction factor over these instances as well as densities of both the input graph and the first kernel in Appendix~\ref{appendix:benchmark_extensive}. These 73 non-trivially-reducible instances form the testbed for both standalone QAOA and qReduMIS. Both methods use the same QAOA parameters, scaled following Eq.~\ref{eq:scaling_params}, and the same number of layers $p$. We consider $p=2$ and $p=6$. 
The number of quantum shots employed to assess the performance of standard QAOA is $500$. The number of $\mathrm{qshots}$ in qReduMIS is a hyperparameter that can, in principle, be fine-tuned and depend on the kernel at each iteration. For this experiment, the $\mathrm{qshots}$ in qReduMIS was fixed to $5$. The classical shots (or trials) performed for qReduMIS for each problem instance was $20$.  

In the following, we assess the performance of qReduMIS powered by QAOA
against standalone QAOA, evaluating both the optTTS and the expected
approximation ratio. Across both metrics, qReduMIS consistently outperforms
standalone QAOA. For $p=2$, qReduMIS exhibits a significantly improved
optTTS scaling exponent and a higher average approximation ratio.
Increasing to $p=6$ improves both methods. This is a promising trend, as advances in hardware
fidelity will progressively enable the reliable execution of deeper QAOA
circuits, further benefiting qReduMIS and standalone QAOA alike. In the setting of $p=6$, qReduMIS continues to achieve
a significantly higher approximation ratio; however, in this regime the
optTTS of qReduMIS shows no clear scaling trend over the problem-size
range considered, precluding a meaningful fit.

\begin{figure*}[!t]
    \centering
    \includegraphics[width=\linewidth]{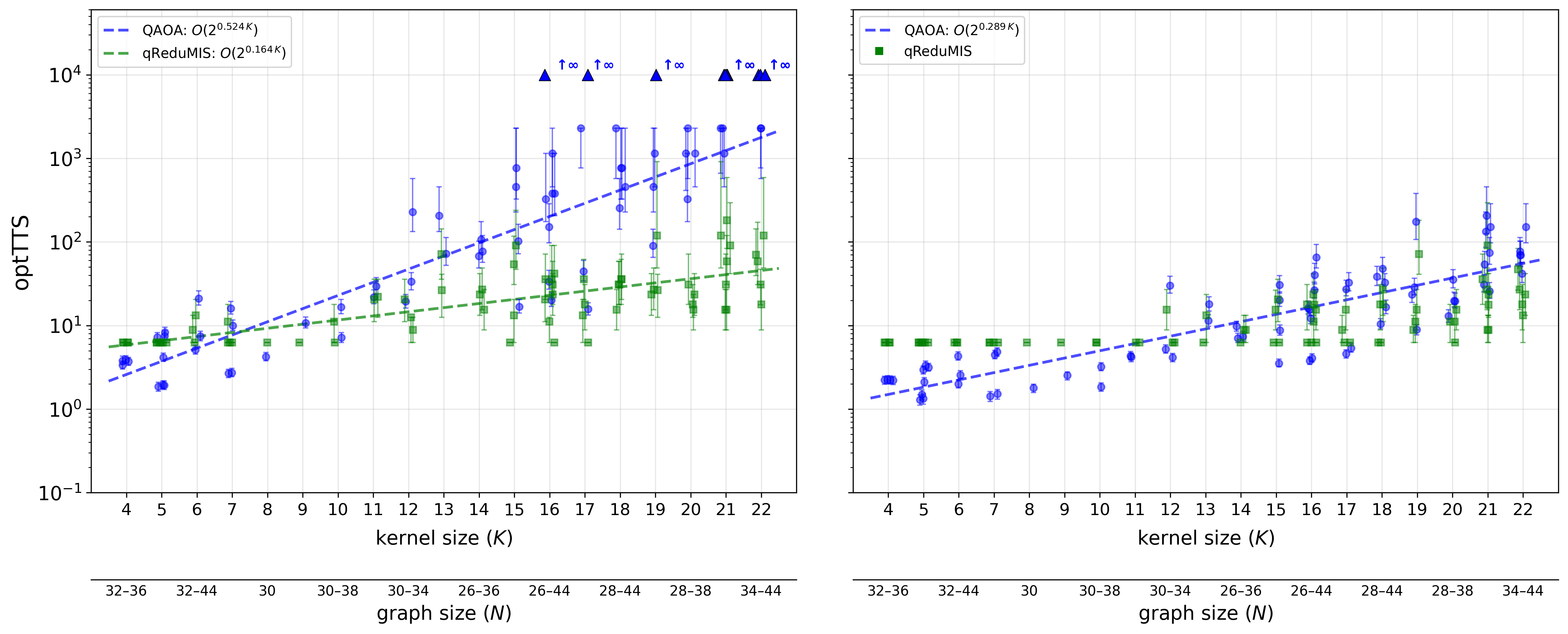}
    \caption{Optimal time-to-solution (optTTS) as a function of the first 
    kernel size $|\mathcal{K}|$ for standalone QAOA (blue circles) and 
    qReduMIS (green squares) on PSP-MIS instances from the Nikkei 225 index. 
    \textbf{Left:} $p=2$; \textbf{Right:} $p=6$. Error bars denote 95\% 
    bootstrap confidence intervals on the median optTTS per problem instance. 
    Dash-dotted lines show ordinary least squares exponential fits $\mathcal{O}(2^{\beta K})$, 
    with scaling exponents reported in each legend. For $p=2$, the fitted exponents are $\beta_{\mathrm{QAOA}} = 0.524 \pm 0.033$ 
    and $\beta_{\mathrm{qReduMIS}} = 0.164 \pm 0.018$, a $3.2$ times reduction (see Appendix~\ref{appendix:benchmark_extensive} for statistical tests). For $p=6$, qReduMIS 
    exhibits no statistically significant exponential scaling over the tested 
    range. Upward triangles mark instances where QAOA failed to find any MIS 
    (optTTS$=\infty$); these are excluded from the fits, making the QAOA 
    exponent a lower bound on its true value. Data points are 
    horizontally offset by $\pm0.25$ for visual clarity.}
    \label{fig:tts_qredumis_vs_qaoa}
\end{figure*}

The optTTS as a function of the (first) kernel size is shown in Fig.~\ref{fig:tts_qredumis_vs_qaoa}. We employed the bootstrap method with 10,000 resamples to estimate a $95\%$ confidence interval for the optTTS calculations of the two methods. Each panel corresponds to a number of layers in the QAOA implementations, both as standalone and as part of qReduMIS. We compare the performance of qReduMIS powered by QAOA to standalone QAOA. The number of quantum shots employed by QAOA as part of qReduMIS is a hyperparameter. We benchmarked optTTS for $\mathrm{qshots} \in \{5, 10, 25, 50, 250, 500\}$ 
and report the results for $\mathrm{qshots}=5$, which yields the best 
optTTS scaling. Refer to Fig.~\ref{fig:quantum_shots_qredumis} in Appendix for a detailed analysis for one particular problem instance.

For $p=2$ both standalone QAOA and qReduMIS exhibit an exponential scaling of optTTS as a function of the size of the (first) kernel $|\mathcal{K}|$, as shown with the fittings in the figure. The Ordinary Least Squares (OLS) regression is over the median values of each single problem instance, corresponding to 95$\%$ bootstrap confidence intervals. Refer to Appendix~\ref{appendix:benchmark_extensive} for details of the fits. The exponent of the scaling for QAOA is significantly reduced as $p$ is increased from two to six, as expected. Note that for $p=2$, for some instances and trials, the success probability of standalone QAOA is null ($P_{\text{MIS}}=0$), causing the optTTS to be infinite, which is indicated with the triangles showing $\infty$. We report the fraction of these among intervals of problem sizes in a heatmap in Fig.~\ref{fig:heatmap_inf_optTTS} in Appendix~\ref{appendix:benchmark_extensive}. These points are therefore excluded from the fitting. Since the excluded instances correspond to the hardest cases---where QAOA fails entirely---the fitted scaling exponent for QAOA represents a lower bound on its true scaling difficulty; the actual performance gap between QAOA and qReduMIS is therefore likely larger than what the fits suggest. qReduMIS reduces the scaling exponent by a factor of $3.2$ relative to standalone QAOA ($\beta_{\text{qReduMIS}} = 0.164$ versus $\beta_{\text{QAOA}} = 0.524$). Regarding the number of QPU calls in qReduMIS, Fig.~\ref{fig:qpu_to_solution_distribution_emulator} (left panel) in Appendix shows the number of QPU calls to solution remains consistently low, not exceeding $2$ across all problem instances considered. 

In the case of $p=6$, qReduMIS shows a nearly constant optTTS. This is because, for most of the kernel sizes considered, qReduMIS solves the problem to optimality by performing a classical reduction followed by only a single QPU call and a subsequent reduction. This behavior is expected: although the success probability of standalone QAOA decays exponentially with kernel size, for kernels below size $18$ the median success probability remains above $\sim 20\%$, and even a modest number of quantum shots suffices to either find the optimal solution directly or correctly identify an in-set frozen node whose removal, together with its neighbors, unblocks the next classical reduction to reach optimality. Accordingly, the QPU calls to solution are concentrated around one, as shown in Fig.~\ref{fig:qpu_to_solution_distribution_emulator} (right panel) in the Appendix. It is only for larger kernel sizes that we observe a dispersion towards higher values, producing a slight increase in the optTTS. Notably, we identify a crossing point beyond which qReduMIS begins to outperform standalone QAOA for some instances, though this advantage is less pronounced than in the $p=2$ case. One potential explanation is that the quantum-shot budget was kept the same as for $p=2$; given the increased expressiveness of the deeper circuit, fewer shots may be needed to achieve comparable or better frozen-node identification, and reducing the shot count would directly improve the optTTS. Regarding the scaling, we expect that a clearer decay in optTTS will emerge for larger problem sizes, enabling a more robust fitting of the scaling behavior.

Significant dispersion in optTTS is observed for a given kernel size. This arises because each kernel-size bin aggregates instances from 
different input graphs $\mathcal{G}$ with varying kernel densities, 
and density is an additional driver of problem hardness \cite{cazals2025identifying}. The distribution of the kernel densities per kernel size is shown in the lower panel in Fig.~\ref{fig:testbed_emulator} in the Appendix. 

To assess the quality of the solution estimation, we consider the average approximation ratio. Fig.~\ref{fig:expected_AR_qaoa_qredumis} presents the average approximation ratio as a function of the (first) kernel size, with the input graph size shown on the secondary x-axis. As before, results are reported for both $p=2$ and $p=6$. While QAOA exhibits a significant decline in performance as the problem size increases, qReduMIS remains robust. For $p=2$, a slight—though not statistically significant—decay is observed at larger sizes; however, the median per kernel size remains consistently equal to $1$.

\begin{figure*}
    \centering
    \includegraphics[width=\linewidth]{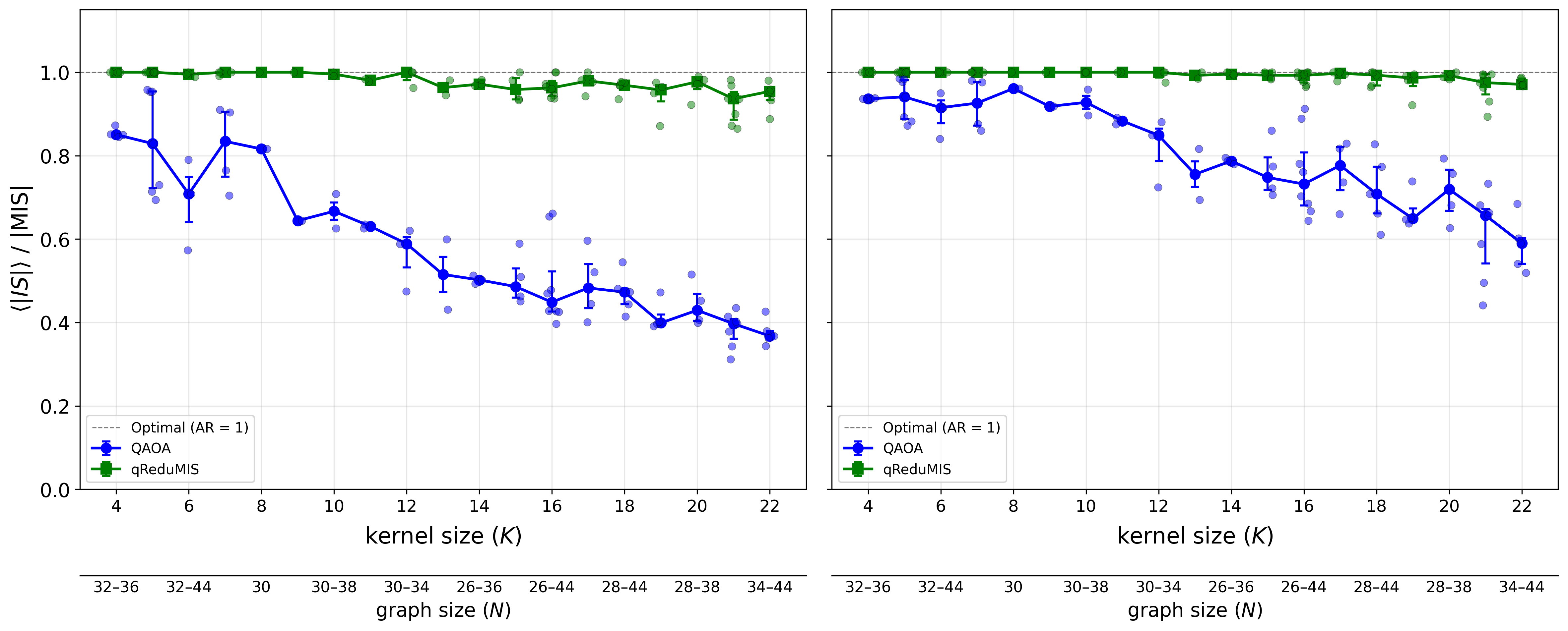}
    \caption{Average approximation ratio $\langle|\text{IS}|\rangle / 
|\text{MIS}|$ as a function of kernel size $|\mathcal{K}|$ for standalone 
QAOA and qReduMIS on PSP-MIS instances from the Nikkei 225 index. 
\textbf{Left:} $p=2$; \textbf{Right:} $p=6$. Translucent scatter points 
show individual instance values; markers with error bars indicate the 
median and interquartile range per kernel size. Methods are horizontally 
offset for clarity, and the secondary $x$-axis shows the corresponding 
input graph size $N$. While standalone QAOA exhibits a significant decline 
in solution quality with increasing kernel size, qReduMIS maintains 
$\langle|\text{IS}|\rangle / |\text{MIS}| \approx 1$ across all sizes for 
both $p=2$ and $p=6$, demonstrating robust near-optimal performance 
independent of problem scale.}
    \label{fig:expected_AR_qaoa_qredumis}
\end{figure*}

\textbf{Benchmark against classical solver.} To contextualize the performance of our quantum-classical hybrid approach, we perform a benchmark over the same problem instances with simulated annealing (SA), a well-established classical metaheuristic for combinatorial optimization. We emphasize that the primary contribution of this work is not to claim advantage over state-of-the-art classical solvers---indeed, SA is exceptionally effective on the instance sizes
accessible to current quantum hardware. The SA results, reported in detail in Appendix~\ref{appendix:mkt_graphs_hw},
confirm that SA trivially solves most Nikkei market-graph instances at
the scales considered here (kernel sizes $K \leq 22$), with optTTS
values clustering at the estimator floor of $\approx 2/3$. As shown in Fig.~\ref{fig:fig_with_sa}, only a
small fraction of instances yield non-trivial optTTS, which is
insufficient to extract a reliable scaling exponent for SA at these
sizes.

We include this benchmark to establish a concrete classical baseline
and to delineate the regime in which quantum-classical hybrid methods
must improve to become competitive.  At the instance sizes accessible
to current quantum hardware (kernel sizes $K \leq 22$), SA solves the
vast majority of problems trivially, precluding a meaningful scaling
comparison. As hardware matures and larger, harder instances become
accessible, classical heuristics such as SA are known to
struggle---precisely the regime in which the favorable scaling of
qReduMIS over standalone QAOA may translate into a practical
advantage.

\textbf{Performance with conductance-like hardness.} We further investigate performance as a function of the conductance-like hardness 
parameter $\mathbb{H}$ in Appendix~\ref{appendix:benchmark_extensive}, 
confirming that qReduMIS degrades gracefully with increasing hardness 
while standalone QAOA deteriorates sharply.

\textbf{qReduMIS for arbitrary graph topologies. }Beyond the structured market graphs studied here, we also benchmark 
qReduMIS on unstructured random 3-regular graphs in 
Appendix~\ref{appendix:3regular}. On these instances—where the 
classical reduction is initially ineffective due to the graph's 
regularity—qReduMIS still achieves $P_{\text{MIS}}=1$ across all 
problem sizes tested, confirming that the framework's advantage is 
not specific to the topology of financial correlation graphs.

\section{Methods \label{sec:methodologoy}}

\textbf{QAOA implementation both as standalone and in qReduMIS.} At each QPU call, we execute QAOA over the kernel graph. For this, we encode the MIS problem in Eq.~\ref{eq:hamiltonian_mis} by first mapping the classical binary variables $x_i \in \{0,1\}$, which indicate asset selection, to quantum variables $z_i \in \{-1,1\}$ via the relation $x_i = (1-z_i)/2$. This transformation allows us to express the problem in terms of Pauli-$Z$ operators, which are natural for quantum circuit implementations. The resulting cost Hamiltonian, $\hat{H}_{\mathrm{cost}}$, penalizes configurations that violate the independence constraint and rewards the inclusion of assets in the independent set. Explicitly, the cost Hamiltonian takes the form
\begin{align}
\hat{H}_{\mathrm{cost}}  =\ &  \sum_{i} \frac{z_i}{2} + \frac{1}{4} \sum_{(i,j) \in \mathcal{E}} \left(- z_i - z_j + z_i z_j \right)
\label{eq:hamiltonian_qaoa}
\end{align}
where the first term encourages the selection of assets, and the second term enforces the independence constraint by assigning a higher energy cost to pairs of adjacent assets that are both selected. The unitary that implements the desired dynamics in Eq.~\ref{eq:hamiltonian_qaoa} is given by $\mathcal{U}(\theta)=\Pi_{l=1}^{p}U_{\mathrm{mix}}(\beta_{l}) U_{\mathrm{cost}}(\gamma_{l})$, involving a series of $p$ layers alternating cost and mixing unitaries, $U_{\mathrm{cost}}(\gamma_{l})=\exp(-i\gamma_{l}\hat{H}_{\mathrm{cost}})$ and $U_{\mathrm{mix}}(\beta_{l})=\exp(-i\beta_{l}\hat{H}_{\mathrm{mix}})$, respectively, which are generated by $\hat{H}_{\mathrm{cost}}$ and $\hat{H}_{\mathrm{mix}}=\sum_{i}\hat{X}_{i}$, respectively. The parameters $\theta=(\boldsymbol{\gamma}, \boldsymbol{\beta})$ can either be optimized in an outer (closed) loop, inferred from parameter setting heuristics \citep{zhou:20, he:24, wurtz:21,hao2025end}, or pre-optimized \citep{augustino:24}. 

A recent work by \v{C}epait\.{e} \etal~\cite{vcepaite2025quantum} tackled the MIS problem with QAOA and reported an analytic format of setting QAOA parameters $\gamma, \beta$. Given a $p$-layer QAOA, the $j$-th parameter $\theta_j^p$ with $\theta \in \{\gamma, \beta\}$ is set as follows:
\begin{equation}
    \theta_j^p = c^\theta_1(j,p) + \frac{c^\theta_2(j,p)}{\braket{d}^{c^\theta_3(j,p)}+c^\theta_4(j,p)},
\label{eq:scaling_params}
\end{equation}
where $\{c^\theta_i(j,p)\}_{i=1,\ldots,4}$ are four fitted parameters and $\braket{d}$ is the average degree of the input graph. 

For the experiments we used these fitted parameters $\{c^\theta_i(j,p)\}_{i=1,\ldots,4}$. While in standalone QAOA, this scaling is performed only once for the input graph, in qReduMIS, this scaling is performed for the kernel $\mathcal{K}_{i}$ at each $i$-th iteration. Note that these kernels may have different degree distributions and thus the QAOA circuit may have different parameters. However, this is a simple calculation to scale the same given four fitted parameters. Hyperparameters of the semi-greedy qReduMIS implementation used in the benchmarks (in-set criterion, restricted-candidate-list size, frozen-node selection) are reported in Appendix~\ref{appendix:info_numerical_sim}.

\textbf{Backends.} We perform experiments of standalone QAOA and qReduMIS powered by QAOA on the Quantinuum H2-1 emulator with a statevector backend that closely models the noise of the real device. For real hardware, we perform our experiments on Quantinuum Helios with 98 trapped-ion qubits.

\textbf{Figures of merit.} The success probability $P_{\mathrm{MIS}}$ is defined as the fraction of independent trials that return a maximum independent set of the target graph. For QAOA, the success probability is the fraction of the quantum shots that return a maximum independent set. For qReduMIS, it is the fraction of the classical shots---each comprising one or more
iterative reduce-and-sample rounds---whose final output is a maximum independent set of the target graph.

Another figure of merit we study is the scaling of the \textit{time-to-solution} (TTS), defined as in Refs.~\cite{PhysRevX.8.031016,Kowalsky_2022,mohseni2022ising,leng2023quantum}:
\begin{equation} 
\mathrm{TTS}(t_f) = t_f R(t_f), \quad R(t_f) = \frac{\ln (1 - p_d) }{\ln (1 - p_S (t_f))}, 
\label{eq:TTS_tf_1} 
\end{equation}
where $t_f$ is the runtime for each run, $p_S(t_f)$ is the success probability of hitting the target solution given $t_f$, and $p_d$ is the desired probability that at least one run yields the target solution. Here we consider $p_d=0.99$. We adopt the optimal TTS, defined as~\cite{PhysRevX.8.031016}
\begin{equation} \mathrm{optTTS} = \min_{t_f} \mathrm{TTS}(t_f). \label{eq:def_opttts_1} \end{equation}
We quantify $t_f$ in terms of the cumulative number of quantum shots across all iterations. Although the wall-clock time per shot decreases as the kernel shrinks, this variation scales at most polynomially with the original problem size and does not affect the superpolynomial scaling analysis. We therefore adopt a uniform unit-cost model, treating each quantum shot as a single unit of time. Note that by doing this, the estimated optTTS for qReduMIS is an upper bound to the optTTS we will expect in practice.  

To compute optTTS for qReduMIS, we estimate the success probability of a single trial given a maximum runtime $t_f$ as 
\begin{equation}
\begin{split}
    p_S(t_f) &= \mathbb{E}_{t} \left[ \mathds{1} \left( \mathrm{runtime}_t \leq t_f \right) \right].
\end{split}
\label{eq:prob_qredumis}
\end{equation}

The runtime for the $t$-th trial of qReduMIS is
\begin{equation}
    \mathrm{runtime}_t = 
    \begin{cases} 
        D_t^\star \cdot \mathrm{qshots}, & \text{if optimal} \\
        \infty, & \text{otherwise}
    \end{cases}
\label{eq:runtime_qredumis}
\end{equation}
where $D_t^\star$ is the number of QPU calls required to reach the optimal solution, and $\mathrm{qshots}$ is the fixed number of quantum shots per iteration.

For standalone QAOA, the quantum circuit is fixed for a given problem instance, and each quantum shot can be regarded as an independent trial with a single iteration. Thus, the optTTS simplifies to:
\begin{equation}
\mathrm{optTTS} = \ln (1 - p_d) \min_{t_f} \frac{t_f}{\ln (1 - p_S (t_f))} = \frac{\ln (1 - p_d)}{\ln (1 - p_S)} 
\end{equation}
where $p_S$ is the success probability of a single QAOA quantum shot. For SA, we estimate the optTTS in the same way. 

We note that, in principle, the runtime and optTTS carry a wall-clock time
unit (e.g., seconds). However, wall-clock times depend on hardware and
implementation details that are not the focus of this work.
We therefore treat each quantum shot as a single unit of time, so that
runtime and optTTS are reported in units of quantum shots rather than
seconds. This convention makes the analysis hardware-agnostic and
emphasizes the algorithmic scaling behavior, which is the central object
of study. Full details on the runtime model, the quantum-shot schedule, and the empirical estimation of optTTS are provided in Appendix~\ref{appendix:empirical_tts}. The ground truth solutions for each of the problem instances considered were obtained with OR-Tools \cite{cpsatlp}.

\textbf{Postprocessing of independent set measurements.} Independent 
set violations in the QPU output are repaired via a postprocessing 
fixup applied identically to both standalone QAOA and the QAOA 
component of qReduMIS. By default, we use a simple removal heuristic; 
for the hardware experiments (Section~\ref{subsection:large_PSP}), we 
employ a stronger technique that also greedily adds non-conflicting 
nodes. This stronger
postprocessing is needed under hardware noise to give standalone QAOA a
fair chance of sampling the optimal solution; without it, no MIS is
recovered for all the four instances considered from the four major indices. The reinsertion step disproportionately
benefits standalone QAOA, since qReduMIS only requires marginal
distributions for frozen-node identification, so this choice provides a
favorable handicap to standalone QAOA in the comparison. This is further discussed in Appendix~\ref{appendix:numerical_experiments}. 

\section{Conclusion and Outlook}
\label{conclusion}

In this work, we have presented an end-to-end pipeline for quantum-informed portfolio selection, leveraging the hybrid qReduMIS framework \cite{schuetz2025qredumis} that combines classical (exact) reduction techniques with quantum optimization, to sample frozen nodes. By extending qReduMIS to gate-based QAOA on trapped-ion hardware, we solved PSP-MIS instances with up to 225 assets on Quantinuum's 98-qubit Helios system, with QAOA circuits acting on kernels of up to 78 qubits. This dual property---addressing input graphs whose direct encoding would exceed the device's qubit budget (as in the Nikkei~225 with $225$ assets vs.\ $98$ qubits), while also outperforming standalone QAOA on the instances both methods can encode---is, to our view, the most distinctive practical consequence of the reduce-and-sample paradigm. These problems correspond to asset universes of
major stock indices, demonstrating that the pipeline can operate at scales
relevant to practical portfolio construction. 

Our empirical benchmarks confirm the scalability and versatility of the proposed pipeline. We show that using the quantum computer in an inherently different way---to sample over frozen nodes rather than expecting it to output the optimal solution---successfully unblocks the next classical reduction, yielding robust performance across multiple figures of merit for PSP-MIS instances of varying sizes and hardness. qReduMIS consistently
outperforms standalone QAOA, both on hardware---where standalone QAOA
fails to find the optimal solution for the two largest indices---and on
the noisy emulator, where qReduMIS achieves a $3.2\times$ smaller optimal
time-to-solution scaling exponent at $p=2$. We further observe that
increasing the QAOA depth from $p=2$ to $p=6$ improves the performance
of both methods, an encouraging trend that suggests further gains are
within reach as advances in hardware fidelity progressively enable the
reliable execution of deeper QAOA circuits. Moreover, the qReduMIS framework is broadly applicable beyond structured market graphs, as demonstrated on unstructured 3-regular graphs.

We note that at the problem sizes accessible to current quantum hardware, classical MIS solvers---including simulated annealing---remain highly competitive, solving most instances
trivially and leaving insufficient variation for a meaningful scaling
analysis. Our comparison is therefore intentionally quantum-vs-quantum: it
demonstrates that leveraging QPU measurements to identify frozen
nodes---rather than expecting the quantum processor to directly output
optimal solutions---and iteratively combining this with exact (polynomial) classical
reductions constitutes a fundamentally more effective use of limited
quantum resources. qReduMIS
is a flexible framework poised to benefit from continued advances in both
classical and quantum components, including improvements in hardware,
algorithms, and software---an expectation already supported by the
performance gains observed when moving from $p=2$ to $p=6$ QAOA layers.

Looking ahead, the qReduMIS framework could be extended to alternative quantum subroutines and other combinatorial optimization problems. Additionally, designing an adaptive quantum-shot schedule could further improve the runtime scaling. As quantum hardware continues to mature, hybrid algorithms such as qReduMIS are well positioned to tackle large-scale optimization problems at scales where classical heuristics begin to struggle. \\

\section{Author contributions}
All authors participated in the development and discussion of the use case for the pipeline and the algorithmic development. RB, SS, and KP helped to frame and refine the use case. ZH contributed in setting up the experiments. MJAS devised the original project and algorithm idea. RY developed the implementation and performed the experiments, and was a major contributor in writing the manuscript.  RR contributed in deriving the risk bounds. JS contributed in developing the figures of merit. JS, MS, and RA contributed in analyzing the results. All authors contributed to the paper writing and approved the final manuscript.

\begin{acknowledgments}
We thank Ruslan Shaydulin, Pranav Deshpande, Akshay Seshadri, and Giuseppe
Di Cera for valuable discussions and support. 
\end{acknowledgments}

\section{Disclaimer}

This paper was prepared for informational purposes with contributions from the Global Technology Applied Research center of JPMorgan Chase \& Co and AWS Center for Quantum Computing. This paper is not a product of the Research Department of JPMorgan Chase \& Co., AWS Center for Quantum Computing, or any of their affiliates. Neither JPMorgan Chase \& Co. nor AWS Center for Quantum Computing nor any of their affiliates makes any explicit or implied representation or warranty and none of them accept any liability in connection with this paper, including, without limitation, with respect to the completeness, accuracy, or reliability of the information contained herein and the potential legal, compliance, tax, or accounting effects thereof. This document is not intended as investment research or investment advice, or as a recommendation, offer, or solicitation for the purchase or sale of any security, financial instrument, financial product or service, or to be used in any way for evaluating the merits of participating in any transaction.

\bibliography{bibliography} 

@article{markowitz1990foundations,
  author  = {Markowitz, H. M.},
  title   = {Foundations of portfolio theory},
  journal = {Harry Markowitz: Selected Works},
  series  = {S},
  pages   = {481--490},
  year    = {1990}
}

@book{elton2009modern,
  author    = {Elton, E. J. and Gruber, M. J. and Brown, S. J. and Goetzmann, W. N.},
  title     = {Modern portfolio theory and investment analysis},
  publisher = {John Wiley \& Sons},
  year      = {2009}
}

@article{fabozzi2008portfolio,
  author  = {Fabozzi, F. J. and Markowitz, H. M. and Gupta, F.},
  title   = {Portfolio selection},
  journal = {Handb. Financ.},
  volume  = {2},
  pages   = {3},
  year    = {2008}
}

@article{boginski:05,
  author  = {Boginski, V. and Butenko, S. and Pardalos, P. M.},
  title   = {Statistical analysis of financial networks},
  journal = {Comput. Stat. Data Anal.},
  volume  = {48},
  pages   = {431},
  year    = {2005}
}

@article{merton1972analytic,
  author  = {Merton, R. C.},
  title   = {An analytic derivation of the efficient portfolio frontier},
  journal = {J. Financ. Quant. Anal.},
  volume  = {7},
  pages   = {1851},
  year    = {1972}
}

@book{bouchaud2003theory,
  author    = {Bouchaud, J.-P. and Potters, M.},
  title     = {Theory of financial risk and derivative pricing: from statistical physics to risk management},
  publisher = {Cambridge University Press},
  year      = {2003}
}

@article{feo:94,
  author  = {Feo, T. A. and Resende, M. G. C. and Smith, S. H.},
  title   = {A greedy randomized adaptive search procedure for maximum independent set},
  journal = {Oper. Res.},
  volume  = {42},
  pages   = {860},
  year    = {1994},
  doi     = {10.1287/opre.42.5.860}
}

@article{gemsa:16,
  author  = {Gemsa, A. and N{\"o}llenburg, M. and Rutter, I.},
  title   = {Evaluation of labeling strategies for rotating maps},
  journal = {ACM J. Exp. Algorithmics},
  volume  = {21},
  year    = {2016},
  issn    = {1084-6654},
  doi     = {10.1145/2851493}
}

@article{hale:80,
  author  = {Hale, W. K.},
  title   = {Frequency assignment: Theory and applications},
  journal = {Proc. IEEE},
  volume  = {68},
  pages   = {1497},
  year    = {1980}
}

@misc{dong:22,
  author       = {Dong, Y. and Goldberg, A. V. and Noe, A. and Parotsidis, N. and Resende, M. G. C. and Spaen, Q.},
  title        = {A metaheuristic algorithm for large {Maximum Weight Independent Set} problems},
  year         = {2022},
  howpublished = {arXiv:2203.15805}
}

@misc{kalra:08,
  author       = {Kalra, A. and Qureshi, F. and Tisi, M.},
  title        = {Portfolio asset identification using graph algorithms on a quantum annealer},
  year         = {2018},
  howpublished = {SSRN},
  url          = {https://ssrn.com/abstract=3333537}
}

@article{hidaka2023correlation,
  author  = {Hidaka, R. and Hamakawa, Y. and Nakayama, J. and Tatsumura, K.},
  title   = {Correlation-diversified portfolio construction by finding maximum independent set in large-scale market graph},
  journal = {IEEE Access},
  year    = {2023}
}

@misc{cazals2025identifying,
  author       = {Cazals, P. and Fran\c{c}ois, A. and Henriet, L. and Leclerc, L. and Marin, M. and Naghmouchi, Y. and Coelho, W. d. S. and Sikora, F. and Vitale, V. and Watrigant, R. and others},
  title        = {Identifying hard native instances for the maximum independent set problem on neutral atoms quantum processors},
  year         = {2025},
  howpublished = {arXiv:2502.04291}
}

@article{he2026regularized,
  title={Regularized Warm-Started Quantum Approximate Optimization and Conditions for Surpassing Classical Solvers on the Max-Cut Problem},
  author={He, Zichang and Apte, Anuj and Augustino, Brandon and Babakhani, Arman and Khan, Abid and Omanakuttan, Sivaprasad and Shaydulin, Ruslan},
  journal={arXiv:2603.10191},
  year={2026}
}

@article{xiao2017exact,
  author  = {Xiao, M. and Nagamochi, H.},
  title   = {Exact algorithms for maximum independent set},
  journal = {Inf. Comput.},
  volume  = {255},
  pages   = {126},
  year    = {2017}
}

@article{kadowaki:98,
  author  = {Kadowaki, T. and Nishimori, H.},
  title   = {Quantum annealing in the transverse {Ising} model},
  journal = {Phys. Rev. E},
  volume  = {58},
  pages   = {5355},
  year    = {1998},
  doi     = {10.1103/PhysRevE.58.5355}
}

@misc{farhi:00,
  author       = {Farhi, E. and Goldstone, J. and Gutmann, S. and Sipser, M.},
  title        = {Quantum computation by adiabatic evolution},
  year         = {2000},
  howpublished = {arXiv:quant-ph/0001106}
}

@article{farhi:01,
  author  = {Farhi, E. and Goldstone, J. and Gutmann, S. and Lapan, J. and Lundgren, A. and Preda, D.},
  title   = {A quantum adiabatic evolution algorithm applied to random instances of an {NP}-complete problem},
  journal = {Science},
  volume  = {292},
  pages   = {472},
  year    = {2001},
  doi     = {10.1126/science.1057726}
}

@article{das:08,
  author  = {Das, A. and Chakrabarti, B. K.},
  title   = {Colloquium: Quantum annealing and analog quantum computation},
  journal = {Rev. Mod. Phys.},
  volume  = {80},
  pages   = {1061},
  year    = {2008},
  doi     = {10.1103/RevModPhys.80.1061}
}

@article{hauke:20,
  author  = {Hauke, P. and Katzgraber, H. G. and Lechner, W. and Nishimori, H. and Oliver, W.},
  title   = {Perspectives of quantum annealing: methods and implementations},
  journal = {Rep. Prog. Phys.},
  volume  = {83},
  pages   = {054401},
  year    = {2020}
}

@misc{farhi:14,
  author       = {Farhi, E. and Goldstone, J. and Gutmann, S.},
  title        = {A quantum approximate optimization algorithm},
  year         = {2014},
  howpublished = {arXiv:1411.4028}
}

@article{zhou:20,
  author  = {Zhou, L. and Wang, S.-T. and Choi, S. and Pichler, H. and Lukin, M. D.},
  title   = {Quantum approximate optimization algorithm: performance, mechanism, and implementation on near-term devices},
  journal = {Phys. Rev. X},
  volume  = {10},
  pages   = {021067},
  year    = {2020},
  doi     = {10.1103/PhysRevX.10.021067}
}

@article{quiroz2025quantifying,
  author  = {Quiroz, G. and Titum, P. and Lotshaw, P. and Lougovski, P. and Schultz, K. and Dumitrescu, E. and Hen, I.},
  title   = {Quantifying the impact of precision errors on quantum approximate optimization algorithms},
  journal = {Phys. Rev. Res.},
  volume  = {7},
  pages   = {023240},
  year    = {2025}
}

@article{mandl2024amplitude,
  author  = {Mandl, A. and Barzen, J. and Bechtold, M. and Leymann, F. and Wild, K.},
  title   = {Amplitude amplification-inspired {QAOA}: improving the success probability for solving {3SAT}},
  journal = {Quantum Sci. Technol.},
  volume  = {9},
  pages   = {015028},
  year    = {2024}
}

@article{ebadi:22,
  author  = {Ebadi, S. and Keesling, A. and Cain, M. and Wang, T. T. and Levine, H. and Bluvstein, D. and Semeghini, G. and Omran, A. and Liu, J.-G. and Samajdar, R. and others},
  title   = {Quantum optimization of {Maximum Independent Set} using {Rydberg} atom arrays},
  journal = {Science},
  volume  = {376},
  pages   = {1209},
  year    = {2022},
  doi     = {10.1126/science.abo6587}
}

@misc{schuetz2025qredumis,
  author       = {Schuetz, M. J. and Yalovetzky, R. and Andrist, R. S. and Salton, G. and Sun, Y. and Raymond, R. and Chakrabarti, S. and Acharya, A. and Shaydulin, R. and Pistoia, M. and others},
  title        = {{qReduMIS}: A quantum-informed reduction algorithm for the maximum independent set problem},
  year         = {2025},
  howpublished = {arXiv:2503.12551}
}

@misc{otterbach:17,
  author       = {Otterbach, J. S. and Manenti, R. and Alidoust, N. and Bestwick, A. and Block, M. and Bloom, B. and Caldwell, S. and Didier, N. and Fried, E. S. and Hong, S. and others},
  title        = {Unsupervised machine learning on a hybrid quantum computer},
  year         = {2017},
  howpublished = {arXiv:1712.05771}
}

@article{harrigan2021quantum,
  author  = {Harrigan, M. P. and Sung, K. J. and Neeley, M. and Satzinger, K. J. and Arute, F. and Arya, K. and Atalaya, J. and Bardin, J. C. and Barends, R. and Boixo, S. and others},
  title   = {Quantum approximate optimization of non-planar graph problems on a planar superconducting processor},
  journal = {Nat. Phys.},
  volume  = {17},
  pages   = {332},
  year    = {2021},
  doi     = {10.1038/s41567-020-01105-y}
}

@article{pagano:20,
  author  = {Pagano, G. and Bapat, A. and Becker, P. and Collins, K. S. and De, A. and Hess, P. W. and Kaplan, H. B. and Kyprianidis, A. and Tan, W. L. and Baldwin, C. and others},
  title   = {Quantum approximate optimization of the long-range {Ising} model with a trapped-ion quantum simulator},
  journal = {Proc. Natl. Acad. Sci. U.S.A.},
  volume  = {117},
  pages   = {25396},
  year    = {2020},
  doi     = {10.1073/pnas.2006373117}
}

@article{bravyi:20,
  author  = {Bravyi, S. and Kliesch, A. and Koenig, R. and Tang, E.},
  title   = {Obstacles to variational quantum optimization from symmetry protection},
  journal = {Phys. Rev. Lett.},
  volume  = {125},
  pages   = {260505},
  year    = {2020},
  doi     = {10.1103/PhysRevLett.125.260505}
}

@article{finvzgar2024quantum,
  author  = {Fin\v{z}gar, J. R. and Kerschbaumer, A. and Schuetz, M. J. and Mendl, C. B. and Katzgraber, H. G.},
  title   = {Quantum-informed recursive optimization algorithms},
  journal = {PRX Quantum},
  volume  = {5},
  pages   = {020327},
  year    = {2024}
}

@misc{acharya:24,
  author       = {Acharya, A. and Yalovetzky, R. and Minssen, P. and Chakrabarti, S. and Shaydulin, R. and Raymond, R. and Sun, Y. and Herman, D. and Andrist, R. S. and Salton, G. and others},
  title        = {Decomposition pipeline for large-scale portfolio optimization with applications to near-term quantum computing},
  year         = {2024},
  howpublished = {arXiv:2409.10301}
}

@misc{wybo2026scalable,
  author       = {Wybo, E. and R{\"o}nkk{\"o}, J. and Hirviniemi, O. and Fin\v{z}gar, J. R. and Leib, M.},
  title       = {A scalable quantum-enhanced greedy algorithm for maximum independent set problems},
  year         = {2026},
  howpublished = {arXiv:2601.21923}
}

@article{schuetz2025quantum,
  author  = {Schuetz, M. J. and Andrist, R. S. and Salton, G. and Yalovetzky, R. and Raymond, R. and Sun, Y. and Acharya, A. and Chakrabarti, S. and Pistoia, M. and Katzgraber, H. G.},
  title   = {Quantum compilation toolkit for {Rydberg} atom arrays with implications for problem hardness and quantum speedups},
  journal = {Phys. Rev. Res.},
  volume  = {7},
  pages   = {033107},
  year    = {2025}
}

@article{macMahon:15,
  author  = {MacMahon, M. and Garlaschelli, D.},
  title   = {Community detection for correlation matrices},
  journal = {Phys. Rev. X},
  volume  = {5},
  pages   = {021006},
  year    = {2015},
  doi     = {10.1103/PhysRevX.5.021006}
}

@misc{farhi2025lower,
  author       = {Farhi, E. and Gutmann, S. and Ranard, D. and Villalonga, B.},
  title        = {Lower bounding the {MaxCut} of high girth 3-regular graphs using the {QAOA}},
  year         = {2025},
  howpublished = {arXiv:2503.12789}
}

@book{garey:90,
  author    = {Garey, M. R. and Johnson, D. S.},
  title     = {Computers and intractability: A guide to the theory of {NP}-completeness},
  publisher = {W. H. Freeman \& Co.},
  address   = {New York, NY, USA},
  year      = {1990},
  isbn      = {0716710455}
}

@inproceedings{butenko:02,
  author    = {Butenko, S. and Pardalos, P. and Sergienko, I. and Shylo, V. and Stetsyuk, P.},
  title     = {{Branch-and-bound} algorithm for finding the maximum independent set},
  booktitle = {Proceedings of the 2002 ACM Symposium on Applied Computing},
  series    = {SAC '02},
  pages     = {542--546},
  year      = {2002},
  publisher = {Association for Computing Machinery},
  address   = {New York, NY, USA},
  doi       = {10.1145/508791.508897}
}

@incollection{strash:16,
  author    = {Strash, D.},
  title     = {On the power of simple reductions for the maximum independent set problem},
  booktitle = {Computing and Combinatorics},
  editor    = {Dinh, T. N. and Thai, M. T.},
  publisher = {Springer International Publishing},
  address   = {Cham},
  pages     = {345--356},
  year      = {2016},
  isbn      = {978-3-319-42634-1}
}

@article{bittel2021training,
  author  = {Bittel, L. and Kliesch, M.},
  title   = {Training variational quantum algorithms is {NP}-hard},
  journal = {Phys. Rev. Lett.},
  volume  = {127},
  pages   = {120502},
  year    = {2021}
}

@inproceedings{galda2021transferability,
  author    = {Galda, A. and Liu, X. and Lykov, D. and Alexeev, Y. and Safro, I.},
  title     = {Transferability of optimal {QAOA} parameters between random graphs},
  booktitle = {2021 IEEE International Conference on Quantum Computing and Engineering (QCE)},
  publisher = {IEEE},
  pages     = {171--180},
  year      = {2021}
}

@misc{brandao2018fixed,
  author       = {Brandao, F. G. and Broughton, M. and Farhi, E. and Gutmann, S. and Neven, H.},
  title        = {For fixed control parameters the quantum approximate optimization algorithm's objective function value concentrates for typical instances},
  year         = {2018},
  howpublished = {arXiv:1812.04170}
}

@misc{vcepaite2025quantum,
  author       = {{\v C}epait{\.e}, I. and Vaishnav, N. and Zhou, L. and Montanaro, A.},
  title        = {Quantum-enhanced optimization by warm starts},
  year         = {2025},
  howpublished = {arXiv:2508.16309}
}

@article{chang2000heuristics,
  author  = {Chang, T.-J. and Meade, N. and Beasley, J. E. and Sharaiha, Y. M.},
  title   = {Heuristics for cardinality constrained portfolio optimisation},
  journal = {Comput. Oper. Res.},
  volume  = {27},
  pages   = {1271},
  year    = {2000}
}

@misc{knazakis_portfolio,
  author       = {Knazakis, A.},
  title        = {{Portfolio Optimization}},
  year         = {2024},
  howpublished = {\url{https://github.com/AchilleasKn/PortfolioOptimization}},
  note         = {GitHub repository, accessed: 2026-04-07}
}

@article{PhysRevX.8.031016,
  author  = {Albash, T. and Lidar, D. A.},
  title   = {Demonstration of a scaling advantage for a quantum annealer over simulated annealing},
  journal = {Phys. Rev. X},
  volume  = {8},
  pages   = {031016},
  year    = {2018},
  doi     = {10.1103/PhysRevX.8.031016}
}

@article{Kowalsky_2022,
  author  = {Kowalsky, M. and Albash, T. and Hen, I. and Lidar, D. A.},
  title   = {3-regular three-{XORSAT} planted solutions benchmark of classical and quantum heuristic optimizers},
  journal = {Quantum Sci. Technol.},
  volume  = {7},
  pages   = {025008},
  year    = {2022},
  doi     = {10.1088/2058-9565/ac4d1b}
}

@article{mohseni2022ising,
  author  = {Mohseni, N. and McMahon, P. L. and Byrnes, T.},
  title   = {Ising machines as hardware solvers of combinatorial optimization problems},
  journal = {Nat. Rev. Phys.},
  volume  = {4},
  pages   = {363},
  year    = {2022}
}

@misc{leng2023quantum,
  author       = {Leng, J. and Hickman, E. and Li, J. and Wu, X.},
  title        = {Quantum {Hamiltonian} descent},
  year         = {2023},
  howpublished = {arXiv:2303.01471}
}

@misc{cpsatlp,
  author       = {Perron, L. and Didier, F.},
  title        = {{CP-SAT}},
  year         = {2026},
  howpublished = {\url{https://developers.google.com/optimization/cp/cp_solver/}},
  note         = {Accessed: 2026-04-07}
}

@misc{he:24,
  author       = {He, Z. and Shaydulin, R. and Herman, D. and Li, C. and Raymond, R. and Sureshbabu, S. H. and Pistoia, M.},
  title        = {Parameter setting heuristics make the quantum approximate optimization algorithm suitable for the early fault-tolerant era},
  year         = {2024},
  howpublished = {arXiv:2408.09538}
}

@article{wurtz:21,
  author  = {Wurtz, J. and Lykov, D.},
  title   = {Fixed-angle conjectures for the quantum approximate optimization algorithm on regular {MaxCut} graphs},
  journal = {Phys. Rev. A},
  volume  = {104},
  pages   = {052419},
  year    = {2021},
  doi     = {10.1103/PhysRevA.104.052419}
}

@misc{augustino:24,
  author       = {Augustino, B. and Cain, M. and Farhi, E. and Gupta, S. and Gutmann, S. and Ranard, D. and Tang, E. and Kirk, K. V.},
  title        = {Strategies for running the {QAOA} at hundreds of qubits},
  year         = {2024},
  howpublished = {arXiv:2410.03015}
}

@article{andrist:23,
  author  = {Andrist, R. S. and Schuetz, M. J. and Minssen, P. and Yalovetzky, R. and Chakrabarti, S. and Herman, D. and Kumar, N. and Salton, G. and Shaydulin, R. and Sun, Y. and others},
  title   = {Hardness of the maximum-independent-set problem on unit-disk graphs and prospects for quantum speedups},
  journal = {Phys. Rev. Res.},
  volume  = {5},
  pages   = {043277},
  year    = {2023},
  doi     = {10.1103/PhysRevResearch.5.043277}
}

@article{hoeffding1963probability,
  title={Probability inequalities for sums of bounded random variables},
  author={Hoeffding, Wassily},
  journal={Journal of the American statistical association},
  volume={58},
  number={301},
  pages={13--30},
  year={1963},
  publisher={Taylor \& Francis}
}

@article{he2025performance,
  title={Performance of quantum approximate optimization with quantum error detection},
  author={He, Zichang and Amaro, David and Shaydulin, Ruslan and Pistoia, Marco},
  journal={Communications Physics},
  volume={8},
  number={1},
  pages={217},
  year={2025},
  publisher={Nature Publishing Group UK London}
}

@article{perlin2026fault,
  title={Fault-tolerant execution of error-corrected quantum algorithms},
  author={Perlin, Michael A and He, Zichang and Armenakas, Anthony Alexiades and Andres-Martinez, Pablo and Hao, Tianyi and Herman, Dylan and Jin, Yuwei and Mayer, Karl and Self, Chris and Amaro, David and others},
  journal={arXiv:2603.04584},
  year={2026}
}

@article{omanakuttan2025threshold,
  title={Threshold for fault-tolerant quantum advantage with the quantum approximate optimization algorithm},
  author={Omanakuttan, Sivaprasad and He, Zichang and Zhang, Zhiwei and Hao, Tianyi and Babakhani, Arman and Boulebnane, Sami and Chakrabarti, Shouvanik and Herman, Dylan and Sullivan, Joseph and Perlin, Michael A and others},
  journal={arXiv:2504.01897},
  year={2025}
}

@article{jin2025iceberg,
  title={Iceberg beyond the tip: Co-compilation of a quantum error detection code and a quantum algorithm},
  author={Jin, Yuwei and He, Zichang and Hao, Tianyi and Omanakuttan, Sivaprasad and Amaro, David and Tannu, Swamit and Shaydulin, Ruslan and Pistoia, Marco},
  journal={arXiv:2504.21172},
  year={2025}
}

@article{hao2025end,
  title={End-to-end protocol for high-quality quantum approximate optimization algorithm parameters with few shots},
  author={Hao, Tianyi and He, Zichang and Shaydulin, Ruslan and Larson, Jeffrey and Pistoia, Marco},
  journal={Physical Review Research},
  volume={7},
  number={3},
  pages={033179},
  year={2025},
  publisher={APS}
}

@article{mandra2017exponentially,
  title={Exponentially biased ground-state sampling of quantum annealing machines with transverse-field driving Hamiltonians},
  author={Mandra, Salvatore and Zhu, Zheng and Katzgraber, Helmut G},
  journal={Physical review letters},
  volume={118},
  number={7},
  pages={070502},
  year={2017},
  publisher={APS}
}

@article{decross:23,
  title   = {Qubit-Reuse Compilation with Mid-Circuit Measurement and Reset},
  author  = {DeCross, Matthew and Chertkov, Eli and Kohagen, Megan and Foss-Feig, Michael},
  journal = {Phys. Rev. X},
  volume  = {13},
  issue   = {4},
  pages   = {041057},
  year    = {2023},
  month   = {Dec},
  publisher = {American Physical Society},
  doi     = {10.1103/PhysRevX.13.041057},
  url     = {https://link.aps.org/doi/10.1103/PhysRevX.13.041057}
}

@article{montanez:25,
  title   = {Towards a Linear-Ramp {QAOA} Protocol: Evidence of a Scaling Advantage in Solving Some Combinatorial Optimization Problems},
  author  = {Montanez-Barrera, J. A. and Michielsen, Kristel},
  journal = {npj Quantum Information},
  volume  = {11},
  pages   = {131},
  year    = {2025},
  doi     = {10.1038/s41534-025-01082-1},
  url     = {https://doi.org/10.1038/s41534-025-01082-1},
  note    = {arXiv:2405.09169}
}

@misc{dasgupta:26,
  title         = {Quantum Variational Approaches to the Maximum Independent Set Problem at Utility Scale},
  author        = {Dasgupta, Kalyan and Mukherjee, Sumanta and Verma, Dhriti and Sajja, Surya Shravan Kumar and Singh, Abhishek and Phan, Dzung and Kalagnanam, Jayant},
  year          = {2026},
  eprint        = {2606.28866},
  archivePrefix = {arXiv},
  primaryClass  = {quant-ph},
  url           = {https://arxiv.org/abs/2606.28866}
}

@article{sciorilli:25,
	author = {Sciorilli, Marco and Borges, Lucas and Patti, Taylor L. and Garc{\'\i}a-Mart{\'\i}n, Diego and Camilo, Giancarlo and Anandkumar, Anima and Aolita, Leandro},
	date = {2025/01/08},
	doi = {10.1038/s41467-024-55346-z},
	id = {Sciorilli2025},
	isbn = {2041-1723},
	journal = {Nature Communications},
	number = {1},
	pages = {476},
	title = {Towards large-scale quantum optimization solvers with few qubits},
	url = {https://doi.org/10.1038/s41467-024-55346-z},
	volume = {16},
	year = {2025}
}

@article{byun:22,
  title = {Finding the Maximum Independent Sets of Platonic Graphs Using Rydberg Atoms},
  author = {Byun, Andrew and Kim, Minhyuk and Ahn, Jaewook},
  journal = {PRX Quantum},
  volume = {3},
  issue = {3},
  pages = {030305},
  numpages = {10},
  year = {2022},
  month = {Jul},
  publisher = {American Physical Society},
  doi = {10.1103/PRXQuantum.3.030305},
  url = {https://link.aps.org/doi/10.1103/PRXQuantum.3.030305}
}

@article{ransford202698,
  title={A 98-qubit trapped-ion quantum computer with all-to-all connectivity},
  author={Ransford, Anthony and Allman, MS and Arkinstall, Jake and Campora, JP and Cooper, Samuel F and Delaney, Robert D and Dreiling, Joan M and Estey, Brian and Figgatt, Caroline and Hall, Alex and others},
  journal={Nature},
  pages={1--6},
  year={2026},
  publisher={Nature Publishing Group}
}

@article{brady2023iterative,
  title={Iterative Quantum Algorithms for Maximum Independent Set: A Tale of Low-Depth Quantum Algorithms},
  author={Brady, Lucas T. and Hadfield, Stuart},
  journal={arXiv preprint arXiv:2309.13110},
  year={2023}
}

\newpage

\appendix 

\section{Glossary \label{appendix:glossary}}

qReduMIS is a quantum-classical hybrid algorithmic framework and involves several levels of execution. In this paper, we use some terms to distinguish specific execution levels, explained in Table~\ref{tab:glossary}.

\begin{table}[htbp]
\centering
\caption{Glossary of terms for execution at different levels in qReduMIS.}
\label{tab:glossary}
\renewcommand{\arraystretch}{1.2}
\newcommand{\glosshead}[1]{\textbf{#1}}
\newcommand{\glossdef}[1]{\parbox[t]{0.66\linewidth}{\raggedright #1\strut}}
\begin{tabular}{ll}
\toprule
\glosshead{Term} & \glosshead{Definition} \\
\midrule
(classical) trial & \glossdef{An end-to-end iterative hybrid execution of a problem instance. Many trials are performed per instance.} \\
iteration         & \glossdef{One classical-reduction step followed by one QPU call. A trial comprises several iterations.} \\
QPU call          & \glossdef{The execution of a fixed quantum circuit; one QPU call per iteration.} \\
qshot             & \glossdef{A single quantum measurement; many qshots are performed within one QPU call.} \\
\bottomrule
\end{tabular}
\end{table}

\section{The Portfolio Selection Problem (PSP) from the financial perspective \label{appendix:PSP}}

\textbf{Correlation estimation.} One way of estimating the pair-wise correlation is by considering the daily returns of $n$ different assets over a period of $\mathcal{T}$ days, we define the matrix $X_{n \times \mathcal{T}} \in \mathbb{R}^{n \times  \mathcal{T} }$, where the $i$-th row corresponding to the $i$-th asset is a vector containing the daily returns of that asset across the period of $\mathcal{T}$ days. With it, the usual sample covariance estimator is $\Sigma = (1/\mathcal{T}) XX^T$, whose elements are $\sigma_{i,j}$. We consider a \textit{normalized correlation} between assets $i,j$, denoted as $c_{i,j}$, which is defined as 
\begin{equation}
c_{i,j} = \frac{\sigma_{i,j}}{\sqrt{\sigma_{i,i} \cdot \sigma_{j,j}}}.
\end{equation}
It is clear that $-1 \le c_{i,j} \le 1$ because $\Sigma$ is a positive semidefinite matrix. 
Denoting the \textit{correlation matrix} as $C$ whose elements are $c_{i,j}$, it holds that 
\begin{equation}
C = D^{-1/2} \Sigma D^{-1/2},
\end{equation}
where $ D \equiv \mbox{Diag}(\sigma_{1,1}, \sigma_{2,2}, \ldots, \sigma_{n,n})$. 

Having $C$, we can construct a graph $\mathcal{G}=(\mathcal{V}, \mathcal{E})$ with $n=|\mathcal{V}|$ nodes whose  edges $\mathcal{E}$ are based on an effective (thresholded) adjacency matrix $A$ where we set $A_{i,j}=1$ if and only if the absolute value of the correlation between assets $i$ and $j$ (given by $c_{i,j}$) is greater than the threshold parameter $\lambda$, and $A_{i,j}=0$ otherwise. The use of the absolute value, consistent with the problem formulation in the main text, reflects the role of this step as \emph{diversification}, not hedging: two strongly anticorrelated assets are not independent sources of return but mirror images of each other, and treating them symmetrically with strongly positively correlated pairs ensures that the selected basket carries genuinely independent information. This convention also does not affect the bounds in Theorem~\ref{theo:bound_risk}. Accordingly, with the threshold $\lambda$ acting as a correlation filter, pairs of assets are then classified as correlated only if the corresponding correlation coefficient exceeds a given (tunable) threshold $\lambda$, with all other correlations being discarded.
Although simple, this method is known to be robust to noise \citep{macMahon:15}, as the weakest correlations that are most prone to random fluctuations are removed.

\textbf{Risk bounds with PSP-MIS.} Two popular asset-allocation strategies---the uniform allocation and the inverse-variance allocation, both feasible asset allocations whose backtested effectiveness on real-world data is reported in Ref.~\cite{hidaka2023correlation}---admit explicit upper bounds on portfolio risk when restricted to a PSP-MIS solution. These bounds are stated in Theorem~\ref{theo:bound_risk}.

\begin{restatable}{theorem}{mainone}\label{theo:bound_risk}
    Let $S'$ be the solution of PSP which is an MIS of $\mathcal{G}(\mathcal{V}, \mathcal{E})$ whose edges are constructed with the aforementioned threshold $0 \le \lambda \le 1$ on the correlated assets. Let $\bm{w}_u$ be the uniform asset allocation such that $w_i = 1/m$ if $i \in S'$ and $w_i = 0$ otherwise, where $m \equiv |S'|$. Let $\sigma_{i,j}$ be the pair-wise covariance among the $i$ and $j$ assets. Then, it holds 
    $$
    \bm{w}_{u}^T \Sigma \bm{w}_u \le \frac{1-\lambda + \lambda m}{(m)^2}\sum_{i \in S'} \sigma_{i,i}. 
    $$
    Moreover, let $\bm{w}_{v}$ be the inverse-variance allocation asset such that $w_i = \frac{1/\sigma_{i,i}}{\sum_{i \in S'}1/\sigma_{i,i} }$ if $i \in S'$ and $w_i = 0$ otherwise. Then, it holds
    $$
    \bm{w}_{v}^T \Sigma \bm{w}_v \le \frac{1-\lambda + \lambda m}{\sum_{i \in S'} 1/\sigma_{i,i}}.
    $$
\end{restatable}

\begin{proof}
    It follows from the construction of the graph asset with threshold $\lambda$ and the fact that $S'$ is an independent set, and therefore $\sigma_{i,j} \le \lambda \sqrt{\sigma_{i,i} \sigma_{j,j}}$ holds for any $i,j \in S'$. Let us assume $\bm{w}_u$. In this case, 
    \begin{eqnarray*}
        \bm{w}_u^T \Sigma \bm{w}_u &=& \sum_{i,j \in S'} w_i w_j \sigma_{i,j} \\
        &=& \sum_{i\in S'} w_i^2 \sigma_{i,i} + \sum_{i\neq j \in S'} w_i w_j \sigma_{i,j} \\
        &\le& \frac{1}{(m)^2} \left( \sum_{i\in S'}\sigma_{i,i} +  \lambda \sum_{i\neq j \in S'} \sqrt{\sigma_{i,i} \sigma_{j,j}}   \right) \\
        &\le& \frac{1}{(m)^2} \left( \sum_{i\in S'}\sigma_{i,i} +  \lambda \sum_{i\neq j \in S'} \frac{\sigma_{i,i} + \sigma_{j,j}}{2}   \right) \\
        &=& \frac{1}{(m)^2} \left( \sum_{i\in S'}\sigma_{i,i} +  \lambda (m - 1)\sum_{i \in S'} \sigma_{i,i}   \right)\\
        &=& \frac{1-\lambda + \lambda m}{(m)^2}\sum_{i \in S'} \sigma_{i,i}, 
    \end{eqnarray*}
    where the first inequality follows from the construction of the market graph, and the second from the fact that the arithmetic mean is at least the geometric mean

    For the inverse variance allocation, we can similarly derive 
    \begin{eqnarray*}
        \bm{w}_v^T \Sigma \bm{w}_v &=&  
        \frac{1}{(\sum_{i\in S'} \frac{1}{\sigma_{i,i}})^2} \left( \sum_{i\in S'}\frac{1}{\sigma_{i,i}} +  \sum_{i\neq j \in S'} \frac{\sigma_{i,j}}{\sigma_{i,i} \sigma_{j,j}}   \right)\\
        &\le& \frac{1}{(\sum_{i\in S'} \frac{1}{\sigma_{i,i}})^2} \left( \sum_{i\in S'}\frac{1}{\sigma_{i,i}} +  \sum_{i\neq j \in S'} \frac{\lambda}{\sqrt{\sigma_{i,i} \sigma_{j,j}}}   \right) \\
        &\le& \frac{1}{(\sum_{i\in S'} \frac{1}{\sigma_{i,i}})^2}  \left( \sum_{i\in S'}\frac{1}{\sigma_{i,i}} +  \frac{\lambda}{2} \sum_{i\neq j \in S'} \frac{1}{\sigma_{i,i}} + \frac{1}{\sigma_{j,j}}  \right) \\
        &=& \frac{1-\lambda + \lambda m}{\sum_{i \in S'} 1/\sigma_{i,i}}. 
    \end{eqnarray*}
\end{proof}

As can be seen, the bounds imply that for sufficiently large MIS, the risk of such allocations is proportional to $\lambda$ times the average risk of assets in the MIS of graphs built with threshold $\lambda$. It is clear that the size of the MIS grows as $\lambda$ gets larger, i.e., the size of MIS is $m$ (the whole assets) when $\lambda = 1$, but so does the average risk of the assets in the MIS. Therefore, $\lambda$ acts as a parameter to balance the size of MIS with the average risk of assets in the MIS. 
%\ms{Can we possibly get rid of $N'$ in favor of $m$? Seems we use both here to denote the size of the MIS. Can we consolidate?}

\section{Optimal Time-to-Solution: Definitions and Empirical Estimation} \label{appendix:empirical_tts}

\textbf{TTS definition.} As generally adopted \cite{PhysRevX.8.031016,Kowalsky_2022,mohseni2022ising,leng2023quantum}, the TTS is defined as \begin{equation} \mathrm{TTS}(t_f) = t_f R(t_f), \quad R(t_f) = \frac{\ln (1 - p_d) }{\ln (1 - p_S (t_f))} \label{eq:TTS_tf_app} \end{equation} where $t_f$ is the runtime for each run, $p_S (t_f)$ is the success probability for the solver to hit the target solution given the runtime $t_f$, and $p_d$ is the desired probability for at least one run to obtain the target solution. The successes of different runs are assumed to be statistically independent, and $R(t_f)$ can be thought of as the effective number of runs needed to obtain at least one target solution. $R(t_f) \approx \ln (1/(1 - p_d)) / p_S(t_f) \propto 1/p_S(t_f)$ when $p_S (t_f)$ is small. We utilize a robust definition of TTS~\cite{PhysRevX.8.031016}, the optimal TTS, which is 
\begin{equation} 
\mathrm{optTTS} = \min_{t_f} \mathrm{TTS}(t_f). 
%\label{eq:def_opttts_app} 
\end{equation}

We comment that the definition of TTS often also conventionally contains a parallelization factor $1/f$ where $f$ is the number of parallel workers that can be deployed to work on independent runs of the same problem instance. In the framework of qReduMIS, if there are $f$ independent QPUs with each being able to run the complete circuit, the parallelization may also be utilized at each QPU call, which would similarly introduce a factor of $1/f$ for TTS. However, in this paper, we assume that no such parallelization is used, and thus $f=1$.

When the solver succeeds in all $M$ runs at runtime $t_f$, the
maximum-likelihood estimator $\hat p_S = 1$ yields $\ln(1-\hat p_S) =
-\infty$ and a meaningless $\mathrm{TTS} = 0$. Following
Ref.~\cite{PhysRevX.8.031016}, we replace the point
estimate by the Jeffreys upper bound
$\hat p_S^{\max} = 1 - 1/(2M)$, which is the smallest failure
probability statistically compatible with $M$ all-success observations.
The resulting $\mathrm{optTTS}$ is therefore an upper bound on
the true $\mathrm{optTTS}$ for such instances and can only make the solver appear
slower, never faster. For $p_d = 0.99$, this cap implies the floor
\begin{equation}
\mathrm{optTTS} \;\ge\; \frac{\ln(1-p_d)}{\ln\!\bigl(1/(2M)\bigr)}\, t_f^{\min}
              \;=\; \frac{2\ln 10}{\ln(2M)}\, t_f^{\min} \ge  \frac{2\ln 10}{\ln(2M)},
\label{eq:floor_opttts}
\end{equation}
where $t_f^{\min}$ is the shortest runtime at which $p_S(t_f) = 1$ is
observed. 

\textbf{Runtime model for qReduMIS.} For the $t$-th trial, the end-to-end run involves $D_t$ iterations with subroutine labels $C_{1,t}$, $Q_{1,t}$, $C_{2,t}$, $Q_{2,t}$, \ldots, $C_{D_t-1,t}$, $Q_{D_t-1,t}$, $C_{D_t,t}$, where $C_{i,t}$ denotes the $i$-th classical subroutine and $Q_{i,t}$ the $i$-th quantum subroutine. In terms of realistic wall-clock times, the cost of the classical subroutines is negligible compared to that of the quantum subroutines. For each $Q_{i,t}$, the number of quantum shots $\mathrm{qshots}_{i,t}$ is fixed, and we also neglect any overhead time associated with circuit compilation on the quantum device.
As the algorithm progresses, the kernel graph is progressively reduced in size, leading to fewer qubits. 
% and shallower quantum circuits in successive iterations \ms{Why shallower depth if (for example) $p$ is fixed?}
Consequently, the wall-clock time per quantum shot decreases over the course of the algorithm. This variation in per-shot cost is governed by the shrinking circuit size, which scales at most polynomially with the original problem size. Since our primary interest lies in characterizing the superpolynomial (e.g., exponential) scaling behavior of the algorithm, this polynomial variation does not affect the asymptotic analysis. We therefore adopt a uniform unit-cost model, treating each quantum shot as a single unit of time regardless of circuit size, and express $t_f$ accordingly as the total number of quantum shots.
The runtime for the $t$-th trial is given by 

\begin{equation}
    \mathrm{runtime}_t = 
    \begin{cases} 
        D_t^\star \cdot \mathrm{qshots}, & \text{if optimal solution is reached} \\
        \infty, & \text{if optimal solution not reached}
    \end{cases}
%\label{eq:runtime_qredumis_app}
\end{equation}
where $D_t^\star$ is the number of QPU calls required to reach the optimal solution. We consider each $C$ and $Q$ component of the algorithm as a block, and we inspect the local and incumbent solutions after each block execution, comparing against the ground truth solution. The optimality of the solution at the $i$-th iteration can be achieved by either the local solution $\mathcal{S}$ or the incumbent solution $\mathcal{W}$. These sets are constructed for the selected nodes across both the classical reduction and the quantum part. At each QPU call, qReduMIS identifies the largest independent set measured and also performs the frozen node(s) selection. While the local solution $\mathcal{S}$ is updated with the (in-set) frozen node(s), the incumbent $\mathcal{W}$ is updated with the largest solution measured by the QPU, if the updated incumbent improves the current one. By doing this, $\mathcal{W}$ is always ahead of $\mathcal{S}$ and the optimal number of QPU calls $D_t^\star$ corresponds to the incumbent finding the optimal solution. 

\textbf{Quantum-shot schedule.} In principle, the per-iteration shot budget $\mathrm{qshots}_{i,t}$ can be scheduled adaptively across iterations of a trial, e.g., guided by Hoeffding-type concentration bounds \cite{hoeffding1963probability}. Different schedules trade off total time cost against per-call statistical accuracy. The impact of the (constant) shot budget on optTTS is illustrated in Fig.~\ref{fig:quantum_shots_qredumis} for a representative instance, which motivates the choice $\mathrm{qshots}_{i,t}=\mathrm{qshots}$ used throughout this work.

\textbf{Empirical estimation of optTTS.} For a given problem instance, obtaining the optTTS from numerical experiments of multiple trials requires statistical estimation of $\mathrm{TTS}(t_f)$ for many values of $t_f$ and an optimization procedure with respect to $t_f$. This may seem computationally impractical, but we estimate it using a protocol that leverages the information from $M$ multiple runs or trials.
The success probability is estimated as 

\begin{equation}
\begin{split}
    p_S(t_f) &= \mathbb{E}_{t} \left[ \mathds{1} \left( \mathrm{runtime}_t \leq t_f \right) \right] \\
    &\approx \frac{\sum_t \mathds{1} \left( \mathrm{runtime}_t \leq t_f \right)}{\sum_t 1} = \frac{N_S(t_f)}{M}
\end{split}
\end{equation}
where $N_S(t_f)$ denotes the number of successful runs within runtime $t_f$.
Typically, the solver adopts an iterative algorithm that continues to improve the solution as $t_f$ increases, so we keep the solver running until the target solution is reached and record the required runtime $t_i$ for the run indexed by $i$. From $M$ runs, we obtain a sorted list of times to target: $t_1 \leq t_2 \leq \cdots \leq t_M$. Within a given runtime $t_f$, the runs with $t_i \leq t_f$ have reached the target, whereas the runs with $t_i > t_f$ have not. With increasing $t_f$, the estimated $p_S(t_f)$ changes its value only at $t_f = t_i$ ($i = 1, 2, \cdots, M$), so it suffices to evaluate $\mathrm{TTS}(t_f)$ only at $t_f = t_i$. We use the convention $N_S(t_i) = i - 1/2$, so $p_S(t_i) = (i - 1/2)/M$. The minimum of $\mathrm{TTS}(t_f)$ is then estimated from the $M$ numerical values of $p_S(t_i)$. The statistical error of optTTS is estimated by a bootstrap procedure of resampling the $M$ runs.

\textbf{OptTTS for standalone QAOA.} For standalone QAOA, the quantum circuit is fixed for a given problem instance, and each quantum shot can be regarded as an independent trial with a single iteration. QAOA can thus be viewed as an annealing-type algorithm with a fixed $t_f = t_{f0} = 1$ (our unit of measure). The optTTS simplifies to:
\begin{align*}
\mathrm{optTTS} &= \ln (1 - p_d) \min_{t_f} \frac{t_f}{\ln (1 - p_S (t_f))} \\
               &= \frac{\ln (1 - p_d)}{\ln (1 - p_S)},
\end{align*}
where $p_S$ is the success probability of a single QAOA quantum shot.

\section{Successfully identifying frozen nodes is easier than finding the optimal solution \label{appendix:runtime}}

We now present a theoretical analysis of the success probability of identifying frozen nodes and compare that with the success probability of finding an optimal solution. We focus on the in-set criterion but this can be extended to out-set criterion in a straightforward manner. An in-set (out-set) frozen node is a node identified as with high (low) probability to be part of a maximum independent set. While in qReduMIS this is done empirically from quantum measurements, we refer to the ground-truth set of in-set (out-set) frozen nodes as the unique set of nodes that belong to at least one (none) maximum independent set. In this section, we refer to them directly as set of frozen nodes.

\begin{restatable}{theorem}{mainthree}\label{theo:prob_frozen}

    Let $\mathcal{F}$ denote the set of (in-set) frozen nodes, which comprises the unique set of nodes that are part of at least one MIS solution for the graph $\mathcal{K}$, $\Pr(\mathcal{F})$ the probability of sampling $r=1$ (in-set) frozen nodes per draw (where $r$ denotes the number of candidate frozen nodes drawn from the restricted candidate list; we use $r$ here to avoid clashing with the correlation threshold $\lambda$ used elsewhere in the paper), and let $P_{\text{MIS}}$ be the probability of sampling the optimal solution. A lower bound for $\Pr(\mathcal{F})$ is the following
    \begin{align*}
        \Pr(\mathcal{F}) & > P_{\text{MIS}} + \sum_{i,n} P_{\text{MIS}_{i} \setminus \{n\}|\text{MIS-1}} \cdot P_{\text{MIS-1}|\text{non-MIS}}  \cdot  \\
        & (1-P_{\text{MIS}}),
    \end{align*}
    where $\text{MIS}-1$ refers to the solutions to the $\text{MIS}-1$ problem and we consider the non-maximal independent set of size $|\text{MIS}|-1$, referred as $\text{MIS}_{i}$. In general, a non-maximal independent set of size $k$ is an independent set that can be extended by adding $|\text{MIS}|-k$ vertices to form a maximum independent set. Specifically, non-maximal independent sets of size $k=|\text{MIS}|-1$ are constructed by removing any node $n$ from each possible $\text{MIS}_{i}$ solution to the MIS problem.

    Moreover, under the assumption of uniform sampling among solutions of the same size, each independent set of a given size is sampled with equal probability. Thus, for the $\text{MIS}-1$ problem, $P_{\text{MIS}_{i} \setminus \{n\}|\text{MIS-1}} = \frac{\mathcal{D}_{\text{MIS}} |\text{MIS}|}{\mathcal{D}_{\text{MIS}-1}}$, where $\mathcal{D}_{\text{MIS}}$ and $\mathcal{D}_{\text{MIS}-1}$ denote the number of (degenerated) maximum independent sets and (degenerated) independent sets of size $|\text{MIS}|-1$, respectively. Thus, in this case, the probability of sampling a frozen node is lower bounded by
    \begin{align}
    \Pr(\mathcal{F}) &> P_{\text{MIS}} + \frac{1}{\mathbb{H}} \cdot P_{\text{MIS-1}|\text{non-MIS}} \cdot  (1-P_{\text{MIS}}),
    \label{eq:prob_frozen_nodes}
    \end{align}
    where $\frac{1}{\mathbb{H}} := \frac{\mathcal{D}_{\text{MIS}} |\text{MIS}|}{\mathcal{D}_{\text{MIS}-1}}$, $\mathbb{H}$ is the hardness parameter introduced by Ebadi \etal{}~\cite{ebadi:22}.

    While the assumption of uniform sampling is highly idealized as quantum annealers, for instance, are known to be biased samplers~\cite{schuetz2025quantum, mandra2017exponentially}, it yields a closed-form, instance-defined value for the coefficient $1/\mathbb{H}$. We remark that the structural bound $\Pr(\mathcal{F}) > P_{\text{MIS}}$ itself does not rely on this assumption.
\end{restatable}

\begin{proof}
    
Given a graph $\mathcal{G}$, we define the set of MIS solutions as $\mathcal{M}$ and the set of (in-set) frozen nodes as $\mathcal{F}^{\text{in}}$, which contains the unique set of nodes that belong to all the $\mathcal{D}$ (degenerated) MIS solutions. For ease of notation, we will refer as $\mathcal{F}$. In general, given an independent set measurements $\{\mathcal{I}_{i}\}$, the probability of sampling a (any) frozen node from them is

\begin{align*}
\Pr(\mathcal{F}) = \Pr(\mathcal{F} \mid \text{MIS}) \cdot P_{\text{MIS}} + \Pr(\mathcal{F} \mid \text{non-MIS}) \cdot (1 - P_{\text{MIS}}) 
\end{align*}
where we use the symbol $P$ to refer to the probability of sampling an independent set measurement over a (particular) set of independent set measurements and $\Pr$ to refer to the probability of sampling a frozen node (in $\mathcal{F}$) from a set of independent set measurements. By definition $ \Pr(\mathcal{F} \mid \text{MIS})=1$ and the probability of sampling a frozen node given a non-MIS solution is different from zero. Thus, in the cases where $ P_{\text{MIS}}  < 1$ we can show that

\begin{align*}
\Pr(\mathcal{F}) = P_{\text{MIS}} + \Pr(\mathcal{F} \mid \text{non-MIS}) \cdot (1 - P_{\text{MIS}}) > P_{\text{MIS}}
\end{align*}

To prove this, we show that $\Pr(\mathcal{F} \mid \text{non-MIS}) > 0$. Note that given a solution to the MIS problem, if we remove $i$ nodes from it, this is a solution to the $\text{MIS}-i$ problem. For example, if a set $\mathcal{I}$ is a solution to the MIS problem, then if we remove any node from it $\mathcal{I}\setminus \{n\}$ $\forall$ $n \in \mathcal{I}$, this is a solution to the $\text{MIS}-1$ problem. 

We start by noting that 

\begin{align*}
\Pr(\mathcal{F} \mid \text{non-MIS}) = \sum_{i=1}^{|\text{MIS}|} \Pr(\mathcal{F} \mid \text{MIS-i}) \cdot  P_{\text{MIS-i}|\text{non-MIS}}  \\ 
+ \Pr(\mathcal{F} \mid \text{>MIS}) \cdot  P_{\text{>MIS}|\text{non-MIS}}
\end{align*}
where $P_{\text{MIS-i}|\text{non-MIS}}$ refers to the probability that among all the non-MIS solutions we sample a solution to the MIS-i problem and $ \Pr(\mathcal{F} \mid \text{MIS-i})$ is the probability that given MIS-i solution, we sample a frozen node. The use of $\text{>MIS}$ indicates independent set measurements whose size (i.e., Hamming weight) is larger than the MIS. $P_{\text{>MIS}|\text{non-MIS}}$ refers to the probability of given a non-MIS solution, sample an independent set measurement that exceeds MIS. We take this $P_{\text{>MIS}|\text{non-MIS}}=0$ as ideally we would not expect to measure from the QPU a solution that is not an independent set. It is common practice to postprocess the QPU measurements to remove solutions that break the independence constraint. By splitting that sum into the case of $i=1$ and $i>1$ we obtain 

\begin{align*}
\Pr(\mathcal{F}) = P_{\text{MIS}} + \Pr(\mathcal{F} \mid \text{MIS-1}) \cdot P_{\text{MIS-1}|\text{non-MIS}} \cdot  (1-P_{\text{MIS}})  \\
+ \sum_{i>1}^{|\text{MIS}|} \Pr(\mathcal{F} \mid \text{MIS-i}) \cdot  P_{\text{MIS-i}|\text{non-MIS}} \cdot (1-P_{\text{MIS}})
\end{align*}

How much larger the probability of sampling a frozen node is, in comparison to the probability of sampling the optimal solution, is mostly controlled by the contributions of the probability of sampling each $\text{MIS}-i$, $i=2, \dots, N-1$ independent set measurement that is made by removing nodes from the MIS solutions. 

Next, let us derive an expression for $\Pr(\mathcal{F} \mid \text{MIS-1})$. To do this, we note that the set of $\mathcal{D}_{\text{MIS}-1}$ solutions to the $\text{MIS}-1$ problem consists of two sets: (1) those obtained by removing a single node from a MIS solution, and (2) other independent sets that are not derived in this way. The first type can be described as the sets $\text{MIS} \setminus \{n\}$, $\forall$ $n \in \text{MIS}$, while we refer to the second type collectively as \textit{other}. 

Given a measurement that yields an independent set, the probability of sampling an element from the first type is denoted by $P_{\text{MIS} \setminus \{n\}|\text{MIS-1}}$, and the probability of sampling an element from the second type is denoted by $P_{\text{other}|\text{MIS-1}}$. Using these definitions, we can show that
\begin{align*}
\Pr(\mathcal{F} \mid \text{MIS-1}) = \Pr(\mathcal{F} \mid \text{MIS} \setminus \{n\}) \cdot P_{\text{MIS} \setminus \{n\}|\text{MIS-1}} \\ 
+ \Pr(\mathcal{F} \mid \text{other} ) \cdot  P_{\text{other}|\text{MIS-1}} \\
\end{align*}
where by definition $\Pr(\mathcal{F} \mid \text{MIS} \setminus \{n\})=1$. 

We obtain a lower bound for $\Pr(\mathcal{F})$ by lower bounding this expression by removing the contributions of the $\text{other}$ states (these are discarded rather than estimated, which strengthens the bound on $\Pr(\mathcal{F})$), and we consider the contribution of each of the $i$-th MIS solutions and each of the $n$ nodes as
\begin{align*}
P_{\text{MIS} \setminus \{n\}|\text{MIS-1}} = \sum_{i,n}  P_{\text{MIS}_{i} \setminus \{n\}|\text{MIS-1}}. 
\end{align*}

Then, 

\begin{align*}
\Pr(\mathcal{F}) \;>\;& P_{\text{MIS}} + \\
&\sum_{i,n} P_{\text{MIS}_{i} \setminus \{n\} \mid \text{MIS-1}} \cdot P_{\text{MIS-1} \mid \text{non-MIS}} \cdot (1 - P_{\text{MIS}}).
\end{align*}

Under the assumption of uniform sampling across independent sets of the same size, the probability $P_{\text{MIS} \setminus  \{n\}|\text{MIS-1}}$ is given by the problem instance and it is the proportion of $\text{MIS}-1$ solutions that arise from all possible MIS solutions by deleting a single node from each. It has the following form, for all $i$ and $n$

\begin{align*}
P_{\text{MIS}_{i} \setminus  \{n\}|\text{MIS-1}}  = \frac{\mathcal{D}_{\text{MIS}} \binom{|\text{MIS}|}{1}}{\mathcal{D}_{\text{MIS}-1}} = \frac{\mathcal{D}_{\text{MIS}} |\text{MIS}|}{\mathcal{D}_{\text{MIS}-1}}.
\end{align*}

Thus, in the case of $P_{\text{MIS}} < 1$, how much larger is the probability of sampling a (in-set) frozen node to the probability of sampling an optimal solution is driven by the second term, which is weighted by two terms. The first one is given by the problem instance and it is the fraction of solutions to $\text{MIS}-1$ that are constructed from $\text{MIS}$ solution. Intuitively, the closer this number to 1, the better because it means that all the solutions to $\text{MIS}-1$ are made of (in-set) frozen nodes. Thus, the larger the probability $P_{\text{MIS-1}|\text{non-MIS}}$, the larger will be this contribution. Note that  $P_{\text{MIS-1}|\text{non-MIS}}$ is driven by the quantum algorithm of choice and in practice, it depends on the hardware. 

Note that for ease of reading, we have lower bounded $\Pr(\mathcal{F}) $ and only showed the contributions to this probability with the solutions to the $\text{MIS}-1$. However, as shown before, in the sum there are contributions from all the  $\text{MIS}-i$ solutions, weighted by the probability to sampling these solutions. More generally, using that $\sum_{i=1}^{|\text{MIS}|} \frac{\mathcal{D}_{\text{MIS}} \binom{|MIS|}{i}}{\mathcal{D}_{\text{MIS}-i}} >\frac{\mathcal{D}_{\text{MIS}} |\text{MIS}|}{\mathcal{D}_{\text{MIS}-1}}$, we can write

\begin{align}
\Pr(\mathcal{F}) &> P_{\text{MIS}} + \frac{\mathcal{D}_{\text{MIS}} |\text{MIS}|}{\mathcal{D}_{\text{MIS}-1}} \cdot P_{\text{MIS-1}|\text{non-MIS}} \cdot  (1-P_{\text{MIS}}) \\ 
&= P_{\text{MIS}} + \frac{1}{\mathbb{H}} \cdot P_{\text{MIS-1}|\text{non-MIS}} \cdot  (1-P_{\text{MIS}}).
\end{align}

Note that the coefficient $\frac{\mathcal{D}_{\text{MIS}} |\text{MIS}|}{\mathcal{D}_{\text{MIS}-1}}$ is $\frac{1}{\mathbb{H}}$, where $\mathbb{H}$ is the hardness parameter introduced by Ebadi \etal{}~\cite{ebadi:22}.

\end{proof}

We have lower bounded the $\Pr(\mathcal{F})$ by considering non-maximal independent sets of different sizes. Among these, we focus on those of size $|\text{MIS}|-1$ as these are arguably the states with the highest sampling likelihood after the optimal solutions, when using quantum algorithms intended to find the optimal solution. This leads to an expression where the probability of sampling a frozen node, relative to sampling the optimal solution, is governed by the contribution from non-maximal independent sets of size $|\text{MIS}|-1$ and the fraction of measurements of size $|\text{MIS}|-1$ among all non-MIS measurements. 

In the case of uniform sampling of measurements of same size, we see that this is controlled by the hardness of the problem instance itself. In both SA and QA, the sampling process can be viewed as generating configurations according to a probability distribution that, at equilibrium or in the low-temperature limit, becomes approximately uniform over all independent sets of a fixed size. It has been empirically reported for these algorithms that the probability of sampling a maximum independent set, $P_{\text{MIS}}$, decreases exponentially as the hardness parameter $\mathbb{H}$ increases~\cite{ebadi:22, schuetz2025qredumis}. 

This analysis highlights the effectiveness of using the QPU as a co-processor to sample frozen nodes. Even if the QPU output rarely includes measurements corresponding to maximum independent sets, the ability to sample non-maximal independent sets of size $|\text{MIS}|-1$ ensures that $\Pr(\mathcal{F})$ may remain sufficiently large to provide meaningful information for classical reduction in subsequent iterations. This enables qReduMIS to ultimately identify an optimal solution. Furthermore, since $\Pr(\mathcal{F}) > P_{\text{MIS}}$, the requirements for high-quality measurements from the QPU are relaxed; that is, it is not necessary to obtain extensive statistics for the optimal solution. For instance, in the context of QAOA, this means that robust parameter choices for $\gamma$ and $\beta$ do not need to be perfectly optimized to sample the optimal solution. As a result, the qReduMIS approach reduces the demands on quantum and classical resources required for successful problem solving.

\section{Experimental results \label{appendix:numerical_experiments}}

\textbf{Postprocessing of independent set measurements.} There are multiple techniques, ranging from simple heuristics such as removing the first violating node to more sophisticated approaches like greedily adding non-conflicting nodes after removal. The choice of strategy significantly affects both the size of the recovered IS and the resulting success probability. While the performance of QAOA directly depends on this strategy as it determines the end-to-end success probability of the algorithm, in qReduMIS the end-to-end algorithm's success probability depends on the success probability of sampling correct frozen nodes. As established by Theorem~\ref{theo:prob_frozen}, under the stated assumptions the probability of sampling a frozen node from a given set of independent-set measurements is strictly larger than the probability of sampling a maximum independent set. Thus, by design, the performance of qReduMIS is less susceptible to the postprocessing fixup. For this reason, we opt for the most simple technique for postprocessing, if not said otherwise. In particular, for the hardware experiments, in order to leverage the hardware for larger and deeper circuits, we opt for the most sophisticated one (remove and greedily add with check) with the goal of preventing null success probability in QAOA.

\subsection{Large-scale quantum-informed PSP-MIS on hardware \label{appendix:mkt_graphs_hw}}

\begin{table*}
    \centering
    \begin{tabular}{l|cc|cc|cc}
        \hline
        \multirow{2}{*}{Stock Market} 
            & \multicolumn{2}{c|}{Success Probability} 
            & \multicolumn{2}{c|}{optTTS}
            & \multicolumn{2}{c}{Average Approximation Ratio} \\
        \cline{2-7}
            & QAOA & qReduMIS & QAOA & qReduMIS & QAOA & qReduMIS \\
        \hline
        DAX 100 
            & 0.09\,[0.04, 0.15] & 0.95\,[0.85, 1.00] 
            & 48.83\,[28.34, 112.81] &35.56\,[24.97, 52.84] 
            & 0.83\,[0.81, 0.85] & 1.00\,[0.99, 1.00] \\
        FTSE 100 
            & 0.03\,[0.00, 0.07] & 0.70\,[0.50, 0.90]
            & 151.19\,[63.46, 458.21] & 185.58\,[88.58, 286.41]
            & 0.73\,[0.71, 0.75] & 0.98\,[0.96, 0.99] \\
        S\&P 100 
            & 0.00 &  0.40\,[0.20, 0.60] 
            & $\infty$ & 468.67\,[215.28, 1181.40]
            & 0.74\,[0.73, 0.76] & 0.96\,[0.94, 0.97] \\
        Nikkei 225 
            & 0.00 & 0.95\,[0.85, 1.00]
            & $\infty$ & 88.58\,[49.94, 123.49]
            & 0.80\,[0.80, 0.82] & 0.99\,[0.98, 1.00] \\
        \hline
    \end{tabular}
    \caption{Performance comparison of standalone QAOA and qReduMIS on the first kernel $\mathcal{K}$ for each stock market index. The table reports the success probability ($P_{\text{MIS}}$), optimal time-to-solution (optTTS), and expected approximation ratio (AR), each with 95\% confidence intervals obtained via bootstrap resampling with 10,000 replacements. For standalone QAOA, 100 quantum shots were used; for qReduMIS, 20 classical trials with 10 quantum shots per QPU call. All experiments use $p=2$ QAOA layers on Quantinuum's Helios hardware (98 qubits).}
    \label{tab:performance_results}
\end{table*}

We execute QAOA (with $p=2$) over the kernel $\mathcal{K}$. We postprocessed the $100$ independent set measurements and from them, we compute the figures of merit for standalone QAOA. For qReduMIS, we perform 20 classical trials ($\text{cshots}=20$) for each of the indices. At each iteration of qReduMIS, we performed the classical reduction on a local classical computer, and we submitted the QAOA job (with 10 quantum shots) for the updated kernel $\tilde{\mathcal{K}}$. The independent set measurements are then postprocessed in the same way as those of standalone QAOA, and an in-set frozen node is selected, leading to an updated kernel for the next iteration. 

Table~\ref{tab:performance_results} shows the three figures of merit for the different indexes. These are the values reflected in Fig.~\ref{fig:main_results_hw} in main text with the addition of the optTTS. For vanilla QAOA the S\&P~100 and Nikkei~225 indices yield $P_{\text{MIS}} = 0$, and consequently $\text{optTTS} = \infty$. For the other two indices, the improvement in success probability achieved by qReduMIS is not directly reflected in the $\text{optTTS}$, since the latter also depends on the number of quantum shots performed at each QPU call, which may have been too large for these instances. Reducing the number of quantum shots, provided the end-to-end success probability stays close to the values achieved here, would lower the $\text{optTTS}$ and could reveal an improvement over standalone QAOA. As we discussed in the main text, cross-method comparisons of $\text{optTTS}$ on the same instances are not robust and for this reason in the main text we limit to discuss the other two figures of merit. 

For all the indices no more than five QPU calls (and five classical reductions) were performed. We kicked off 20 classical trials for each of the four kernels corresponding to the indices and after each QPU call and classical reduction, some of these trials finished because the updated kernel became fully reducible. Note that this does not mean that the optimal solution was found. There are some cases where the kernel graphs were fully reducible but the final solution was sub-optimal. This is because at some point in the iterations, there was an incorrect selection of a frozen node, which unblocked the reduction, but was a node that did not belong to any MIS. To prevent running more experiments on the hardware when it was not necessary, we stopped the iterations once the ground truth solution, computed with OR-Tools, was found, as we report the optTTS.

\subsection{Extensive benchmark on PSP-MIS instances \label{appendix:benchmark_extensive}}

This appendix reports two complementary emulator experiments on Quantinuum's H2-1 emulator that use different quantum-shot budgets and should not be confused: (i) the kernel-size scan used to estimate the optTTS scaling (Figs.~\ref{fig:probability_mis_mkt_graphs} and \ref{fig:qpu_to_solution_distribution_emulator}), in which qReduMIS uses $\mathrm{qshots}=5$ per QPU call, and (ii) the hardness-parameter scan at fixed $N=22$ (Fig.~\ref{fig:metrics_hardness}), in which qReduMIS uses $\mathrm{qshots}=10$ per QPU call. The hardware experiments on Helios (Appendix~\ref{appendix:mkt_graphs_hw}) also use $\mathrm{qshots}=10$.

\begin{table}[h!]
\centering
\begin{tabular}{|c|c|}
\hline
\textbf{Interval} & \textbf{Number of Seeds} \\
\hline
4--5   & 11 \\
6--9   & 9  \\
10--13 & 9  \\
14--15 & 7  \\
16--16 & 8  \\
17--18 & 9  \\
19--20 & 8  \\
21--21 & 7  \\
22--22 & 5  \\
\hline
\end{tabular}
\caption{Number of seeds per interval of (first) kernel size utilized for the benchmark on Quantinuum's emulator.}
\label{tab:num_seeds_mkt_emulator}
\end{table}

\textbf{The testbed.} For each problem size $N$ (number of assets), we generate $20$ random instances by sampling $N$ assets without replacement from the NIKKEI 225 index, using the data from Chang \etal~\cite{chang2000heuristics}. For every instance we set the correlation sensitivity as in Eq.~\ref{eq:average_corr} to obtain the graph $\mathcal{G}$. The size of the (first) kernel is not a controllable parameter: it is fully determined by the input instance. The reduction technique, described in the main text, is based on the exposed-node removal. Thus, the number of nodes it eliminates depends on the density of the input graph and on the distribution of the exposed corner nodes. As a result, for a fixed number of nodes in the input graph, the size and structure of the (first) kernel graph vary from instance to instance. Fig.~\ref{fig:market_testbed} shows the resulting dispersion of the size of the first kernel graph $\mathcal{K}$ as a function of the number of nodes in the input graph $\mathcal{G}$.

\begin{figure}
    \centering
    \includegraphics[width=\linewidth]{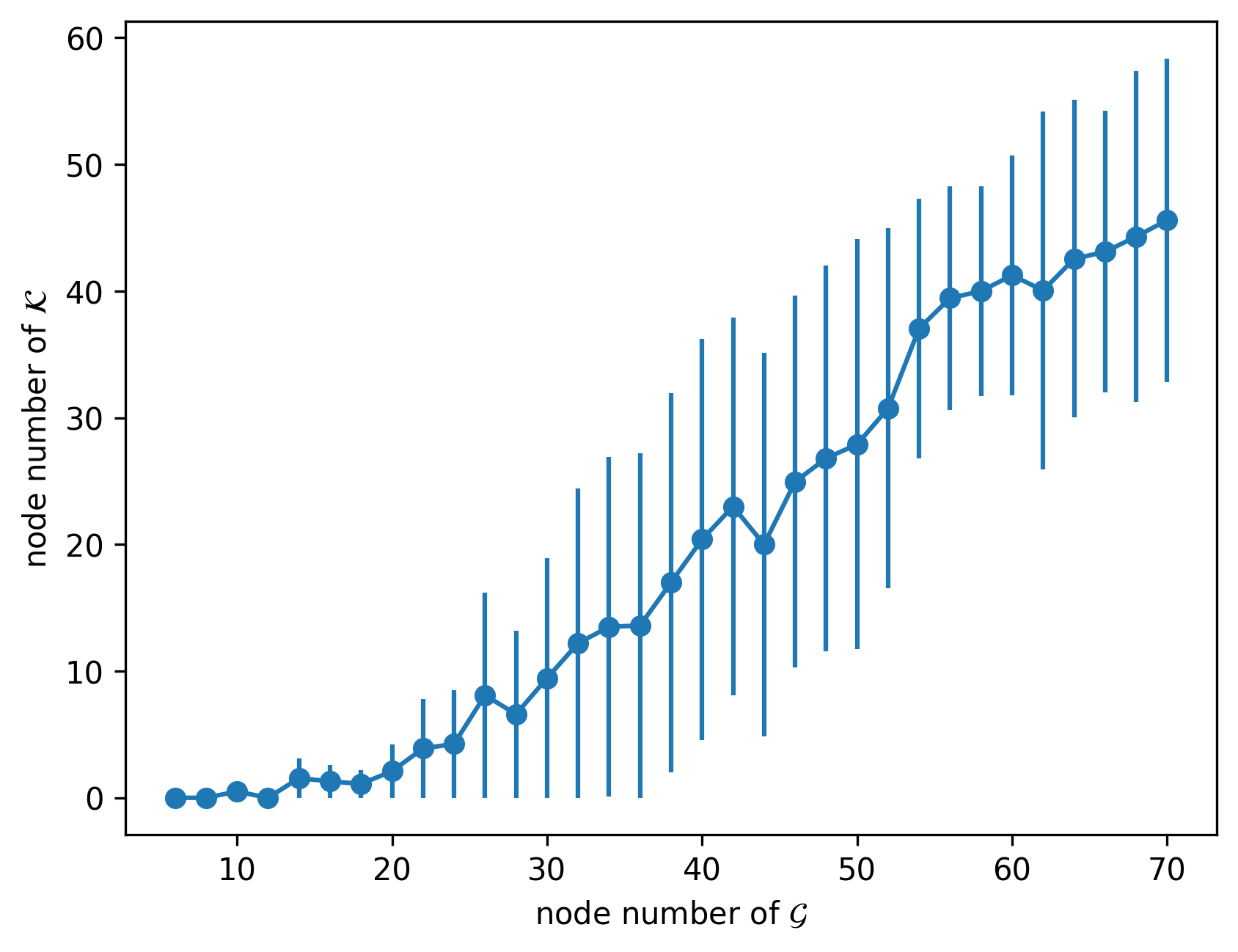}
    \caption{Size of the first kernel graphs $\mathcal{K}$ as a function of the size of the input graphs $\mathcal{G}$ across $20$ random instances. To make the market graphs, the correlation sensitivity used was the average correlation among the set of assets.}
    \label{fig:market_testbed}
\end{figure}

We limit to input graphs whose (first) kernel sizes have up to $22$ nodes. Considering the dispersion in kernel sizes, to select the problem instances from the ones shown in Fig.~\ref{fig:market_testbed}, we consider intervals of kernel sizes such that they have a similar number of problem instances per interval. Fig.~\ref{fig:testbed_emulator} shows the instances considered for the extensive benchmark: the top panel reports the distribution of kernel sizes per interval, and the bottom panel the distribution of kernel density. Both the input-graph and kernel-size intervals are visible in the figure. We consolidate the number of seeds per interval of (first) kernel size in Table~\ref{tab:num_seeds_mkt_emulator}.

\begin{figure*}
    \centering
    \includegraphics[width=0.7\linewidth]{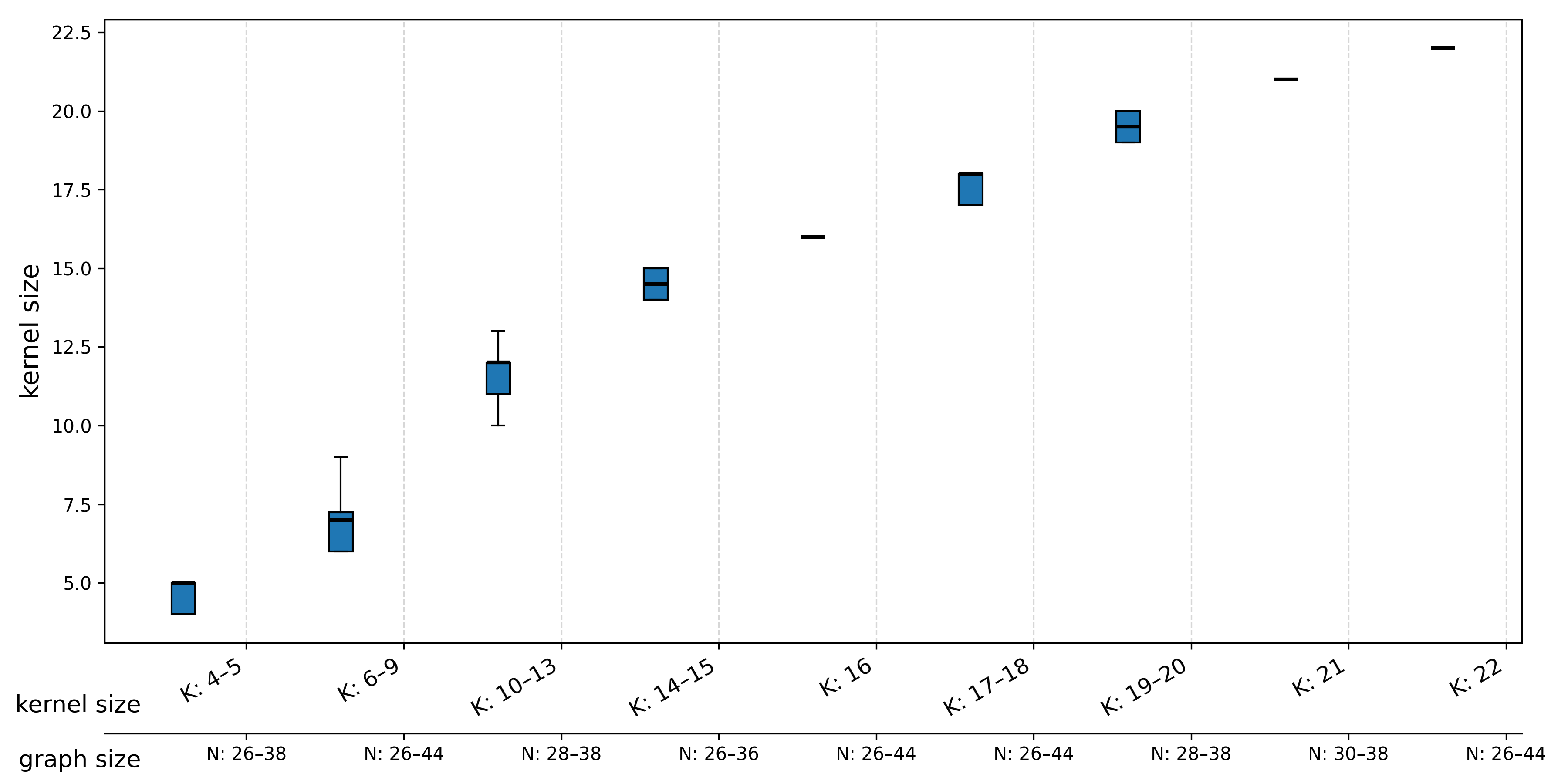}
    \includegraphics[width=0.7\linewidth]{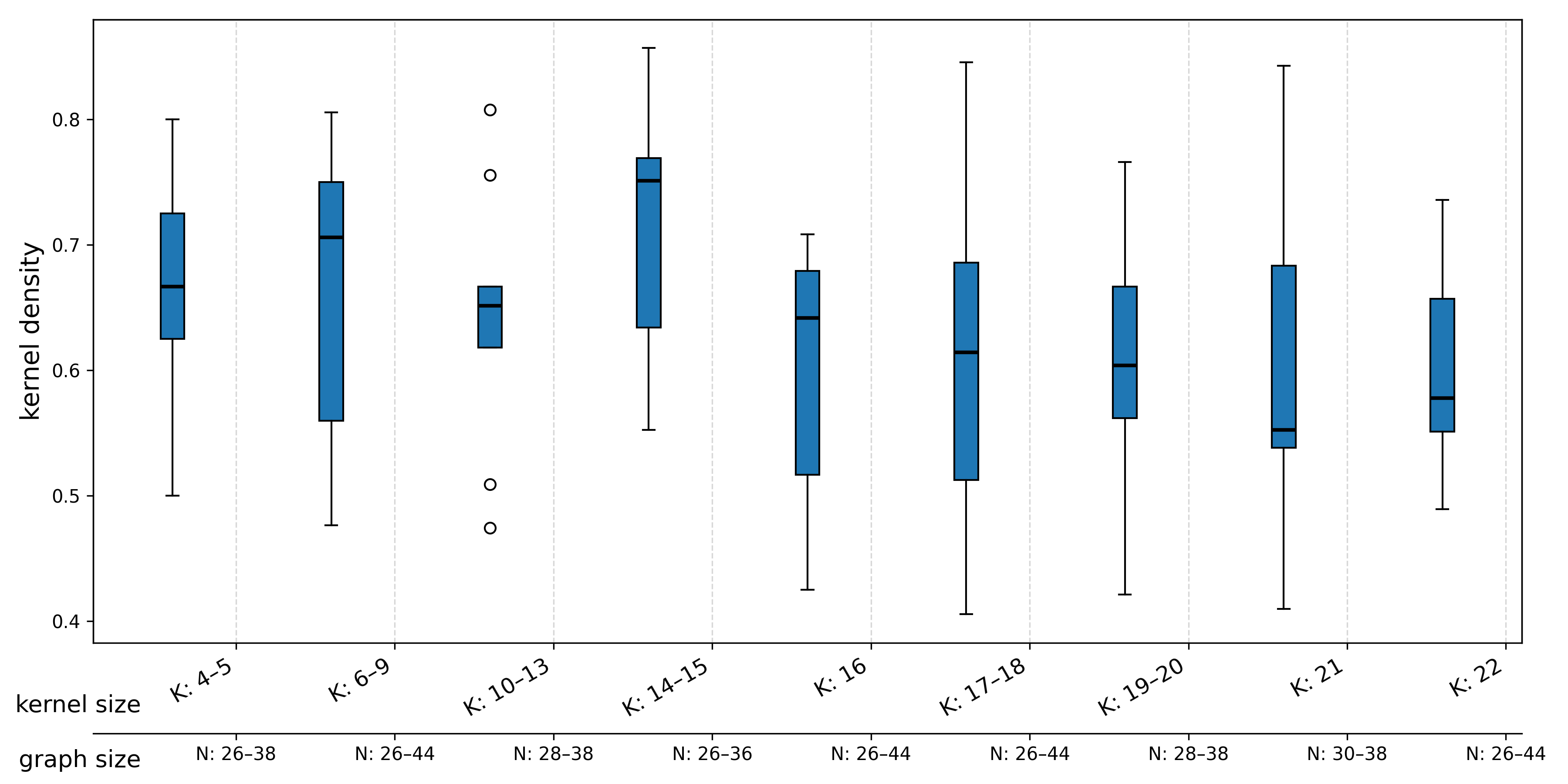}
    \caption{Characteristics of the (first) kernel graphs in the testbed. \textbf{Top:} The distribution of the (first) kernel size across the intervals in the testbed. \textbf{Lower:} the density distribution of the kernels. The boxplot of these distributions is shown as a function of the kernel-size interval, and the interval of input-graph sizes corresponding to each kernel-size bin is also indicated.}
    \label{fig:testbed_emulator}
\end{figure*}

\textbf{Scaling of TTS.} We compare the exponential scaling rates of the two solvers---standalone QAOA and qReduMIS---as a function of the kernel size~$K$ in Fig.~\ref{fig:tts_qredumis_vs_qaoa} in main text. We modeled the bootstrap-median of the optTTS for each problem instance as an exponential function of~$K$, the size of its (first) kernel graph. 

\begin{equation}
    \text{optTTS} \;\sim\; 2^{\,\beta_m\, K},
\end{equation}
or equivalently, in log-space,
\begin{equation}
    \log_2 \text{optTTS} \;=\; \alpha_m + \beta_m\, K + \varepsilon,
\end{equation}
where $\beta_m$ is the exponential scaling exponent, $\alpha_m$ is an
intercept capturing baseline overhead, and $\varepsilon$ represents
zero-mean residual noise. The parameters $\hat{\alpha}_m$ and
$\hat{\beta}_m$ were estimated via ordinary least-squares (OLS)
regression on all instances with finite, strictly positive
bootstrap-median TTS values. Note that for QAOA with $p=2$, some individual seeds yield infinite optTTS values, but the bootstrap-median is infinite only at $K=22$ (refer to Fig.~\ref{fig:heatmap_inf_optTTS}).

\begin{table}[t]
\centering
\caption{Ordinary least-squares fits of 
$\log_2 \widetilde{T} = \alpha + \beta\,K + \varepsilon$ for each
method and QAOA depth~$p$.}
\label{tab:individual-fits}
\begin{tabular}{llcccc}
\toprule
Method   & $p$ & $\hat{\beta} \pm \mathrm{SE}$ & $R^{2}$ 
         & $p$-value ($\hat{\beta}$) \\
\midrule
QAOA     & 2 & $0.5235 \pm 0.0328$ & 0.80 & ${<}\,10^{-10}$  \\
qReduMIS & 2 & $0.1642 \pm 0.0179 $ & 0.54 & ${<}\,10^{-10}$  \\
\midrule
QAOA     & 6 & $0.2894 \pm 0.0210$ & 0.73  & ${<}\,10^{-10}$ \\
qReduMIS & 6 & \multicolumn{4}{c}{\textit{No significant exponential fit}} \\
\bottomrule
\end{tabular}
\end{table}

Table~\ref{tab:individual-fits} summarizes the OLS regression results
for each method and QAOA depth combination. For $p=2$, both QAOA and qReduMIS exhibited statistically significant exponential scaling in the kernel size~$K$. However, the scaling exponents differed markedly. The two-sample $t$-test on the slope difference yielded
\begin{multline}
    \Delta\hat{\beta}
    = \hat{\beta}_{\text{QAOA}} - \hat{\beta}_{\text{qReduMIS}}
    = 0.359 \pm 0.037, \\
    t = 9.61,
    \qquad
    p < 10^{-10},
\end{multline}
decisively rejecting the null hypothesis of equal scaling exponents at any conventional significance level.

The qReduMIS exponent is $3.2{\times}$ smaller than that of
standalone QAOA. In practical terms, for every unit increase in kernel
size~$K$, the optTTS of QAOA grows by a factor of
$2^{0.524} \approx 1.44$, while that of qReduMIS grows by only
$2^{0.164} \approx 1.12$. This difference compounds multiplicatively:
at $K = 20$, the model predicts an optTTS ratio of
$2^{(0.524 - 0.164) \times 20} = 2^{7.20} \;\approx\; 147.$ The better scaling for qReduMIS is very much expected, given that a substantial part of the overall workload is done by classical reduction with \textit{poly} runtime that helps reduce the overall (exponential) runtime.

Although the $R^2$ values are modest, reflecting the high instance-to-instance variability inherent to NP-hard combinatorial problems, the statistically significant slope difference confirms that the scaling advantage of qReduMIS over standalone QAOA is robust and not an artifact of noise. Moreover, for a single value of kernel size, we can consider different graph instances with different densities (as shown in the lower panel in Fig.~\ref{fig:testbed_emulator}) and the density is another metric driving hardness. 

At circuit depth $p = 6$, QAOA still exhibited significant exponential
scaling, whereas qReduMIS showed \emph{no statistically significant exponential
growth} over the tested range $K \in [4, 22]$. This represents the strongest qualitative evidence
for the advantage of the qReduMIS approach: at sufficient circuit depth, the combination of graph reduction and QAOA eliminates detectable exponential scaling entirely within the experimentally accessible regime.

\begin{figure}
    \centering
    \includegraphics[width=\linewidth]{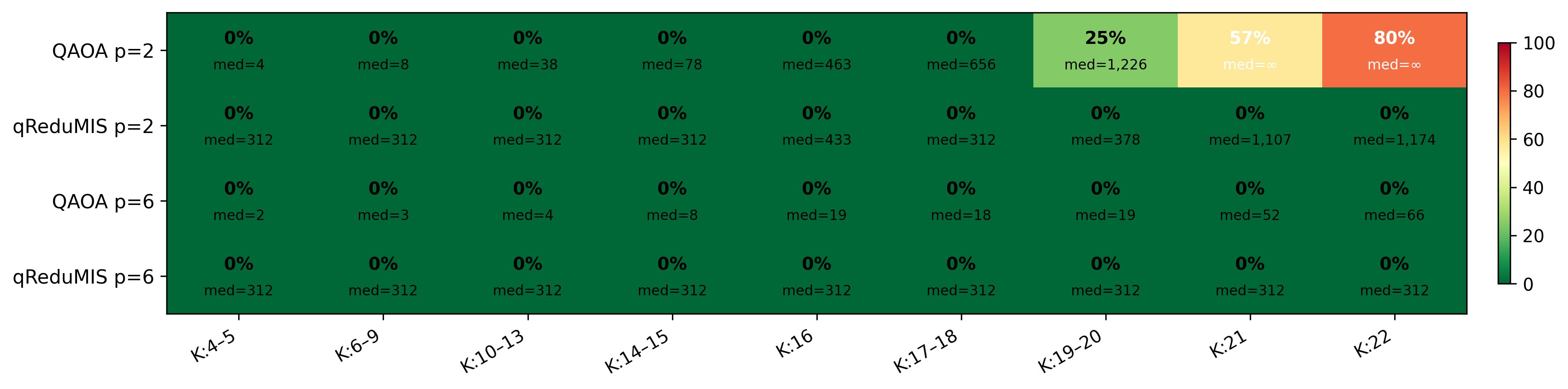}
    \caption{Heatmap showing the fraction of instances ($\%$) with $optTTS = \infty$ per interval of (first) kernel sizes for the different methods considered. qReduMIS refers to the qReduMIS powered by QAOA with the number of layers $p$ indicated. The median optTTS per interval is reported. Noticeably for QAOA with $p=2$ and kernels with 22 nodes, the median optTTS is $\infty$, as the fraction of trials whose success probability is zero is $80\%$.} 
    \label{fig:heatmap_inf_optTTS}
\end{figure}

\begin{figure*}[!t]
    \centering
    \includegraphics[width=\linewidth]{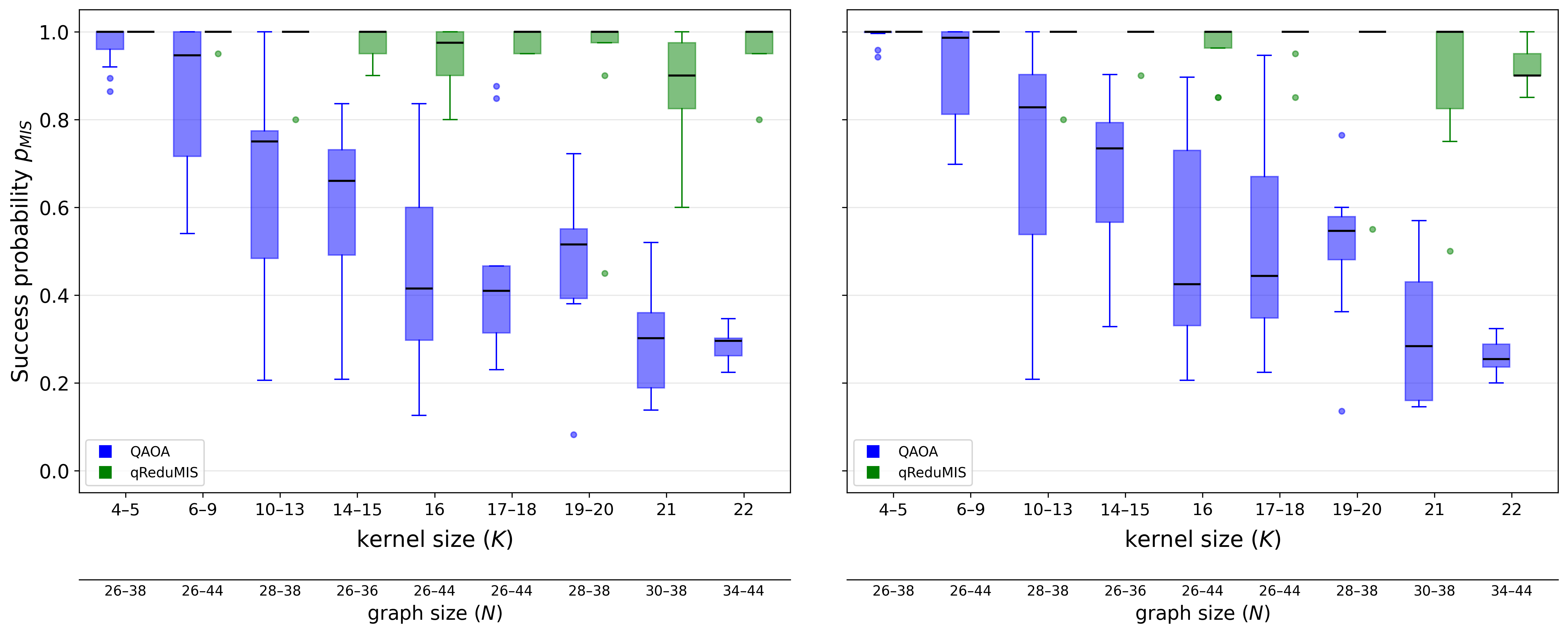}
    \caption{Success probability $P_{\text{MIS}}$ as a function of the first 
    kernel size $|\mathcal{K}|$ for standalone QAOA (blue) and qReduMIS (green) 
    on PSP-MIS instances from the Nikkei 225 index. \textbf{Left:} $p=2$ 
    layers; \textbf{Right:} $p=6$ layers. Box plots show the distribution 
    across problem instances within each kernel-size interval; the secondary 
    $x$-axis reports the corresponding input graph size $N$. Standalone QAOA 
    exhibits clear suppression of $P_{\text{MIS}}$ with increasing kernel size, 
    dropping to 0.078 at $|\mathcal{K}|=22$ for both $p=2$ and $p=6$. In 
    contrast, qReduMIS with $p=6$ achieves a much higher $P_{\text{MIS}}$ across all 
    problem sizes. QAOA uses 500 quantum shots; qReduMIS uses 20 classical 
    trials with 5 quantum shots per QPU call.}
    \label{fig:probability_mis_mkt_graphs}
\end{figure*}

\textbf{Performance with success probability.} For completeness, we
include the comparison of the end-to-end success probability
$P_{\text{MIS}}$ of standalone QAOA and qReduMIS powered by QAOA. The
success probability measures the fraction of independent trials that
return the optimal solution, without accounting for the computational
cost per trial. Since qReduMIS expends a larger budget of quantum shots
per trial than standalone QAOA, the optTTS metric reported in the main
text provides a more balanced comparison by incorporating this cost.
Nevertheless, the success probability remains an informative figure of
merit in its own right: it directly quantifies algorithmic reliability,
and its practical relevance may increase as hardware evolves to support
parallel circuit execution---in which case the wall-clock cost per trial
would be governed by the number of sequential QPU calls (which remains
consistently low for qReduMIS) rather than the total shot count.
Additionally, the quantum-shot schedule is a tunable hyperparameter, and
future adaptive strategies may further reduce the per-trial budget,
making the success probability an increasingly representative indicator
of end-to-end performance.

The success probability is shown in Fig.~\ref{fig:probability_mis_mkt_graphs} as a function of intervals of (first) kernel size $|\mathcal{K}_{1}|$ (hereafter denoted simply $\mathcal{K}$). We aggregate results into intervals of $|\mathcal{K}|$, ensuring a comparable number of samples (between six and ten seeds) in each interval. The exact number of seeds per interval is provided in Table~\ref{tab:num_seeds_mkt_emulator}. The x-axis also indicates the interval of input-graph sizes $N$ corresponding to each interval of (first) kernel size. Note that the box plots already capture the instance-to-instance variability across the 73 graphs within each kernel-size bin, making per-instance bootstrap confidence intervals redundant for this aggregate view. As expected, increasing the number of QAOA layers improves the performance of both standalone QAOA and qReduMIS. However, standalone QAOA exhibits a clear suppression in success probability as the kernel size grows, dropping to $0.078$ at a kernel size of $22$ for both $p=2$ and $p=6$. qReduMIS shows a milder decline for $p=2$, while for $p=6$, it achieves a success probability of exactly $1$ across all problem sizes, causing the box plots to collapse into flat bars. The dispersion and outliers within each box arise because each box aggregates instances spanning a range of kernel sizes and densities, both of which drive the hardness of the MIS problem. The distributions of these quantities are shown in Fig.~\ref{fig:testbed_emulator}, where notably, the kernel density exhibits considerable variation within the intervals.

\textbf{QPU calls.} To estimate the optTTS we are interested in the actual number of QPU calls to reach the optimal solution. We refer to it as \textit{QPU calls to solution} or $D^*$ in the main text. In most cases, the optimal solution was found in fewer iterations 
than the total number of QPU calls performed. This is because the intrinsic stopping criterion of qReduMIS is running out of nodes as nodes are removed in each iteration and it may be that the optimal solution is either measured in the first QPU call or constructed recursively through the previous iterations. The distributions in Fig.~\ref{fig:qpu_to_solution_distribution_emulator} correspond to the average QPU calls to solution performed per problem instance corresponding to the (first) kernel size. As expected, when increasing the number of layers in QAOA from two to six, the QPU calls are reduced and concentrated around one. The difference becomes significant for the larger problem sizes considered. Although for the larger sizes the QAOA with $p=2$ does not sample the optimal solution within the $\mathrm{qshots}=5$ measurements, as the success probability of standalone QAOA is very low, and zero in some cases, incurring another QPU call over the updated kernel suffices to find the optimal solution. 
For both cases, given the problem instances considered, at most two QPU calls suffice to find the optimal solution with high probability, and this is even in the cases that standalone QAOA over the kernel exhibits low success probability. 

\begin{figure*}
    \centering
    \includegraphics[width=\linewidth]{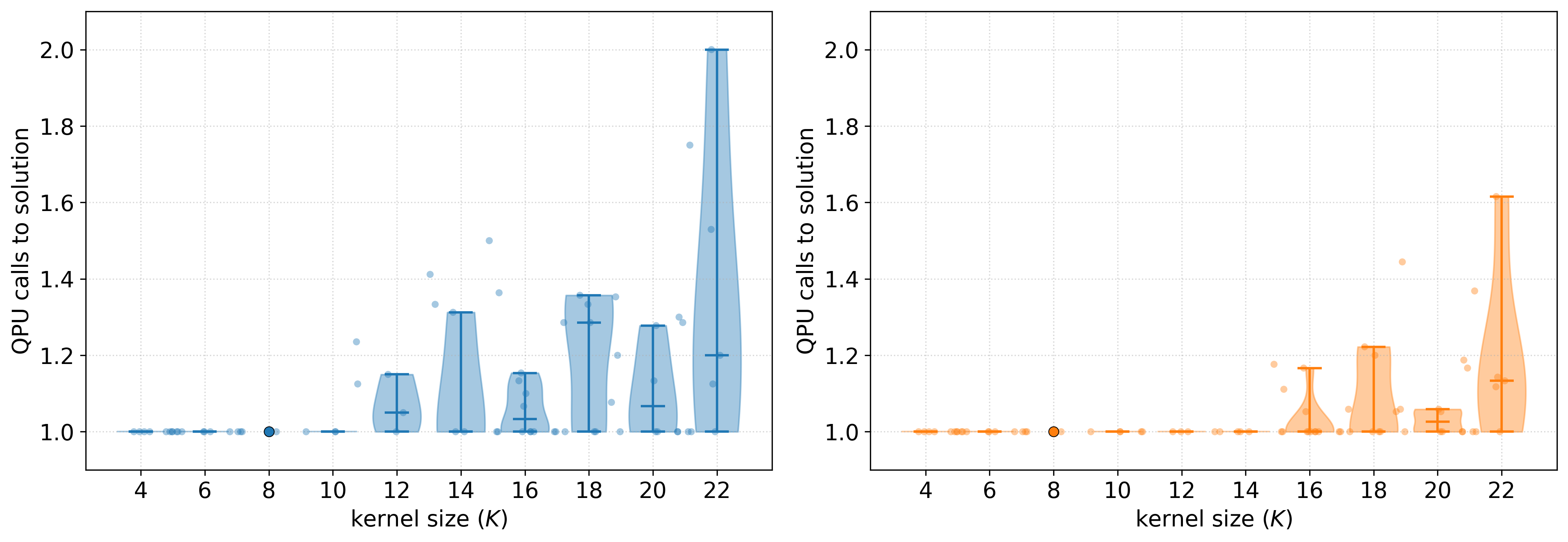}
    \caption{Distribution of QPU calls to solution ($D^*$) performed by qReduMIS as a function of the kernel size $K$, for $p=2$ (left) and $p=6$ (right). The distribution is over the average QPU calls to solution per problem instance. qReduMIS is executed for $\text{cshot}=20$ trials, and the average number of QPU calls to solution is reported and shown for each kernel size.}
    \label{fig:qpu_to_solution_distribution_emulator}
\end{figure*}

\textbf{QPU plays a crucial role in qReduMIS.} We investigate how the number of quantum shots ($\mathrm{qshots}$) used in each QPU call of qReduMIS impacts performance. For this, we consider the performance on one Nikkei 225 market-graph instance of size $30$ for different values of $\mathrm{qshots}$ using the H2-1 emulator. The figures of merit are displayed in Fig.~\ref{fig:quantum_shots_qredumis}. Note that the success probability $P_{\text{MIS}}$ is the end-to-end algorithm success probability (Eq.~\ref{eq:prob_qredumis}), given by whether the algorithm finds the optimal solution of the given graph or not. As expected, it increases with the number of $\mathrm{qshots}$ (left panel), as well as the average AR (right panel). Despite this improvement in $p_{\mathrm{MIS}}$, the optTTS also grows with $\mathrm{qshots}$ (centre panel). The empirical curve lies below the linear reference $\mathrm{optTTS}\propto\mathrm{qshots}$, indicating that its growth is sublinear.

Thus, the minimum optTTS is attained at the smallest value of $\mathrm{qshots}$ for which the success probability remains non-negligible. For reference, the success probability of standalone QAOA over the kernel graph is $0.4\%$. However, with qReduMIS, by taking only five quantum shots (in each QPU call) the end-to-end success probability is boosted to $0.35$. Note that the probability of sampling the optimal solution from the independent set measurements of QAOA with only five samples is $\sim 2\%$. Nevertheless, as already discussed, qReduMIS does not require to sample the optimal solution in the QPU call to output an optimal solution, as we see in this example. While in practice qReduMIS performed more than one QPU call, it was observed that the optimal solution was found with only one QPU call followed by a classical reduction. 

\begin{figure*}
    \centering
    \includegraphics[width=1.0\linewidth]{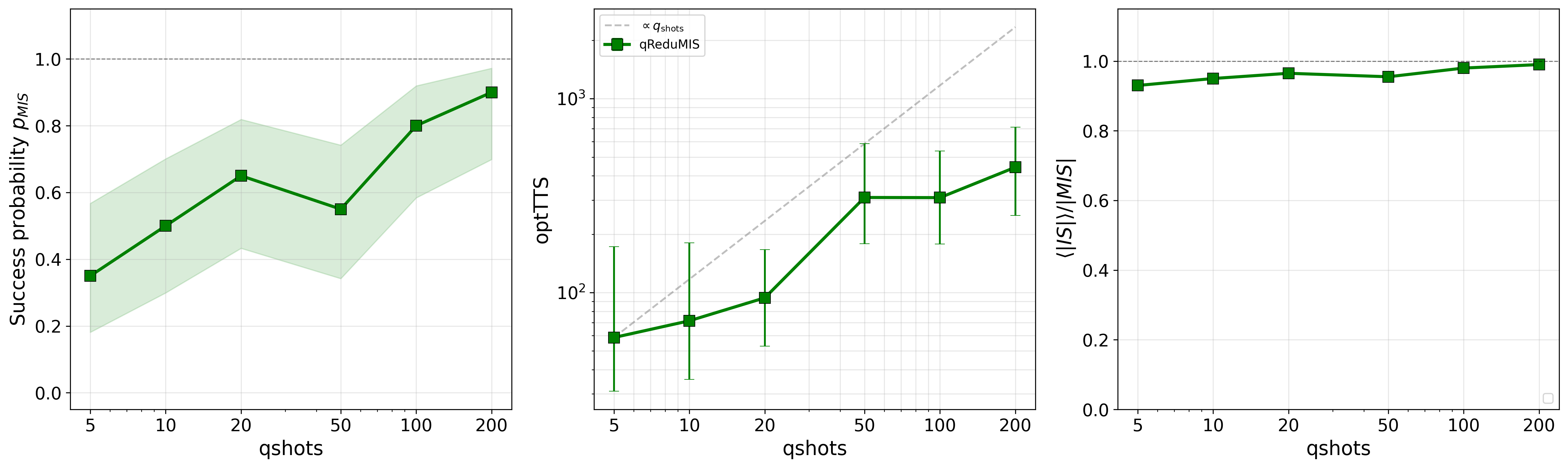}
    \caption{%
  Effect of the number of quantum shots ($\mathrm{qshots}$) on qReduMIS performance
  for a Nikkei 225 market-graph instance ($N=30$, QAOA depth $p=2$,
  hardness $\mathbb{H}=3.375$).
  \textbf{Left:}~Success probability $P_{\mathrm{MIS}}$ (fraction of classical shots
  finding the ground-truth MIS), with 95\% Wilson confidence intervals (shaded). The horizontal dashed line at $1$ indicates the best possible value. 
  \textbf{Centre:}~Optimal time-to-solution
  $\mathrm{optTTS}$; the dashed line shows the linear reference
  $\mathrm{optTTS}\propto \mathrm{qshots}$. The experimental curve is below this reference line indicating that the increase in optTTS is slower than linear. 
  \textbf{Right:}~Average approximation ratio
  $\langle |\text{IS}| \rangle / |\text{MIS}|$, showing that solution quality is significantly high across the different values of $\mathrm{qshots}$. The horizontal dashed line at $1$ indicates the best possible value. 
}
    \label{fig:quantum_shots_qredumis}
\end{figure*}

\begin{figure}
    \centering
    \includegraphics[width=\linewidth]{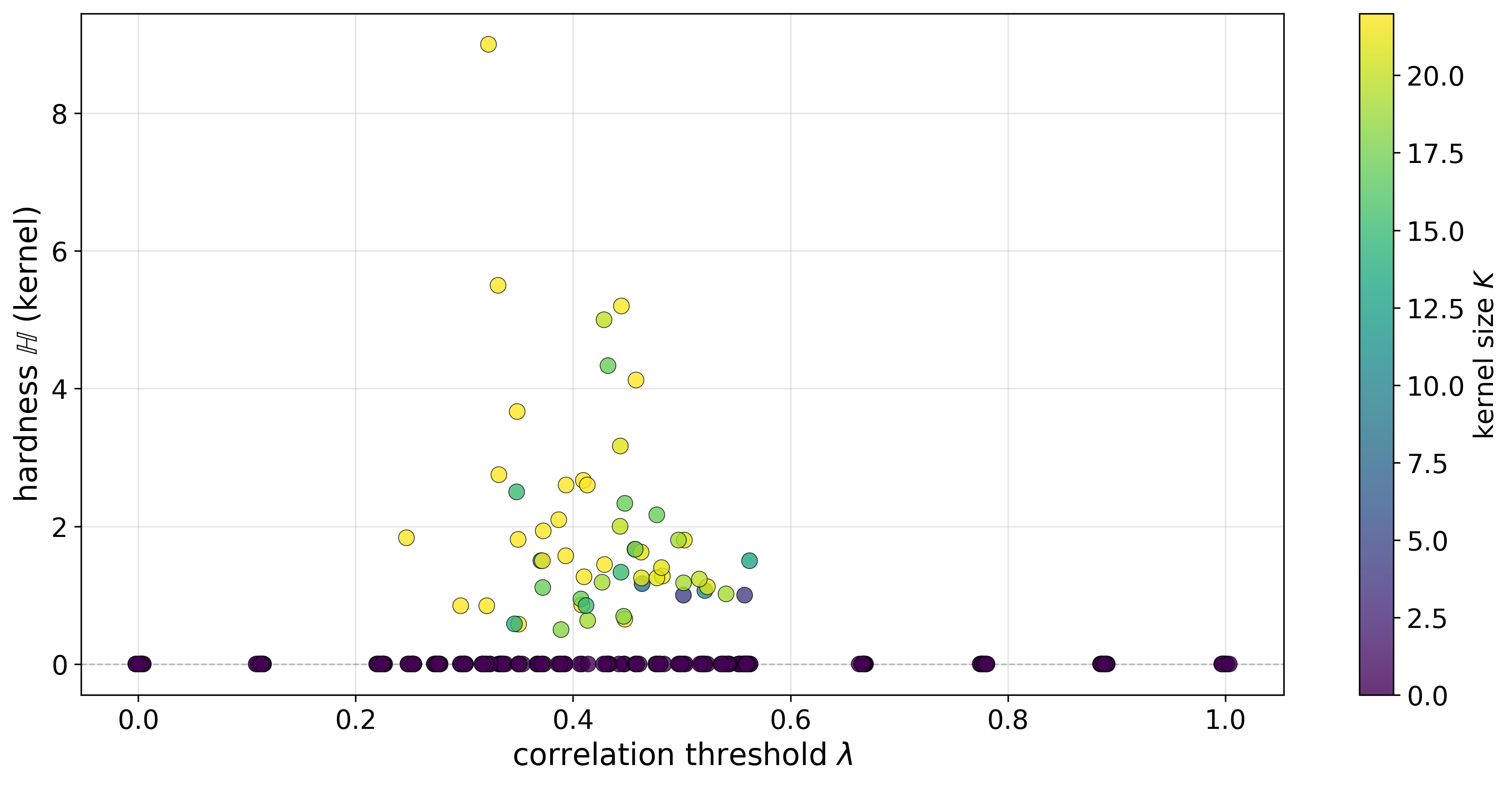}
    \caption{Hardness of the (first) kernel graph as a function of the correlation threshold $\lambda$ utilized to construct the given market graph. The heatmap refers to the size of the (first) kernel obtained after the first reduction over the given market graph.}
    \label{fig:hardness_vs_threshold}
\end{figure}

\textbf{Performance with respect to the conductance-like hardness.} It has been shown for the Quantum Annealing Algorithm (QAA) \cite{ebadi:22, schuetz2025qredumis} and Simulated Annealing (SA) \cite{ebadi:22, schuetz2025qredumis, andrist:23} that algorithmic performance is exponentially suppressed in a conductance-like hardness parameter $\HP$, defined as $\HP = D_{|\mathrm{MIS}|-1} / (|\mathrm{MIS}| \cdot D_{|\mathrm{MIS}|})$, where $D_{\alpha}$ denotes the degeneracy of independent sets of size $\alpha$. In our framework for constructing market graphs, for a given problem size (number of assets), this hardness parameter is directly governed by the correlation threshold $\lambda$ used to convert the weighted graph into an unweighted one, which in turn determines the graph density. Importantly, the density of the kernel graphs--the graphs actually input to the QPU--depends on the classical reduction, whose effectiveness is not, in principle, governed by the notion of hardness.

We evaluate the performance of qReduMIS powered by QAOA as a function of the hardness parameter $\HP$ and benchmark it against standalone QAOA. To construct graphs of varying hardness, we fix the number of nodes (assets) and adjust the threshold (or correlation sensitivity) $\lambda$ to generate different graph densities. Note that the instances considered in Fig.~\ref{fig:tts_qredumis_vs_qaoa} exhibit similar hardness values. This similarity arises because we set $\lambda = \frac{1}{N^2}\sum_{i,j} |C_{ij}|$, where $C$ is the correlation matrix for the set of $N$ assets randomly selected from the Nikkei 225 index. Although different assets are considered in each instance, the ratio of $\lambda$ to the correlations within each asset universe remains nearly constant. Furthermore, the range of problem sizes considered is relatively narrow. As a result, the hardness across instances does not change significantly: for the majority of input graphs, $\HP$ remains below two, with only seven instances exceeding this value. The same behavior is observed for the kernel graphs.

We fix the number of assets to 22, as the largest size we solve with the Quantinuum's H2-1 emulator, and we randomly select from the Nikkei 225 index to construct 10 instances of 22 assets. We sampled
$25$ correlation sensitivity $\lambda$ values in total: ten uniformly spaced values
$\lambda \in [0.001, 1.0]$ spanning the full range, plus fifteen additional values in the range $\lambda \in [0.25, 0.56]$, concentrated in the intermediate regime where the classical reduction yields a non-trivial kernel ($K>0$). Thus, we consider $250$ instances in total. Of these, $57$ have a non-empty kernel after
classical reduction and a well-defined hardness parameter $\mathbb{H}>0$; these
are the instances reported in Fig.~\ref{fig:metrics_hardness}. We compare the performance of qReduMIS powered by QAOA with $p=6$ against running standalone QAOA with same number of layers, both acting over the (first) kernel graph. In qReduMIS we set up the hyper-parameter of quantum shots in each QPU call to $\mathrm{qshots}=10$ and for the experiments with QAOA we utilized 500 quantum shots to estimate the figures of merit. 

Fig.~\ref{fig:metrics_hardness} presents the three performance metrics: success probability $P_{\mathrm{MIS}}$, average approximation ratio $\langle|IS|\rangle / |\text{MIS}|$, and optTTS, as a function of the hardness parameter $\HP$. The figure reveals a clear degradation in performance with
increasing hardness for both solvers, yet qReduMIS consistently outperforms
standalone QAOA across success probability and average approximation ratio. In the low-hardness regime
($\mathbb{H}\lesssim 1$), both methods achieve high success probability
and average approximation ratio, reflecting the low hardness of these instances. However, while standalone QAOA achieves average approximation ratio values below $1$ but above $0.8$ in the low-hardness regime, qReduMIS attains exactly unit value for these instances. As $\HP$ grows beyond $1$, success probabilities and average approximation ratios decline for both approaches, but qReduMIS degrades slowly with hardness, maintaining both metrics appreciably above the corresponding standalone QAOA values. Regarding optTTS, standalone QAOA's optTTS rises steeply with hardness, whereas qReduMIS exhibits only a modest deviation from its otherwise near-constant optTTS. As discussed previously, qReduMIS with $p=6$ achieves a steady optTTS for the majority of instances, as they are solved with a single QPU call. It is only for the harder instances that additional QPU calls become necessary, leading to an increase in optTTS. Nevertheless, within the regime of hardness and problem sizes considered here, the number of such harder instances remains insufficient to extract a robust scaling exponent.

Note that there is a small region of values of hardness parameter for which we do not have results. This is because we could not generate instances whose (first) kernels are characterized by those values of hardness. As discussed, for a fixed problem size, the spectrum of hardness parameter values is limited, and we have not seen a clear relation between the correlation threshold $\lambda$ and hardness nor with the graph density. However, an easy--hard--easy transition of hardness is observed as a function of the density, which is also represented as a function of the threshold in Fig.~\ref{fig:hardness_vs_threshold}. We observe fully reducible instances for both low and large correlation thresholds, making the graphs either fully dense or very sparse. For values of threshold in $[0.2, 0.6]$, we observe non-fully reducible instances corresponding to different sizes of kernels but no clear monotonic relation.

\begin{figure*}
    \centering
    \includegraphics[width=\linewidth]{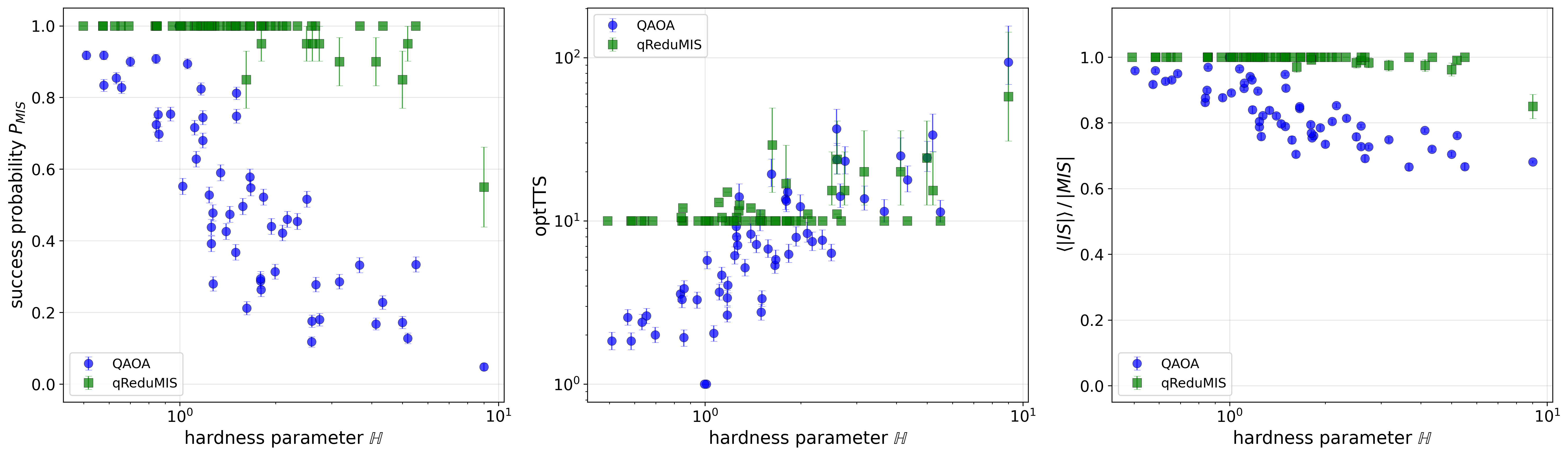}
    \caption{Performance metrics as a function of the hardness parameter
$\mathbb{H}$ for $N=22$ Nikkei 225
market graphs with non-trivial kernel ($K>0$). Each point represents one
problem instance ($57$ per method); error bars indicate $95\%$ bootstrap confidence intervals obtained from
$1000$ binomial resamples over the $500$ QAOA shots and $20$ qReduMIS classical shots,
respectively.
\textbf{(Left)}~Success probability $P_{\mathrm{MIS}}$.
\textbf{(Middle)}~Optimal time-to-solution
$\mathrm{optTTS}$
(log scale).
\textbf{(Right)}~Average approximation ratio
$\langle|IS|\rangle/|\text{MIS}|$, measuring the expected independent-set size
normalised by the MIS size.
Instances are generated from $25$ correlation thresholds ($10$ seeds each);
higher $\mathbb{H}$ indicates a harder landscape.}
    \label{fig:metrics_hardness}
\end{figure*}

\textbf{Benchmark against a classical solver.}
To contextualize the performance of qReduMIS, we benchmark against
simulated annealing (SA), a well-established classical metaheuristic
for combinatorial optimization. We use an optimized SA implementation
following Ref.~\cite{andrist:23}. For a fair comparison with QAOA and
qReduMIS powered by QAOA at circuit depth~$p$, the number of SA steps
is set to match~$p$; 500 replicas are performed with a linear
temperature schedule, and the same bootstrap methodology used for
qReduMIS and QAOA is applied.

Fig.~\ref{fig:fig_with_sa} reports the SA results on the same Nikkei
market-graph instances discussed in the main text
(Fig.~\ref{fig:tts_qredumis_vs_qaoa}). The $y$-axis shows, on a
logarithmic scale, the bootstrap median of the optimal
time-to-solution (optTTS), estimated from 500 independent SA runs per
instance. The majority of data points cluster at an optTTS floor of
$\approx 2/3$, which is the minimum value the estimator can report
when SA succeeds on all 500 runs. It corresponds to
evaluating Eq.~\ref{eq:floor_opttts} at $M=500$, which yields
$2\ln 10 / \ln(10^{3}) = 2/3$. A minority of instances exhibit a
success probability below~1, indicating that SA fails on some
replicas. Although these tend to be the larger instances, some
smaller instances display the same behaviour, because problem size
alone does not determine hardness: graph density is also a driver of difficulty for SA and, more broadly, for
Markov-chain Monte Carlo solvers. As expected, increasing the number
of SA steps improves performance and causes the optTTS distribution to
concentrate further around the $2/3$ floor.

Despite the few instances that depart from this floor, the fraction
of non-trivial optTTS values is too small to support a meaningful
scaling analysis or a fair head-to-head comparison with qReduMIS at
the current instance sizes. We expect that benchmarking on larger
instances will yield enough variation to enable such an analysis. We
note that the qReduMIS results presented here are limited to kernel
sizes of up to 22 nodes, owing to the computational cost of
classically simulating quantum circuits at larger scales.

\begin{figure*}
    \centering
    \includegraphics[width=\linewidth]{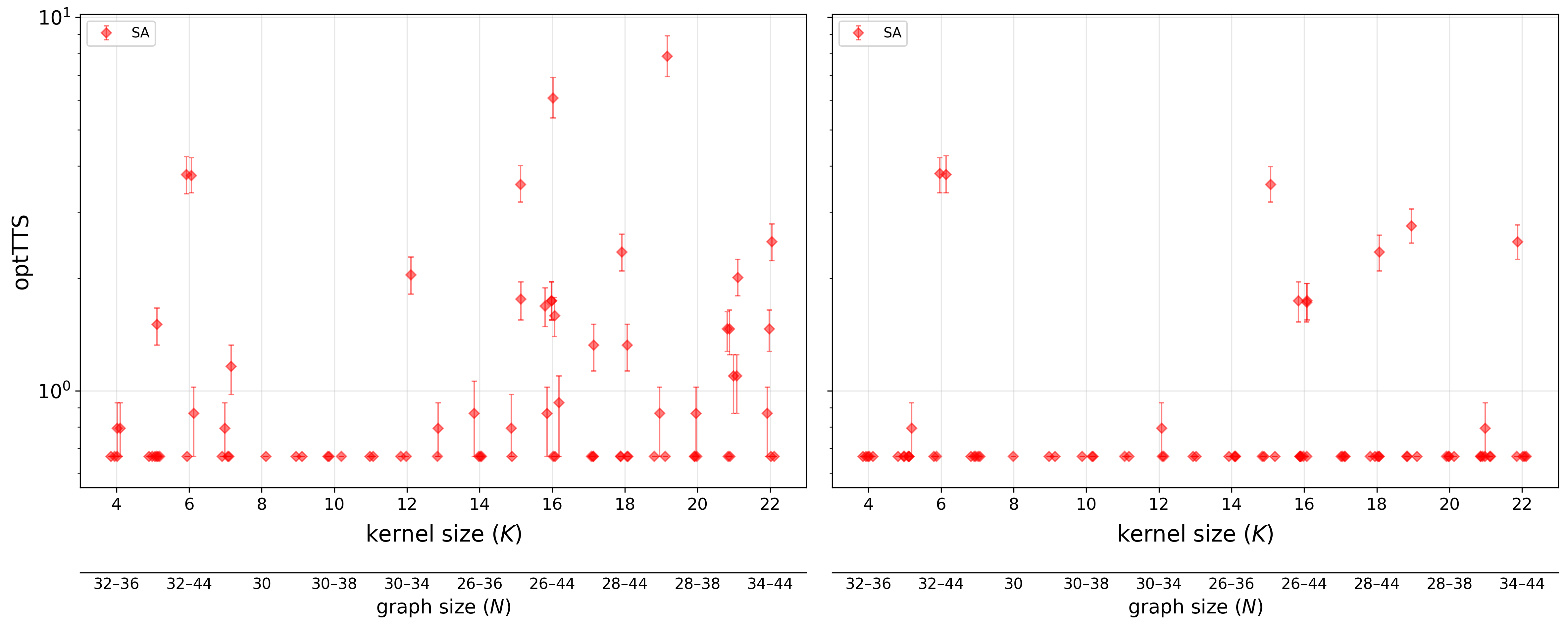}
    \caption{Optimal time-to-solution (optTTS) as a function of the first 
    kernel size $|\mathcal{K}|$ for simulated annealing on PSP-MIS instances from the Nikkei 225 index. 
    \textbf{Left:} $\text{steps}=2$; \textbf{Right:} $\text{steps}=6$. Error bars denote 95\% 
    bootstrap confidence intervals on the median optTTS per problem instance. 
    SA achieves an optTTS concentrated around $1$, meaning that after running for only one replica, the optimal solution is found. There are some instances, as the (first) kernels become larger whose optTTS is significantly increased. However, this does not suffice to perform a scaling analysis of the optTTS with kernel size.}
    \label{fig:fig_with_sa}
\end{figure*}

\subsection{Quantum-informed MIS of 3-regular graphs \label{appendix:3regular}}

We also benchmark qReduMIS (powered by QAOA) on \textit{unstructured} instances. Specifically, we consider random 3-regular graphs, which have been extensively studied in the QAOA literature \cite{farhi2025lower, wurtz:21, harrigan2021quantum}. For MIS in particular, there are recent theoretical results characterizing QAOA performance on large-girth 3-regular graphs \cite{farhi2025lower}.

The reduction strategy used in qReduMIS is (clique-based) isolated-vertex removal. While this reduction is effective on \textit{structured (real-world)} instances, synthetic random $d$-regular graphs are typically \textit{unstructured} and therefore largely immune to this specific reduction. However, within the qReduMIS framework the algorithm proceeds recursively and in the quantum component when using the in-set criterion, frozen nodes are identified from QPU measurements and removed (together with their neighbors) and the resulting kernel graphs do not necessarily remain $d$-regular. This can create opportunities for the reduction rules to take effect in later recursion levels. For this reason, random 3-regular graphs provide a useful testbed to evaluate the generality and robustness of the proposed qReduMIS approach.

\begin{figure*}[!t]
    \centering
    \begin{minipage}[t]{0.49\textwidth}
        \centering
        \includegraphics[width=\linewidth]{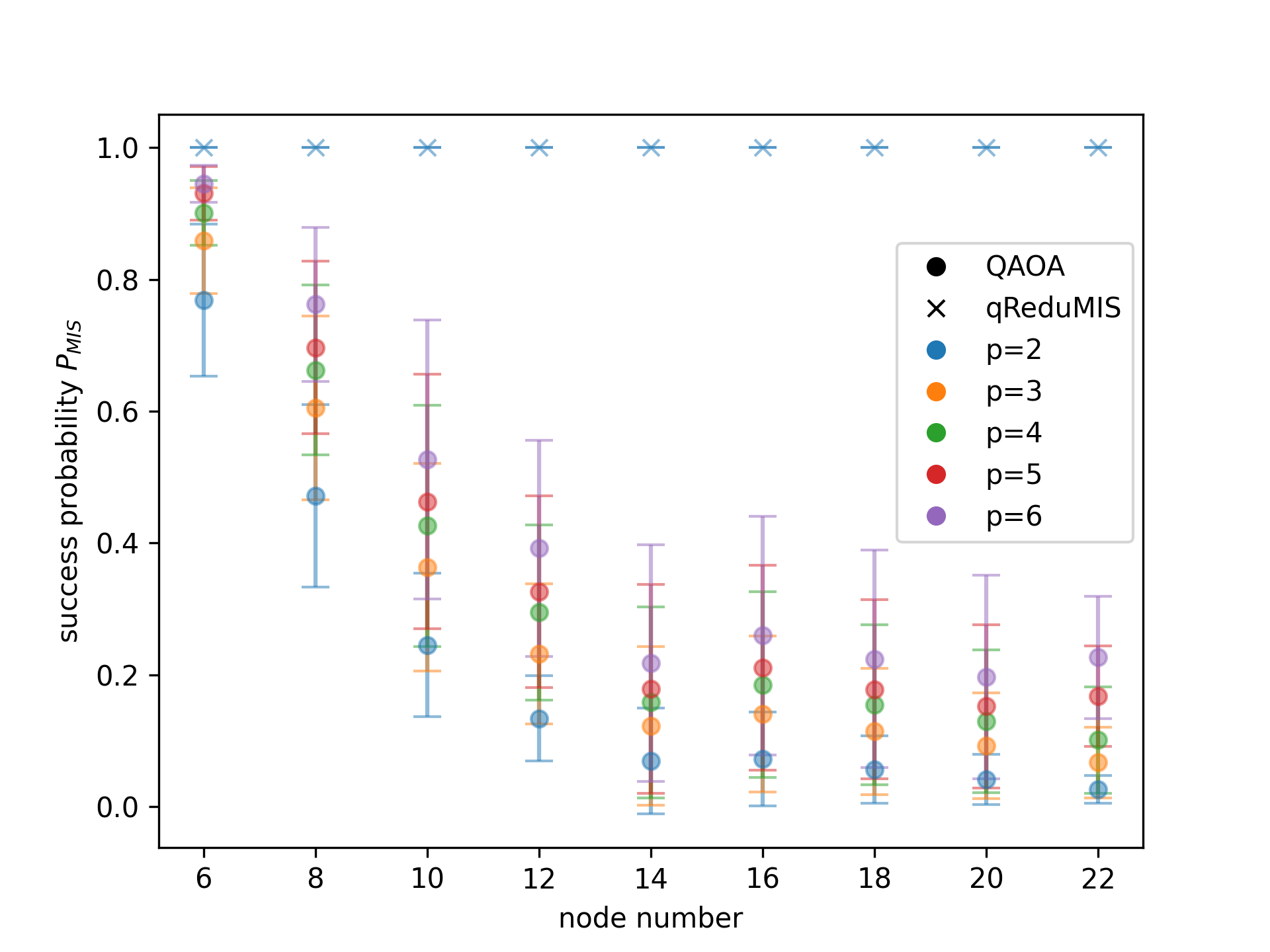}
    \end{minipage}\hfill
    \begin{minipage}[t]{0.49\textwidth}
        \centering
        \includegraphics[width=\linewidth]{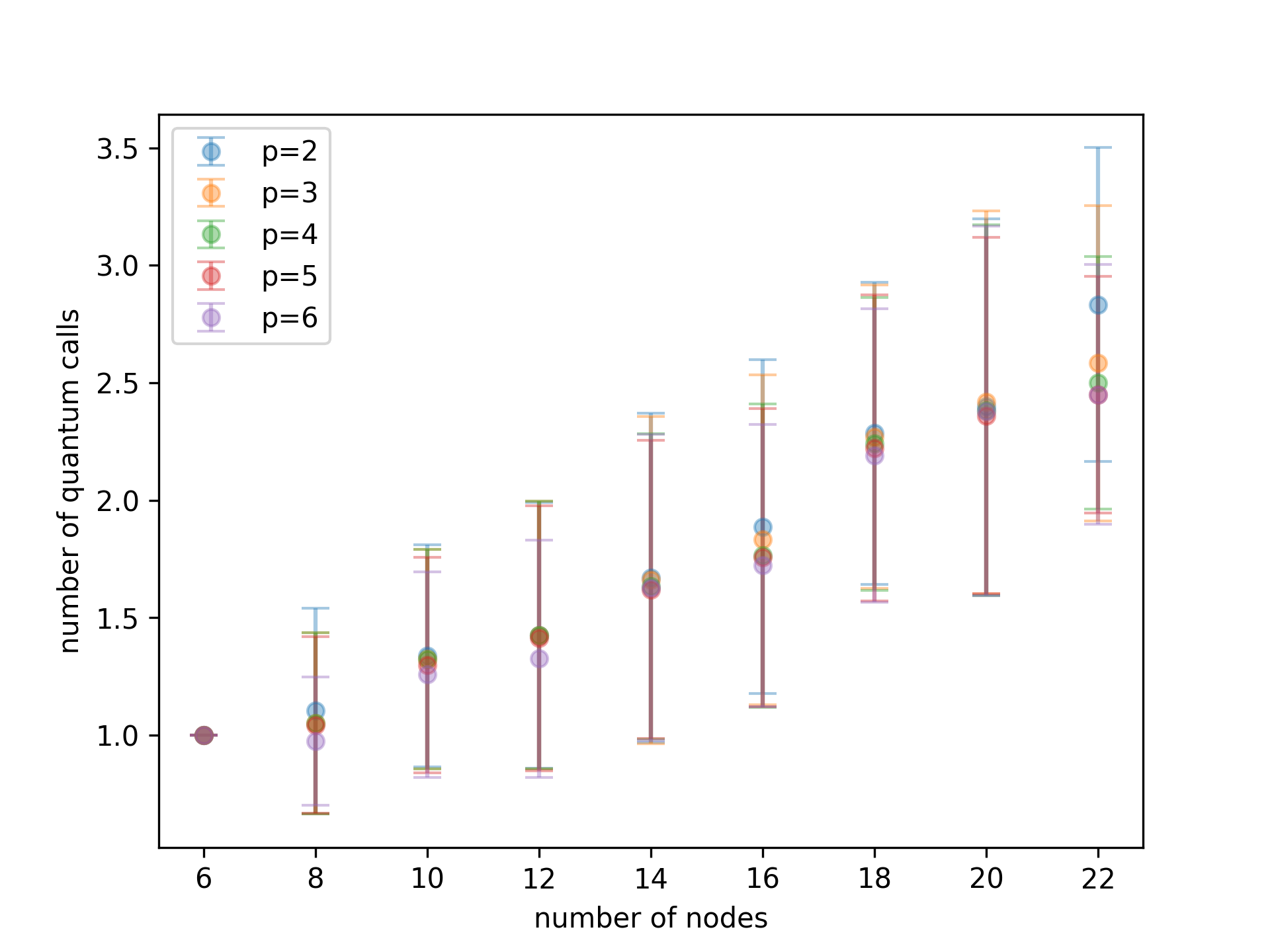}
    \end{minipage}
    \caption{Performance of QAOA and qReduMIS powered by QAOA tackling random 3-regular graphs of different sizes. QAOA, in both cases, is run with different numbers of layers $p$ indicated by color. \textbf{Left:} success probability as a function of the number of nodes; for qReduMIS, $P_{\text{MIS}}=1$ for every $p$ considered, causing the markers to overlap. \textbf{Right:} number of quantum calls of qReduMIS as a function of the number of nodes. Error bars correspond to the standard deviation; for all instances at least one quantum call is performed, so error bars for values below $1$ have no physical meaning.}
    \label{fig:3regular_benchmark_simulator}
\end{figure*}

For different problem sizes, we generate 20 random 3-regular instances (different seeds). We then benchmark standalone QAOA against qReduMIS powered by QAOA for several circuit depths (layers) $p$ using the noiseless statevector simulator. Fig.~\ref{fig:3regular_benchmark_simulator} (top) shows the success probability versus the number of nodes for different QAOA depths $p$ (colors), both for standalone QAOA and when QAOA is used inside qReduMIS. We observe that the success probability of standalone QAOA decays approximately exponentially with problem size. Increasing $p$ mitigates this decay and improves performance, but the overall exponential trend remains. In contrast, qReduMIS powered by QAOA achieves a success probability of $1$ across all tested instances and all considered values of $p$, as indicated by the single set of markers. Notably, even at small size ($N=6$), the best standalone QAOA success probability (across seeds) reaches $0.96$ at $p=6$, whereas qReduMIS attains perfect success already with a substantially smaller depth (e.g., $p=2$). Fig.~\ref{fig:3regular_benchmark_simulator} (bottom) reports the number of quantum calls made by qReduMIS. This quantity remains bounded and increases with problem size, consistent with the added recursion required to progressively identify and eliminate frozen nodes. Noticeably, for all the problem instances, the number of quantum calls is larger than or equal to one. This is because the first reduction does not do anything to the 3-regular graph (i.e., reduction factor is zero) but after the first quantum call, with the in-set criterion, one node is selected and is removed together with its neighbors. This allows the possibility of producing an exposed corner node, thus unblocking the next classical reduction. 

\begin{figure*}[!t]
    \centering
    \includegraphics[width=0.8\linewidth]{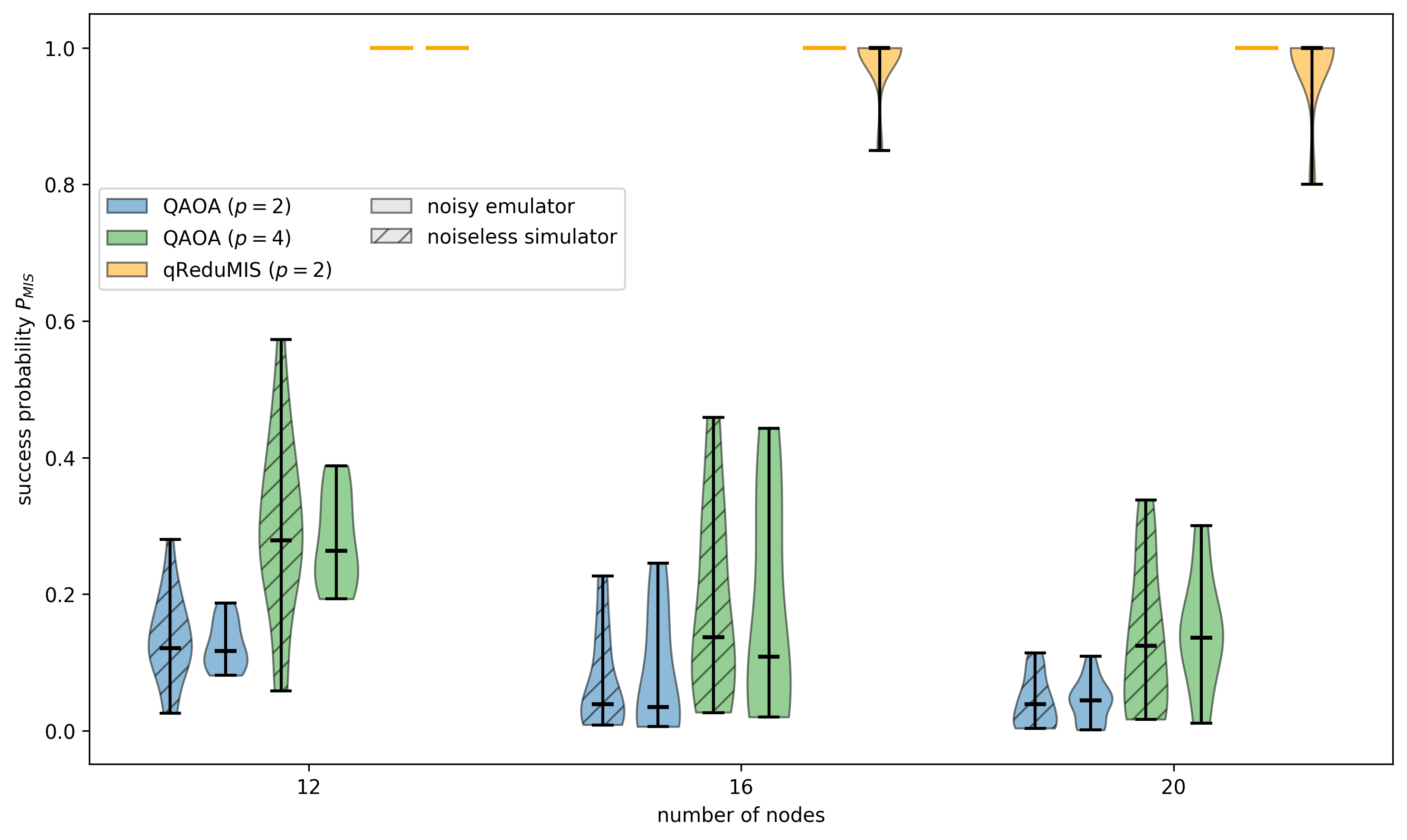}
    \caption{Performance of QAOA and qReduMIS powered by QAOA on random 3-regular graphs of different sizes. Violin plots show the distribution of $P_{\text{MIS}}$ for QAOA with layer number $p=2$ and $p=4$ (in blue and green) and for qReduMIS powered by QAOA with $p=2$ (orange), under both noiseless (solid fill) and noisy (hatched fill) conditions. Interestingly, qReduMIS achieves $P_{\text{MIS}}=1$ without dispersion in noiseless simulation. QAOA noisy results correspond to $1000$ quantum shots, while QAOA within qReduMIS corresponds to $250$ per iteration.}
    \label{fig:emulator_3regular}
\end{figure*}

\begin{figure}[h!]
    \centering
    \includegraphics[width=\linewidth]{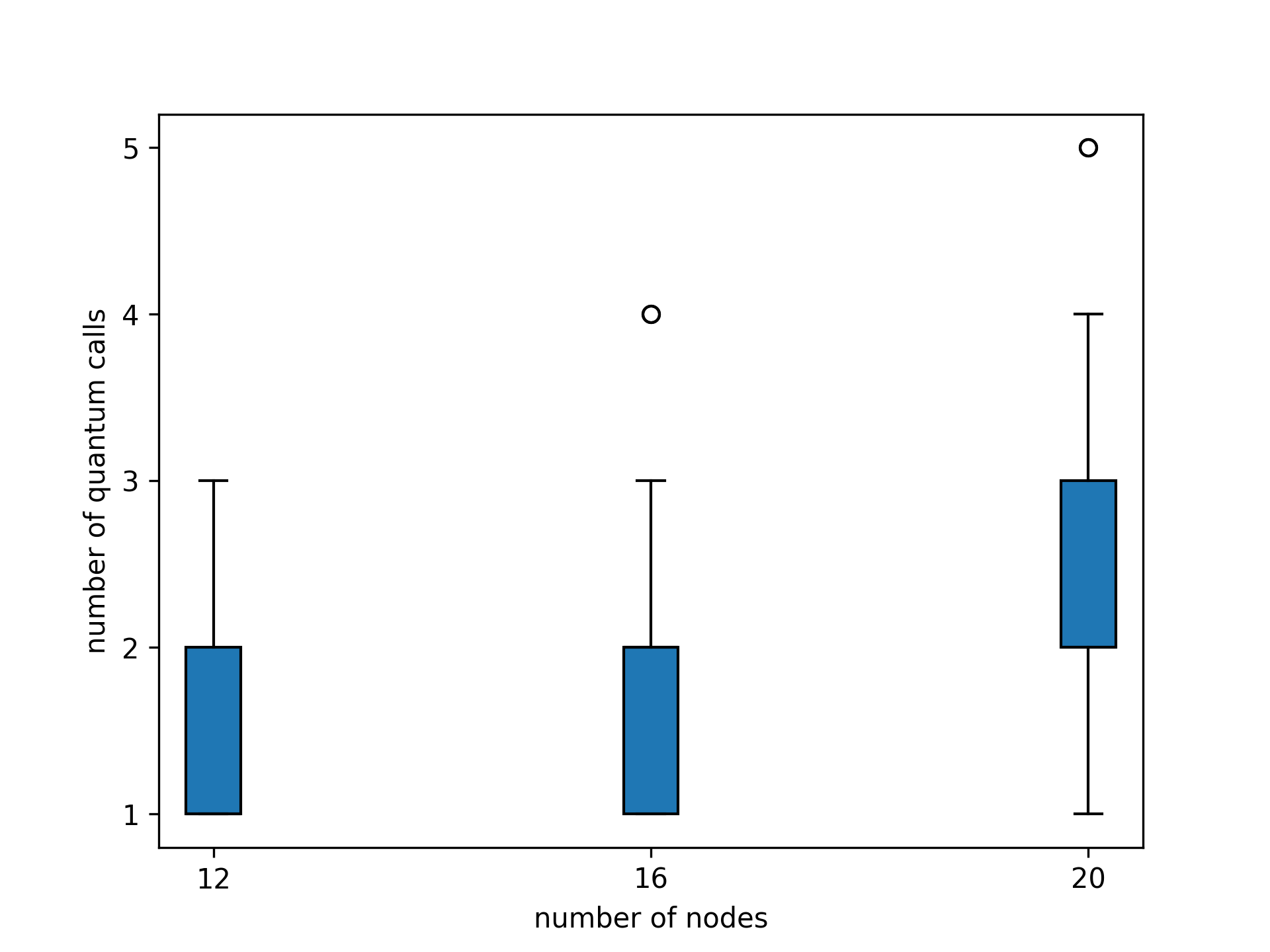}
    \caption{The number of quantum calls of qReduMIS powered by QAOA with $p=2$ on Quantinuum's noisy emulator as a function of the problem size $N$. Each QPU call performed 250 quantum shots. }
    \label{fig:calls_3regular}
\end{figure}

We perform these experiments on the (noisy) Quantinuum H2-1 emulator for selected problem sizes and report the success probabilities in Fig.~\ref{fig:emulator_3regular}, alongside noiseless simulation results for comparison. Standalone QAOA is run with 1000 shots for both $p=2$ and $p=4$. In contrast, qReduMIS is powered by QAOA at depth $p=2$, using 250 shots per iteration. When noise is included, qReduMIS performance can degrade, as reflected by reduced success probabilities for some instances at $N=16$ and $N=20$ (see the violin plots). Nevertheless, qReduMIS substantially outperforms standalone QAOA. Interestingly, the median success probability remains close to its noiseless value, suggesting that noise primarily affects the tails of the distribution rather than the typical-case performance. A similar behavior is observed for standalone QAOA: the noisy and noiseless distributions are comparable, and noise does not uniformly reduce performance. Consistent with previous results, standalone QAOA exhibits improved (median) performance as the number of layers increases. The distribution of the number of QPU calls made by qReduMIS is shown in Fig.~\ref{fig:calls_3regular}. We find that at least 99\% of the instances in our testbed require no more than 4 QPU calls, corresponding to a total of 1,000 shots (since each call uses 250 shots).\\

\subsection{Technical details about the numerical simulations \label{appendix:info_numerical_sim}}

The hardware used for simulations with the Quantinuum emulator was based on AMD CPUs. Depending on the experiment, we utilized up to 56 CPU cores and 60 GB of RAM to parallelize the runs of both qReduMIS and QAOA.

\textbf{Specifications of qReduMIS implementation used in the benchmarks.} We utilize the semi-greedy implementation of qReduMIS with the in-set criterion to select frozen nodes from the QPU output \cite{schuetz2025qredumis}. We pick $r=1$ candidate node from a restricted candidate list (RCL) of size $K_{\mathrm{RCL}}=2$ (we use $r$ here to avoid clashing with the correlation threshold $\lambda$). That is, we sample nodes with the highest probability of being in the set from measurement results $\{\mathcal{I}_{n}\}$ whose Hamming weight corresponds to the largest or second-largest IS size. We make a subset of the top $M=4$ nodes with highest frequency among the nodes corresponding to the measurements result in the RCL. We then randomly select one node from these $M$ and we add that node to the set of selected nodes $\mathcal{S}$. As we use the in-set criterion, the selected node and its neighbors compose $\mathcal{Q}$ and we proceed to remove them from the kernel graph and output a new kernel graph that is the input for the next qReduMIS iteration.

\end{document}